\documentclass[letterpaper]{article} 
\usepackage{aaai25}  
\usepackage{times}  
\usepackage{helvet}  
\usepackage{courier}  
\usepackage{xcolor}
\usepackage{multicol}
\usepackage{multirow}
\usepackage{algorithm}
\usepackage{algorithmic}
\usepackage{enumitem}
\usepackage{utfsym}
\usepackage{subfigure}
\usepackage{amsmath}
\usepackage{amssymb}
\usepackage{amsthm}
\usepackage{tabularx}
\usepackage{subcaption}

\newtheorem{prop}{Proposition}

\newtheorem{definition}{Definition}

\usepackage[hyphens]{url}  
\usepackage{graphicx} 
\usepackage{natbib}  
\usepackage{caption} 
\usepackage{newfloat}
\usepackage{listings}
\DeclareCaptionStyle{ruled}{labelfont=normalfont,labelsep=colon,strut=off} 
\floatstyle{ruled}
\newfloat{listing}{tb}{lst}{}
\floatname{listing}{Listing}
\title{Rethinking Byzantine Robustness in Federated Recommendation from Sparse Aggregation Perspective}
\author{
    Zhongjian Zhang\textsuperscript{\rm 1}\equalcontrib, Mengmei Zhang\textsuperscript{\rm 2}\equalcontrib,
    Xiao Wang\textsuperscript{\rm 3},
    Lingjuan Lyu\textsuperscript{\rm 4}\\
    Bo Yan\textsuperscript{\rm 1},
    Junping Du\textsuperscript{\rm 1},
    Chuan Shi\textsuperscript{\rm 1}\thanks{Corresponding author.}
}
\affiliations{
    \textsuperscript{\rm 1}Beijing University of Posts and Telecommunications,
    \textsuperscript{\rm 2}China Telecom Bestpay,
    \textsuperscript{\rm 3}Beihang University,
    \textsuperscript{\rm 4}Sony AI \\
    zhangzj@bupt.edu.cn, zhangmengmei@bestpay.com.cn, xiao\_wang@buaa.edu.cn,
    lingjuan.lv@sony.com\\
    \{boyan, junpingdu, shichuan\}@bupt.edu.cn\\
}

\usepackage{bibentry}

\begin{document}
\maketitle
\begin{abstract}
To preserve user privacy in recommender systems, federated recommendation (FR) based on federated learning (FL) emerges, keeping the personal data on the local client and updating a model collaboratively. Unlike FL, FR has a unique sparse aggregation mechanism, where the embedding of each item is updated by only partial clients, instead of full clients in a dense aggregation of general FL. 
Recently, as an essential principle of FL, model security has received increasing attention, especially for Byzantine attacks, where malicious clients can send arbitrary updates. The problem of exploring the Byzantine robustness of FR is particularly critical since in the domains applying FR, e.g., e-commerce, malicious clients can be injected easily by registering new accounts. However, existing Byzantine works neglect the unique sparse aggregation of FR, making them unsuitable for our problem. 
Thus, we make the first effort to investigate Byzantine attacks on FR from the perspective of sparse aggregation, which is non-trivial: it is not clear how to define Byzantine robustness under sparse aggregations and design Byzantine attacks under limited knowledge/capability. 
In this paper, we reformulate the Byzantine robustness under sparse aggregation by defining the aggregation for a single item as the smallest execution unit. Then we propose a family of effective attack strategies, named \textbf{Spattack}, which exploit the vulnerability in sparse aggregation and are categorized along the adversary's knowledge and capability. 
Extensive experimental results demonstrate that Spattack can effectively prevent convergence and even break down defenses under a few malicious clients, raising alarms for securing FR systems. 
\end{abstract}
\section{Introduction}
As an essential way to alleviate information overload, recommender systems are widely used in e-commerce~\cite{Ying2018GraphCN}, media~\cite{Wang2018DKNDK, Wu2019NPANN}, and social network~\cite{Fan2019DeepSC}, recommending items that users may be interested in. Despite the remarkable success, conventional recommender systems require centrally storing users' personal data for training, increasing privacy risks. 

Recently, federated learning (FL)~\cite{McMahan2016CommunicationEfficientLO} has emerged as a privacy-preserving paradigm and successfully applied to the recommendation area. In federated recommendation (FR)~\cite{sun2024survey, Luo2022PersonalizedFR}, the global item embeddings are uploaded to a central server for aggregation. Meanwhile, each user's interaction data and privacy features are kept on the local client. In this way, the privacy of local data is well protected.
\begin{figure}[tp]
    \centering
    \vskip -0.1in
    \includegraphics[width=\linewidth]{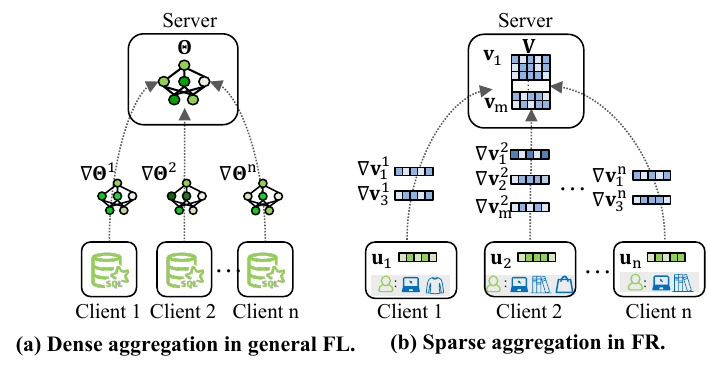}
    \vskip -0.15in
    \caption{Comparisons between dense aggregation of general FL and the unique sparse aggregation of FR.}
    \vskip -0.2in
    \label{fig:dense_sparse}
\end{figure}

Unlike general FL systems, FR has a unique sparse aggregation mechanism. As shown in Fig.~\ref{fig:dense_sparse}, for general FL, each element (circle) of model parameters can be updated by all $n$ clients, named dense aggregation. While for FR, the interactions of users and items are usually sparse~\cite{Ma2008SoRecSR}, resulting in each item's embedding can only be updated by partial clients. For example, client $n$ can only produce and send substantive gradients $\{\nabla \boldsymbol{v}_1^n, \nabla \boldsymbol{v}_3^n\}$ for its interacted items $\{v_1, v_3\}$. For the remaining items, the updates are zero vectors or empty, named sparse aggregation. 

By far, FR has provided satisfactory performance without collecting users' private data, extending recommendation applications to privacy-sensitive scenarios. Despite success, the model security, as an essential principle, has received increasing attention. Here we consider the worst-case attack, i.e., Byzantine attack~\cite{fang2024byzantine, Fang2019LocalMP, Blanchard2017MachineLW}, where attackers are omniscient and collusive, and can control several clients to upload arbitrary malicious gradients. Note that Byzantine robustness is especially critical for FR, since in the domains applying FR, e.g., e-commerce, malicious clients can be injected easily by registering new accounts. This raises one question naturally: \textit{With the unique sparse aggregations, how robust the federated recommendation model is against Byzantine attacks?}

For this question, existing Byzantine works cannot be directly employed, since they mainly focus on dense aggregation in general FL~\cite{RodriguezBarroso2022SurveyOF, Xu2021ByzantinerobustFL, PevaBlanchard2017MachineLW}. Despite a few prior attacks against FR emerges~\cite{yuan2023manipulating, Yu2022UntargetedAA, Rong2022FedRecAttackMP, Wu2022FedAttackEA}, they also neglect to analyze how the sparse aggregation affects the robustness of FR. We answer this problem by solving two challenges: (1) How to define Byzantine robustness under sparse aggregations? Existing Byzantine attacks and defenses are mainly defined based on the dense aggregation mechanism in the general FL. In FR, due to sparse user-item interaction, for an item, its embedding is updated only by its interacted users, and the remaining users upload zero-valued or empty updates. So the aggregated item embedding may be skewed towards zero-valued when directly applying existing dense aggregators, since the zero-value update is majority. Hence, it is vital to transfer them into FR and re-examine their theoretical guarantee and effectiveness. (2) How to design general Byzantine attacks against FR for attackers with different levels of knowledge and capability in reality. Specifically, it is hard to have full knowledge of all users, due to the large number of participating users in FR. Besides, since user-item interactions are usually sparse, Byzantine clients should not update too many items. Otherwise, a monitor based on the number of user interactions can be triggered easily under such aggressive modifications~\cite{Wu2022FedAttackEA}.

In this paper, we make the first effort to investigate the Byzantine robustness of federated recommendation from the perspective of sparse aggregations. 
For the first challenge, we transfer the existing aggregators to FR by treating the aggregation for a single item as the smallest execution unit. Namely, for each item embedding, the gradients are collected and aggregated separately and concurrently. Based on this, we further point out that such a sparse aggregation mechanism of FR will lead to a unique Byzantine vulnerability: items with different degrees receive different amounts of updates, leading to individual robustness. The degrees of all items usually meet long tail distribution in reality~\cite{Abdollahpouri2019ManagingPB}, where most items (named tailed items) are only interacted with by a few users, making them extremely fragile.
For the second challenge, we design a series of attack strategies, named \textbf{Spattack}, based on the vulnerability from sparse aggregation in FR. Then we categorize them along the attacker's knowledge and capability into four classes. To be specific, following~\cite{xie2020fall,Fang2019LocalMP,GiladBaruch2019ALI}, we consider both omniscient attacker (Spattack-O) and limited non-omniscient attacker (Spattack-L). 
Then we further divide them depending on whether limiting the maximum number of each client's poisoned items or not.
In summary, our contributions are three folds:\\
(1) We first systematically study the Byzantine robustness of FR from the perspective of unique sparse aggregation, by treating the aggregation for a single item as the smallest execution unit. We theoretically analyze its convergence guarantee and point out a special vulnerability of FR. \\
(2) We propose a family of effective attack strategies, named Spattack, utilizing the vulnerability from sparse aggregation. Then Spattack can be categorized into four different types along attacker’s knowledge and capability.\\
(3) We perform experiments on multiple benchmark datasets for different FR systems. The results show that our Spattack can prevent the convergence of vanilla even defense FR models by only controlling a few malicious clients. 

\section{Background and Preliminary}
\begin{figure*}[!htbp]
    \centering
    \includegraphics[width=1.0\linewidth]{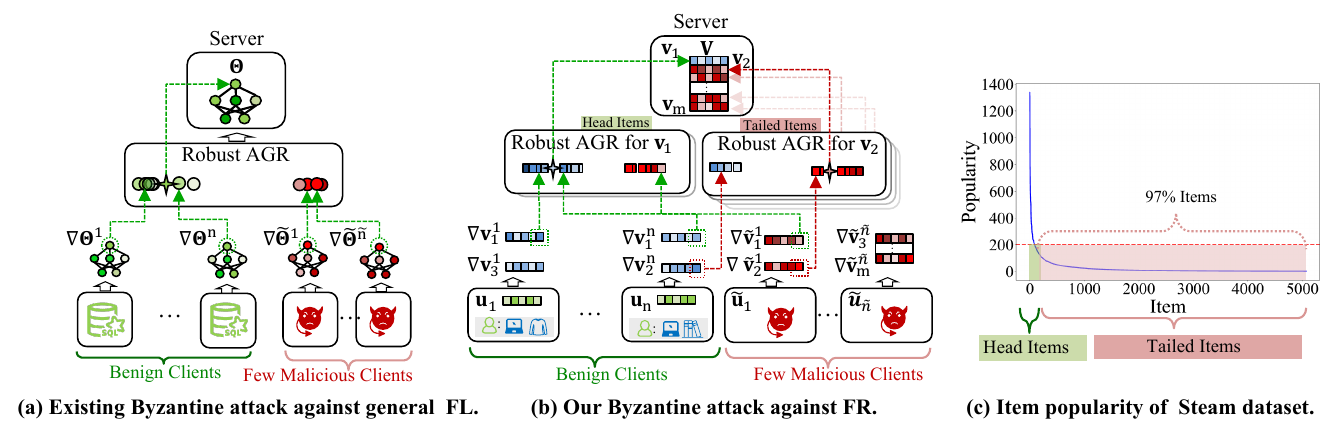}
    \vspace{-22.5pt}
    \caption{Analysis of Byzantine robustness in general FL and FR. Under Byzantine attacks, a robust aggregator can filter outliers in general FL, but fails to defend for tailed items' embedding.}
    \vspace{-17.5pt}
    \label{fig:concept}
\end{figure*}
\subsection{Centralized Recommendation}
Here, a recommender system contains a set of users $\mathcal{U}=\{u_{1}, \cdots, u_n \}$ and a set of items $\mathcal{V}=\{v_{1}, \cdots, v_m\}$, where $n$ and $m$ are the numbers of users and items, respectively.
Each user $u_i \in \mathcal{U}$ has a local training dataset $\mathcal{D}_i$, consisting of implicit feedback tuples $(u_i,v_j,r_{ij})$. These tuples represent user-item interactions (e.g., purchased, clicked), where \( r_{ij} = 1 \) and \( r_{ij} = 0 \) indicate positive and negative instances, respectively, i.e., whether \( u_i \) interacted with \( v_j \).
For each user $u_i$, we define $\mathcal{V}_{u_i}=\left\{v_j \in \mathcal{V}|(u_i,v_j,r_{ij}) \in \mathcal{D}_i \right\}$ as the set of the items that interact with $u_i$.
Let $\boldsymbol{U}=[\boldsymbol{u}_1, \cdots, \boldsymbol{u}_n]$ and $\boldsymbol{V} = [\boldsymbol{v}_1, \cdots, \boldsymbol{v}_m]$ denote the embeddings of users and items, respectively. 
The recommender system is trained to predict the rating score $\hat{r}_{ij}= f_{\boldsymbol{\Theta}}(\boldsymbol{u}_i,\boldsymbol{v}_j)$ between $u_i$ and $v_j$, where $\hat{r}_{ij}$ represents how much $u_{i}$ likes $v_{j}$, $f_{\boldsymbol{\Theta}}$ is the score function, $\{\boldsymbol{U},\boldsymbol{V},\boldsymbol{\Theta}\}$ are learnable parameters. 
Then, the system recommends an item list for each user that they might be interested in by sorting the rating scores. In traditional centralized training, the personal dataset $\mathcal{D}_i$ of each user $u_i$ is stored on a central server, yielding a total dataset $\mathcal{D}$ for model training, which will increase the privacy risks.

\subsection{Federated Recommendation}
Considering privacy issues, in FR, the privacy data $\mathcal{D}_i$ of user $u_i$ is kept on the local device. The shared model parameters $\boldsymbol{V}$ and $\boldsymbol{\Theta}$ are aggregated over clients by sending the local gradients to a central server. According to the base recommender, the parameters $\boldsymbol{\Theta}$ are different: in Matrix Factorization (MF) recommender models, the interaction function is fixed and $\boldsymbol{\Theta}$ is an empty set. In deep learning-based recommender models, $\boldsymbol{\Theta}$ is the set of weights of neural networks. 
Following~\cite{Rong2022FedRecAttackMP}, we adopt the classic and widely used MF as the base recommender for simplicity, where $f$ is fixed to be dot product, i.e., $\hat{r}_{i j}=\boldsymbol{u}_i \odot \boldsymbol{v}_j$. Following~\cite{Rong2022FedRecAttackMP}, we take Bayesian Personalized Ranking (BPR)~\cite{Rendle2009BPRBP}, a pairwise personalized ranking loss, as the local loss of each client:
\begin{align}\label{eq:L_i}
    \mathcal{L}_{i}(\boldsymbol{u}_i, \boldsymbol{V})=-\sum_{ \substack{v_{j}, \ v_{k}  \in \mathcal{V}_{u_i} \\ r_{ij}=1 \ \wedge \ r_{ik}=0}} \ln \sigma\left(\hat{r}_{ij} - \hat{r}_{ik}\right),
\end{align}
where $\sigma$ is the logistic sigmoid function. It assumes that the user prefers the positive items over all negative items. In each training iteration, the central server sends the current item embeddings $\boldsymbol{V}^t$ to all clients. For each user $u_{i}$, the client computes loss $\mathcal{L}_i(\boldsymbol{u}_i^t,\boldsymbol{V}^t)$ then locally updates its private user embedding at epoch $t$ as follows:
\begin{align}
    \boldsymbol{u}_{i}^{t+1} \leftarrow \boldsymbol{u}_{i}^t -\eta \cdot \nabla \boldsymbol{u}_{i}^t,
\end{align}
where $\eta$ is the learning rate. Then $u_i$ uploads its local item embedding gradients $\nabla \boldsymbol{V}^{i,t}$ to a central server. After collecting gradients from all clients, the server updates $\boldsymbol{V}^t$ by:
\begin{align}
    \boldsymbol{V}^{t+1} \leftarrow \boldsymbol{V}^{t}-\eta \cdot \sum_{i \in [n]} \nabla \boldsymbol{V}^{i,t}.
\end{align}
As shown in Fig.~\ref{fig:dense_sparse}(b), for each user $u_i$, the private interaction history (item list) and user embedding (green $\boldsymbol{u}_i$ vector) is preserved on the local client device, and only the gradients of item embeddings $\boldsymbol{V}$ are sent. Throughout the training stage, all users' privacy is well protected. 

\subsection{Byzantine Attack and Defense}
\textbf{Byzantine Attack.} In Byzantine attacks, the attacker aims to degrade model performance and even prevent convergence by controlling a few malicious clients. As shown in Fig.~\ref{fig:concept}(a),  malicious client $\tilde{u}_i$ is allowed to send arbitrary (red) gradient $\nabla \tilde{\boldsymbol{\Theta}}^i$. Following existing Byzantine attack studies~\cite{Xu2021ByzantinerobustFL, Fang2019LocalMP, GiladBaruch2019ALI}, considering the worst case, we assume attackers have full knowledge of all benign gradients $\{\nabla \boldsymbol{\Theta}^1, \cdots, \nabla \boldsymbol{\Theta}^n\}$ and all the malicious clients are collusive by default, which help to understand the severity of model poisoning threats. 

\textbf{Byzantine Defense.} Since servers have no access to the raw training data of clients, the defense is generally implemented on the server side as a robust aggregator, which can filter Byzantine updates and guarantee model convergence. As shown in Fig.~\ref{fig:concept}(a), let $\{\nabla \boldsymbol{\Theta}^1, \cdots, \nabla \boldsymbol{\Theta}^n\}$ be the gradient vectors of $n$ benign clients in FL. The server collects and aggregates the training gradient of each client model using a federated aggregator. In non-robust FL settings, coordinate-wise Mean in form of $\text{MEAN}(\nabla \boldsymbol{\Theta}^1, \cdots, \nabla \boldsymbol{\Theta}^n)=\frac{1}{n} \sum_{i=1}^{n} \nabla \boldsymbol{\Theta}^i$ is an effective aggregation rule. However, $\text{MEAN}$ can be manipulated by several malicious clients~\cite{PevaBlanchard2017MachineLW}. Therefore, multiple robust aggregators~\cite{DongYin2018ByzantineRobustDL, PevaBlanchard2017MachineLW, Xu2021ByzantinerobustFL} are proposed to filter the Byzantine updates. For example, coordinate-wise Median aggregator computes the median for each element $\Theta_i$ in parameter $\boldsymbol{\Theta}$ across all clients, yielding 0.5 breakdown point~\cite{DongYin2018ByzantineRobustDL}. Namely, when the fraction of malicious clients is less than 0.5, the Median aggregator can guarantee the model convergence under Byzantine attacks, yielding the correct gradient (green star) in Fig.~\ref{fig:concept}(a).

\section{Methodology}
In this section, we re-define the Byzantine robustness under sparse aggregation of FR, and theoretically point out the inherent vulnerability. Based on such vulnerability and considering attackers' knowledge and capability, we design a family of attack strategies, named Spattack.

\subsection{Problem Definition}
For Byzantine attacks, attackers can inject some Byzantine users $\tilde{\mathcal{U}}=\{\tilde{u}_1,\cdots,\tilde{u}_{\tilde{n}}\}$, limiting the proportion of malicious ones less than $\rho$, i.e., $\tilde{n}/(n+\tilde{n}) < \rho$.
A malicious client $\tilde{u}_i$ can upload arbitrary gradient values $\nabla \tilde{\boldsymbol{V}}^{i,t}$ at any epoch $t$, to directly perturb the item embedding. The server will collect and aggregate all gradients including benign $\{\nabla {\boldsymbol{V}}^{1,t}, \cdots, \nabla {\boldsymbol{V}}^{n,t}\}$ and malicious $\{ \nabla \tilde{\boldsymbol{V}}^{1,t}, \cdots, \nabla \tilde{\boldsymbol{V}}^{\tilde{n},t}\}$. Let $\text{AGR}(\cdot)$ be the aggregation operator of federated learning, which can be the most common $\text{MEAN}(\cdot)$ or statistically robust $\text{MEDIAN}(\cdot)$. Our Byzantine attacker aims to prevent model convergence, namely, keeping the recommendation loss $\mathcal{L}_i$ from decreasing.
Formally, in FR, the objective of the Byzantine attack is defined as the following optimization problem:
\begin{align}\label{eq:obj_1}
    & \max_{\{\nabla \tilde{\boldsymbol{V}}_i^t: i \in \tilde{n}\}} 
    \sum_{i=1}^{n} \left(\mathcal{L}_i(\boldsymbol{u}_i^{t+1}, \boldsymbol{V}^{t+1}) 
    - \mathcal{L}_i(\boldsymbol{u}_i^t, \boldsymbol{V}^t)\right), \notag \\
    & \text{s.t.} \ \boldsymbol{V}^{t+1} = \boldsymbol{V}^{t} - \eta \cdot 
    \text{AGR}(\{\nabla \boldsymbol{V}^{i,t}: i \in [n]\} \notag \\
    & \quad \quad \quad \quad \ \ \ \cup \{\nabla \tilde{\boldsymbol{V}}^{i,t}: i \in [\tilde{n}]\}), \notag \\
    & \quad \ \ \boldsymbol{u}_i^{t+1} = \boldsymbol{u}_i^{t} - \eta \nabla \boldsymbol{u}_i^t, 
    \quad \text{for } i \in [n], \notag \\
    & \quad \ \ \frac{\tilde{n}}{n+\tilde{n}} \leq \rho,
\end{align}
where attackers aim to find the optimal set of malicious gradients ${\{\nabla \tilde{\boldsymbol{V}}_i^t: i \in [\tilde{n}]\}}$, to raise the loss after updating. 
Before solving this optimization problem, we find that FR has a unique sparse aggregation mechanism defined as follows:
\begin{definition}
\textbf{(Dense/Sparse Aggregation).} Let $\boldsymbol{\theta} \in \mathbb{R}^d$ be the shared model parameter vector. If there exists an element $\theta_{i}$ ($i \in [d]$) for which only a subset of clients can produce valuable updates, the parameter $\boldsymbol{\theta}$ is sparsely aggregated. If all clients can support it, it is a dense aggregation. 
\end{definition}
As shown in Fig.~\ref{fig:concept}(a), in general FL, each element (green circle) of model parameters $\boldsymbol{\Theta}$ is assumed to be involved in all clients' loss functions, i.e., dense aggregation. Different from FL, for a client $u_i$ in FR, not all item embeddings $\boldsymbol{V}=\{\boldsymbol{v}_1,\cdots, \boldsymbol{v}_m\}$ are employed in local loss $\mathcal{L}_i$ in Eq.~\ref{eq:L_i}, i.e., sparse aggregation. For example, client $u_1$ in Fig.~\ref{fig:concept}(b) only computes valuable gradients $\{\nabla \boldsymbol{v}^1_1, \nabla \boldsymbol{v}^1_3\}$ for items $v_1$ and $v_3$, while the gradients for the remaining items are either zero or an empty set. When directly applying the aggregators on all $\boldsymbol{V}_i$, the embedding will be skewed towards zero value. Therefore, we need to adapt them to FR. 

\textbf{Adapting Dense Aggregator to Sparse.}
We adapt existing aggregators from dense to sparse aggregation by treating the aggregation for a single item as the smallest execution unit. As shown in Fig.~\ref{fig:concept}(b), the aggregator is conducted separately for each item. Taking the embedding of $j$-th item as an example, the embedding is updated by: 
\begin{equation}
    \boldsymbol{v}^{t+1}_j = \boldsymbol{v}^{t}_j - \eta \cdot \text{AGR}(\{\nabla \boldsymbol{v}_j^{i,t}| \text{ user } i  \in \mathcal{U}_{v_j} \}), \\
\end{equation}
where $\mathcal{U}_{v_j}$ is the set of users that the item $v_j$ interacts with, $\nabla \boldsymbol{v}_j^{i,t}$ is the gradient of item $v_j$ sent from client $u_i$ at epoch $t$. Only if the user $u_i$ has interaction with item $v_j$, the gradients $\nabla \boldsymbol{v}_j^{i,t}$ can be aggregated to $\boldsymbol{v}_j^{t+1}$ separately and concurrently. Intuitively, the numbers of received gradients are varied for different items, leading to each item having personal robustness. Therefore, we need to theoretically re-examine the convergence guarantee of existing aggregators against Byzantine attacks under the sparse aggregation. 

\subsection{Byzantine Robustness Analysis}

\textbf{Robustness of FR without Defense.} Like general FL, FR without defense often uses the Mean aggregator to compute the average of input gradients, which is highly susceptible to Byzantine attacks. Even one malicious client can also destroy the Mean aggregator as stated in Proposition~\ref{prop1}.

\begin{prop} For each item $v_j$, let $\{\nabla \boldsymbol{v}^{i,t}_j| \text{ user } i  \in \mathcal{U}_{v_j}\}$ be the set of benign gradient vectors at epoch $t$. Consider a Mean aggregator averaging updates for each element.
Let $\nabla \boldsymbol{\tilde{v}}_j$ be a malicious update with arbitrary values in $\mathbb{R}^{d}$. The output of
$\text{MEAN}(\{\nabla \boldsymbol{v}^{i,t}_j| \text{ user } i   \in \mathcal{U}_{v_j}\} \cup {\nabla \tilde{\boldsymbol{v}}_j})= \frac{1}{|\mathcal{U}_{v_j}|+1} (\sum_{i \in \mathcal{U}_{v_j}} \nabla \boldsymbol{v}_{j}^{i,t} + \nabla \tilde{\boldsymbol{v}}_j)$ can be controlled as zero vector by only single malicious $\nabla \boldsymbol{\tilde{v}}_j$. When all the items are attacked, one malicious client can prevent convergence.\label{prop1}
\end{prop} 
\begin{proof}
If the attacker registers one malicious client, where the embedding gradient of each item $v_j$ is $\nabla \boldsymbol{\tilde{v}}_{j}= - \sum_{i \in \mathcal{U}_j} \nabla \boldsymbol{v}_{j}^{i,t}$, the output of aggregator is zero vector, which can prevent convergence. 
\end{proof}

\textbf{Robustness of FR with Defense.}
The most common defense method is to use aggregators that are statistically more robust against outliers than Mean.
In these defenses, FL models have a consistently high breakdown point, e.g., when $\rho < 50\%$, Median can theoretically guarantee the convergence of FL as proved in~\cite{DongYin2018ByzantineRobustDL}.
However, we find that FR models have varied breakdown points for different items, which depends on the item's degree. Specifically, each item embedding can only be updated by specific clients with whom the item interacts. Obviously, the popular item with massive updates is more robust. Unfortunately, in FR, only a few items interact frequently (head items), while the remaining items interact less frequently (tailed items). We plot the popularity (item degree) of the Steam recommendation dataset~\cite{Cheuque2019RecommenderSF} in Fig.~\ref{fig:concept}(c). We find that 97\% tailed items (red long tail area) have interactions less than 200 times, and only 3\% head items (green area) interact frequently over 200 times.
Therefore, existing statistically-based FL defenses will fail to guarantee the convergence of most items.
Formally, let $x$ be the degree of an item and $p(x)$ be its probability. We assume that the probability distribution can be defined as a typical power-law distribution $p(x) = Cx^{-\beta}$, where $C$ is a normalization constant and $\beta$ is the scaling parameter. The failure of defenses can be formally characterized as follows:
\begin{prop} \label{prop2}
Let $\alpha$ be the breakdown point of robust federated aggregator, the amount of benign and malicious clients are $n$ and $\tilde{n}$ respectively, and $\beta$ is the scaling parameter of the power-law distribution of items' degree with constant $C$. Then at least $1-
    \frac{C}{\beta-1} (\frac{1-\alpha}{\alpha} \tilde{n})^{(1-\beta)}$ percent of items' embeddings can be broken down.
\end{prop}
The proof of Proposition~\ref{prop2} refers to the Appendix. Taking the Steam dataset as an example, the degree distribution of items can be modeled as a typical form of power-law distribution as shown in Fig.~\ref{fig:concept}(c). For example, if an attacker can control $\rho=5\%$ clients, each item can receive 197 malicious gradients at most. Clearly, for 97\% tailed items that interact less than 200 times, few malicious (red) updates can become the majority and dominate the aggregation. In this case, the statistically robust Median aggregator will pick the majority (red circles), yielding the malicious output (red star). In conclusion, due to the sparse aggregation vulnerability of FR, statistically robust aggregators in FL can also be easily broken down by Byzantine attacker.

\subsection{Spattack: Byzantine Attack Strategies}

\textbf{Intuition.}
In Eq.~\ref{eq:obj_1}, the attacker aims to keep the recommendation loss from decreasing to prevent recommender convergence.
Considering the unique vulnerability from sparse aggregation, i.e., the majority of tailed items have a lower breakdown point, we can conclude that: (1) The gradients are farther away from true gradients, the more considerable corruption is. (2) More items are disrupted in the training process, leading to more powerful attacks. 
Therefore, the attack objective of the proposed Spattack can be simplified to maximally uploading gradients farther away from true gradients and greedily disrupting the embeddings of items.

\begin{table}[tp]
    \centering
    \caption{Attack Taxonomy. For each malicious client in Spattack, knowledge means knowing benign gradients, and capability refers to poisoning all items.}
    \vspace{-8pt}
    \begin{tabular}{ccccc}
    \hline
    Spattack &O-D  &O-S &L-D  &L-S\\\hline
    Knowledge &\usym{2713} &\usym{2713} &\usym{2717} &\usym{2717} \\\hline 
    Capability  &\usym{2713} &\usym{2717} &\usym{2713}  &\usym{2717} \\
    \hline
    \end{tabular}
    \vspace{-18pt}
    \label{tab:taxonomy}
\end{table}
\textbf{Attack Taxonomy.} 
In real scenarios, depending on the attacker's knowledge about benign gradients and the maximum number of poisoned items in each malicious client, we outline different scenarios of Spattack that can be launched. As shown in Tab.~\ref{tab:taxonomy}, we have four possible scenarios:

\textbf{Spattack-O-D} is considered a worst case, where the attacker is both omniscient and omnipotent, i.e., attackers can obtain benign gradients at each epoch and the maximum number of poisoned items is not limited. 
Following the first intuition that the malicious gradients farther away from true gradients can cause larger corruption, attackers upload the gradients in the opposite direction of the benign ones. Formally, for a item $v_j$, we collect benign gradient $\nabla \boldsymbol{v}_{j}^{i,t}$ from $u_i$, where $u_i$ interacts with $v_j$, i.e., $u_i \in \mathcal{U}_{v_j}$. Then we compute the sum of the collected benign gradients to obtain the expected gradient $\nabla \boldsymbol{\bar{v}}_{j}^{t} = \sum_{u_i \in \mathcal{U}_{v_j}} \nabla \boldsymbol{v}_{j}^{i,t}$. Lastly, each malicious client $\tilde{u}_i \in \tilde{\mathcal{U}}$ will upload malicious gradients $\nabla \boldsymbol{\tilde{v}}_{j}^{i,t}= - \frac{1}{|\tilde{\mathcal{U}}|} \nabla \boldsymbol{\bar{v}}_{j}^{t}$. Following the second intuition that greedily disrupts items, the attack effectiveness will be maximized by uploading poisoning gradients for all items.
In this attack, the non-robust Mean aggregator will output zero gradients, while statistically robust aggregators will select malicious gradients for the majority of tailed items, preventing the convergence of item embeddings. According to Proposition~\ref{prop1} and Proposition~\ref{prop2}, even only having a small portion of malicious clients, Spattack-O-D can still guarantee to disrupt the majority of item embeddings.

\textbf{Spattack-L-D} uploads random noise as malicious gradients for all items, where attackers are non-omniscient but omnipotent, i.e., attackers do not have any knowledge about the benign gradients but can attack all items. Specifically, attackers construct the malicious gradient by randomly sampling from the Gaussian noise and keeping the same noise in all malicious clients. Under the Mean aggregator, the aggregated gradients can be skewed by such noise. Even worse, the statistically robust aggregators, e.g., Median, can pick the uploaded random noise as output for tailed items. So this attack can still prevent model convergence. 

\textbf{Spattack-O-S and Spattack-L-S} only upload malicious gradients for partial items, where attackers are non-omnipotent. Let $\tilde{m}_{max}$ be the maximum number of poisoned items in each malicious client. The larger $\tilde{m}_{max}$, the stronger the attack, but the excessive $\tilde{m}_{max}$ may lead to the attack being detected. To limit malicious users to behaving like benign users, we restrict $\tilde{m}_{max}$ as the maximum number of interactions in benign clients. Specifically, to make the injections of malicious clients as imperceptible and effective as possible, based on the distribution of item popularity, we use a sampling operation to determine the poisoned items for each malicious client. Therefore, the attacker can automatically assign more malicious gradients to the items having more interactions. Then we generate malicious gradients based on the opposite benign gradients (Spattack-O-S) or random noise (Spattack-L-S), respectively.

\begin{table}[tp]
    \centering
    \caption{Statistics of datasets.}
    \vskip -0.1in
    \resizebox{\linewidth}{!}{
    \begin{tabular}{c|c|c|c|c}
    \hline \text { Dataset } & \text { \#Users } & \text { \#Items } & \text { \#Edges }  & \text { Sparsity } \\
    \hline \text { ML100K } & 943 & 1,682 & 100,000 & 93.70 \% \\
    \hline \text { ML1M } & 6,040 & 3,706 & 1,000,209  & 95.53 \% \\
    \hline \text { Steam } & 3,753 & 5,134 & 114,713 & 99.40 \% \\
    \hline
    \end{tabular}}
    \vskip -0.2in
    \label{tab:dataset}
\end{table}
\begin{table*}[ht]
\caption{
Comparison of Spattack with baselines under a 3\% malicious rate. Lower scores represent better attack effectiveness. We additionally report the performance drop (\%) compared with the performance on the clean model.
}
\label{tab:sota}
\vskip -0.15in
\centering
\setlength{\extrarowheight}{2pt}
\resizebox{\linewidth}{!}{
\begin{tabular}{c|c|c|cccccc|cccc}
\hline
Dataset                 & Metric  & Clean & LabelFlip                                                & FedAttack                                                & Gaussian                                                 & LIE                                                      & Cluster                                                  & Fang                                                     & Type L-S                                                & Type L-D                                                         & Type O-S                                                 & Type O-D                                                          \\ \hline
\multirow{7}{*}{ML100K} & HR@5    & \begin{tabular}[c]{@{}c@{}}0.2513\end{tabular} & \begin{tabular}[c]{@{}c@{}}0.2517\\ (-3\%)\end{tabular}  & \begin{tabular}[c]{@{}c@{}}0.2550\\ (-2\%)\end{tabular}  & \begin{tabular}[c]{@{}c@{}}0.2550\\ (-2\%)\end{tabular}  & \begin{tabular}[c]{@{}c@{}}0.2539\\ (-2\%)\end{tabular}  & \begin{tabular}[c]{@{}c@{}}0.2461\\ (-5\%)\end{tabular}  & \begin{tabular}[c]{@{}c@{}}0.1957\\ (-25\%)\end{tabular} & \begin{tabular}[c]{@{}c@{}}0.1018\\ (-59\%)\end{tabular} & \begin{tabular}[c]{@{}c@{}}0.0721\\ (-71\%)\end{tabular}          & \begin{tabular}[c]{@{}c@{}}0.0594\\ (-76\%)\end{tabular} & \textbf{\begin{tabular}[c]{@{}c@{}}0.0530\\ (-79\%)\end{tabular}} \\
                        & nDCG@5  & \begin{tabular}[c]{@{}c@{}}0.1643\end{tabular} & \begin{tabular}[c]{@{}c@{}}0.1706\\ (-3\%)\end{tabular}  & \begin{tabular}[c]{@{}c@{}}0.1721\\ (-2\%)\end{tabular}  & \begin{tabular}[c]{@{}c@{}}0.1729\\ (-1\%)\end{tabular}  & \begin{tabular}[c]{@{}c@{}}0.1724\\ (-2\%)\end{tabular}  & \begin{tabular}[c]{@{}c@{}}0.1678\\ (-4\%)\end{tabular}  & \begin{tabular}[c]{@{}c@{}}0.1229\\ (-30\%)\end{tabular} & \begin{tabular}[c]{@{}c@{}}0.0620\\ (-62\%)\end{tabular} & \begin{tabular}[c]{@{}c@{}}0.0380\\ (-77\%)\end{tabular}          & \begin{tabular}[c]{@{}c@{}}0.0362\\ (-78\%)\end{tabular} & \textbf{\begin{tabular}[c]{@{}c@{}}0.0339\\ (-79\%)\end{tabular}} \\
                        & HR@10   & \begin{tabular}[c]{@{}c@{}}0.4051\end{tabular} & \begin{tabular}[c]{@{}c@{}}0.4083\\ (-2\%)\end{tabular}  & \begin{tabular}[c]{@{}c@{}}0.4094\\ (-2\%)\end{tabular}  & \begin{tabular}[c]{@{}c@{}}0.4116\\ (-2\%)\end{tabular}  & \begin{tabular}[c]{@{}c@{}}0.4116\\ (-2\%)\end{tabular}  & \begin{tabular}[c]{@{}c@{}}0.3982\\ (-5\%)\end{tabular}  & \begin{tabular}[c]{@{}c@{}}0.2919\\ (-30\%)\end{tabular} & \begin{tabular}[c]{@{}c@{}}0.2163\\ (-47\%)\end{tabular} & \begin{tabular}[c]{@{}c@{}}0.1601\\ (-60\%)\end{tabular}          & \begin{tabular}[c]{@{}c@{}}0.0997\\ (-75\%)\end{tabular} & \textbf{\begin{tabular}[c]{@{}c@{}}0.0944\\ (-77\%)\end{tabular}} \\
                        & nDCG@10 & \begin{tabular}[c]{@{}c@{}}0.2131\end{tabular} & \begin{tabular}[c]{@{}c@{}}0.2206\\ (-2\%)\end{tabular}  & \begin{tabular}[c]{@{}c@{}}0.2213\\ (-2\%)\end{tabular}  & \begin{tabular}[c]{@{}c@{}}0.2230\\ (-1\%)\end{tabular}  & \begin{tabular}[c]{@{}c@{}}0.2229\\ (-1\%)\end{tabular}  & \begin{tabular}[c]{@{}c@{}}0.2166\\ (-4\%)\end{tabular}  & \begin{tabular}[c]{@{}c@{}}0.1541\\ (-32\%)\end{tabular} & \begin{tabular}[c]{@{}c@{}}0.0980\\ (-54\%)\end{tabular} & \begin{tabular}[c]{@{}c@{}}0.0658\\ (-69\%)\end{tabular}          & \begin{tabular}[c]{@{}c@{}}0.0492\\ (-77\%)\end{tabular} & \textbf{\begin{tabular}[c]{@{}c@{}}0.0470\\ (-78\%)\end{tabular}} \\ \hline
\multirow{7}{*}{ML1M}   & HR@5    & \begin{tabular}[c]{@{}c@{}}0.3121\end{tabular} & \begin{tabular}[c]{@{}c@{}}0.3051\\ (-1\%)\end{tabular}  & \begin{tabular}[c]{@{}c@{}}0.3056\\ (-1\%)\end{tabular}  & \begin{tabular}[c]{@{}c@{}}0.3053\\ (-1\%)\end{tabular}  & \begin{tabular}[c]{@{}c@{}}0.3054\\ (-1\%)\end{tabular}  & \begin{tabular}[c]{@{}c@{}}0.3033\\ (-2\%)\end{tabular}  & \begin{tabular}[c]{@{}c@{}}0.2827\\ (-9\%)\end{tabular}  & \begin{tabular}[c]{@{}c@{}}0.1007\\ (-68\%)\end{tabular} & \begin{tabular}[c]{@{}c@{}}0.0921\\ (-71\%)\end{tabular}          & \begin{tabular}[c]{@{}c@{}}0.0925\\ (-70\%)\end{tabular} & \textbf{\begin{tabular}[c]{@{}c@{}}0.0907\\ (-71\%)\end{tabular}} \\
                        & nDCG@5  & \begin{tabular}[c]{@{}c@{}}0.2054\end{tabular} & \begin{tabular}[c]{@{}c@{}}0.2013\\ (-2\%)\end{tabular}  & \begin{tabular}[c]{@{}c@{}}0.2021\\ (-1\%)\end{tabular}  & \begin{tabular}[c]{@{}c@{}}0.2017\\ (-1\%)\end{tabular}  & \begin{tabular}[c]{@{}c@{}}0.2018\\ (-1\%)\end{tabular}  & \begin{tabular}[c]{@{}c@{}}0.2004\\ (-2\%)\end{tabular}  & \begin{tabular}[c]{@{}c@{}}0.1858\\ (-9\%)\end{tabular}  & \begin{tabular}[c]{@{}c@{}}0.0581\\ (-72\%)\end{tabular} & \textbf{\begin{tabular}[c]{@{}c@{}}0.0521\\ (-75\%)\end{tabular}} & \begin{tabular}[c]{@{}c@{}}0.0553\\ (-73\%)\end{tabular} & \begin{tabular}[c]{@{}c@{}}0.0549\\ (-73\%)\end{tabular}          \\
                        & HR@10   & \begin{tabular}[c]{@{}c@{}}0.4626\end{tabular} & \begin{tabular}[c]{@{}c@{}}0.4632\\ (-1\%)\end{tabular}  & \begin{tabular}[c]{@{}c@{}}0.4634\\ (-1\%)\end{tabular}  & \begin{tabular}[c]{@{}c@{}}0.4634\\ (-1\%)\end{tabular}  & \begin{tabular}[c]{@{}c@{}}0.4634\\ (-1\%)\end{tabular}  & \begin{tabular}[c]{@{}c@{}}0.4592\\ (-2\%)\end{tabular}  & \begin{tabular}[c]{@{}c@{}}0.3977\\ (-15\%)\end{tabular} & \begin{tabular}[c]{@{}c@{}}0.2141\\ (-54\%)\end{tabular} & \begin{tabular}[c]{@{}c@{}}0.1935\\ (-58\%)\end{tabular}          & \begin{tabular}[c]{@{}c@{}}0.1753\\ (-62\%)\end{tabular} & \textbf{\begin{tabular}[c]{@{}c@{}}0.1679\\ (-64\%)\end{tabular}} \\
                        & nDCG@10 & \begin{tabular}[c]{@{}c@{}}0.2539\end{tabular} & \begin{tabular}[c]{@{}c@{}}0.2522\\ (-1\%)\end{tabular}  & \begin{tabular}[c]{@{}c@{}}0.2528\\ (-1\%)\end{tabular}  & \begin{tabular}[c]{@{}c@{}}0.2526\\ (-1\%)\end{tabular}  & \begin{tabular}[c]{@{}c@{}}0.2526\\ (-1\%)\end{tabular}  & \begin{tabular}[c]{@{}c@{}}0.2506\\ (-2\%)\end{tabular}  & \begin{tabular}[c]{@{}c@{}}0.2231\\ (-12\%)\end{tabular} & \begin{tabular}[c]{@{}c@{}}0.0939\\ (-63\%)\end{tabular} & \begin{tabular}[c]{@{}c@{}}0.0846\\ (-67\%)\end{tabular}          & \begin{tabular}[c]{@{}c@{}}0.0817\\ (-68\%)\end{tabular} & \textbf{\begin{tabular}[c]{@{}c@{}}0.0793\\ (-69\%)\end{tabular}} \\ \hline
\multirow{7}{*}{Steam}  & HR@5    & \begin{tabular}[c]{@{}c@{}}0.5729\end{tabular} & \begin{tabular}[c]{@{}c@{}}0.4792\\ (-15\%)\end{tabular} & \begin{tabular}[c]{@{}c@{}}0.4798\\ (-15\%)\end{tabular} & \begin{tabular}[c]{@{}c@{}}0.4879\\ (-14\%)\end{tabular} & \begin{tabular}[c]{@{}c@{}}0.4862\\ (-14\%)\end{tabular} & \begin{tabular}[c]{@{}c@{}}0.4263\\ (-25\%)\end{tabular} & \begin{tabular}[c]{@{}c@{}}0.0278\\ (-95\%)\end{tabular} & \begin{tabular}[c]{@{}c@{}}0.0426\\ (-93\%)\end{tabular} & \textbf{\begin{tabular}[c]{@{}c@{}}0.0139\\ (-98\%)\end{tabular}} & \begin{tabular}[c]{@{}c@{}}0.0671\\ (-88\%)\end{tabular} & \begin{tabular}[c]{@{}c@{}}0.0685\\ (-88\%)\end{tabular}          \\
                        & nDCG@5  & \begin{tabular}[c]{@{}c@{}}0.3815\end{tabular} & \begin{tabular}[c]{@{}c@{}}0.3157\\ (-17\%)\end{tabular} & \begin{tabular}[c]{@{}c@{}}0.3172\\ (-16\%)\end{tabular} & \begin{tabular}[c]{@{}c@{}}0.3216\\ (-15\%)\end{tabular} & \begin{tabular}[c]{@{}c@{}}0.3209\\ (-15\%)\end{tabular} & \begin{tabular}[c]{@{}c@{}}0.2750\\ (-27\%)\end{tabular} & \begin{tabular}[c]{@{}c@{}}0.0160\\ (-96\%)\end{tabular} & \begin{tabular}[c]{@{}c@{}}0.0261\\ (-93\%)\end{tabular} & \textbf{\begin{tabular}[c]{@{}c@{}}0.0080\\ (-98\%)\end{tabular}} & \begin{tabular}[c]{@{}c@{}}0.0390\\ (-90\%)\end{tabular} & \begin{tabular}[c]{@{}c@{}}0.0408\\ (-89\%)\end{tabular}          \\
                        & HR@10   & \begin{tabular}[c]{@{}c@{}}0.6933\end{tabular} & \begin{tabular}[c]{@{}c@{}}0.6429\\ (-7\%)\end{tabular}  & \begin{tabular}[c]{@{}c@{}}0.6431\\ (-7\%)\end{tabular}  & \begin{tabular}[c]{@{}c@{}}0.6474\\ (-6\%)\end{tabular}  & \begin{tabular}[c]{@{}c@{}}0.6471\\ (-6\%)\end{tabular}  & \begin{tabular}[c]{@{}c@{}}0.6220\\ (-10\%)\end{tabular} & \begin{tabular}[c]{@{}c@{}}0.0619\\ (-91\%)\end{tabular} & \begin{tabular}[c]{@{}c@{}}0.0834\\ (-88\%)\end{tabular} & \textbf{\begin{tabular}[c]{@{}c@{}}0.0322\\ (-95\%)\end{tabular}} & \begin{tabular}[c]{@{}c@{}}0.1308\\ (-81\%)\end{tabular} & \begin{tabular}[c]{@{}c@{}}0.1287\\ (-81\%)\end{tabular}          \\
                        & nDCG@10 & \begin{tabular}[c]{@{}c@{}}0.4207\end{tabular} & \begin{tabular}[c]{@{}c@{}}0.3685\\ (-12\%)\end{tabular} & \begin{tabular}[c]{@{}c@{}}0.3700\\ (-12\%)\end{tabular} & \begin{tabular}[c]{@{}c@{}}0.3732\\ (-11\%)\end{tabular} & \begin{tabular}[c]{@{}c@{}}0.3730\\ (-11\%)\end{tabular} & \begin{tabular}[c]{@{}c@{}}0.3386\\ (-19\%)\end{tabular} & \begin{tabular}[c]{@{}c@{}}0.0269\\ (-94\%)\end{tabular} & \begin{tabular}[c]{@{}c@{}}0.0391\\ (-91\%)\end{tabular} & \textbf{\begin{tabular}[c]{@{}c@{}}0.0138\\ (-97\%)\end{tabular}} & \begin{tabular}[c]{@{}c@{}}0.0593\\ (-86\%)\end{tabular} & \begin{tabular}[c]{@{}c@{}}0.0601\\ (-86\%)\end{tabular}          \\ \hline
\end{tabular}}
\vskip -0.25in
\end{table*}
\section{Experiment}
We conduct extensive experiments to answer the following research questions. \textbf{RQ1}: How does Spattack perform compared with existing Byzantine attacks? \textbf{RQ2}: Can Spattack break the defenses deployed on FR? \textbf{RQ3}: Can Spattack transfer to different FR systems? \textbf{RQ4}: How do hyper-parameters impact on Spattack? Given the limited space, 
please refer to the Appendix for more detailed experiments.

\subsection{Experimental Setup}
\textbf{Datasets and Federated Recommender Systems.} Following~\cite{Rong2022FedRecAttackMP}, Spattack is evaluated on three widely used datasets, including movie recommendation datasets ML1M and ML100K~\cite{Harper2016TheMD}, and game recommendation dataset Steam~\cite{Cheuque2019RecommenderSF}. The dataset statistics refer to Tab.~\ref{tab:dataset}. The test set is divided with the leave-one-out method, where the latest interaction of a user is left as the test set and the remaining interactions as the training set. FedMF~\cite{Rong2022FedRecAttackMP} and the SOTA FedGNN~\cite{Wu2021FedGNN} are selected as evaluation models. More dataset and reproducibility details are in the Appendix.
    
\textbf{Evaluation Protocols.} 
We utilize two common evaluation protocols, including hit ratio (HR) and normalized discounted cumulative gain (nDCG) at ranks 5 and 10. For each user, since ranking the test item among all items is time-consuming, following the widely-used strategy~\cite{He2017NeuralCF}, we randomly sample 100 items that do not interact with the user, then rank the test item among the 100 items. Notably, all metrics are only calculated on benign clients. 

\textbf{Baselines}.
We compare Spattack with two categories of methods. First is the data poisoning attack, where attackers generate malicious gradients by modifying training data. LabelFlip~\cite{ValeTolpegin2020DataPA} flips training labels for poisoning, while FedAttack~\cite{Wu2022FedAttackEA} uses misaligned samples. Second is model poisoning attacks, where attackers directly modify the uploaded gradients. Gaussian~\cite{Fang2019LocalMP} estimates the Gaussian distribution of benign gradients and then samples from it. LIE~\cite{GiladBaruch2019ALI} adds small amounts of noise towards the average of benign gradients. Cluster~\cite{Yu2022UntargetedAA} uploads malicious gradients that aim to make item embeddings collapse into several dense clusters. Fang~\cite{Fang2019LocalMP} adds noise to opposite directions of the average normal gradient. More details can be found in the Appendix.

\textbf{Byzantine Defense Strategies.} We evaluate Spattack performance under the following defense strategies: Mean is the vanilla non-robust aggregator that computes the mean value of gradients for each dimension. Median ~\cite{DongYin2018ByzantineRobustDL} is a statistically robust aggregator with a 0.5 breakdown point, computing the element-wise median value. Trimmed-mean ~\cite{DongYin2018ByzantineRobustDL} trims several extreme values for each dimension and then averages the rest. Krum ~\cite{PevaBlanchard2017MachineLW} picks the gradient that is the most similar to other uploaded gradients. Norm~\cite{AnandaTheerthaSuresh2019CanYR} clips the norm of gradients with a given threshold.

\begin{figure*}[tp]
\centering
\minipage{0.23\textwidth}
\centering
 \subfigure[Steam (Spattack-O-D)]{\includegraphics[ width=\textwidth]{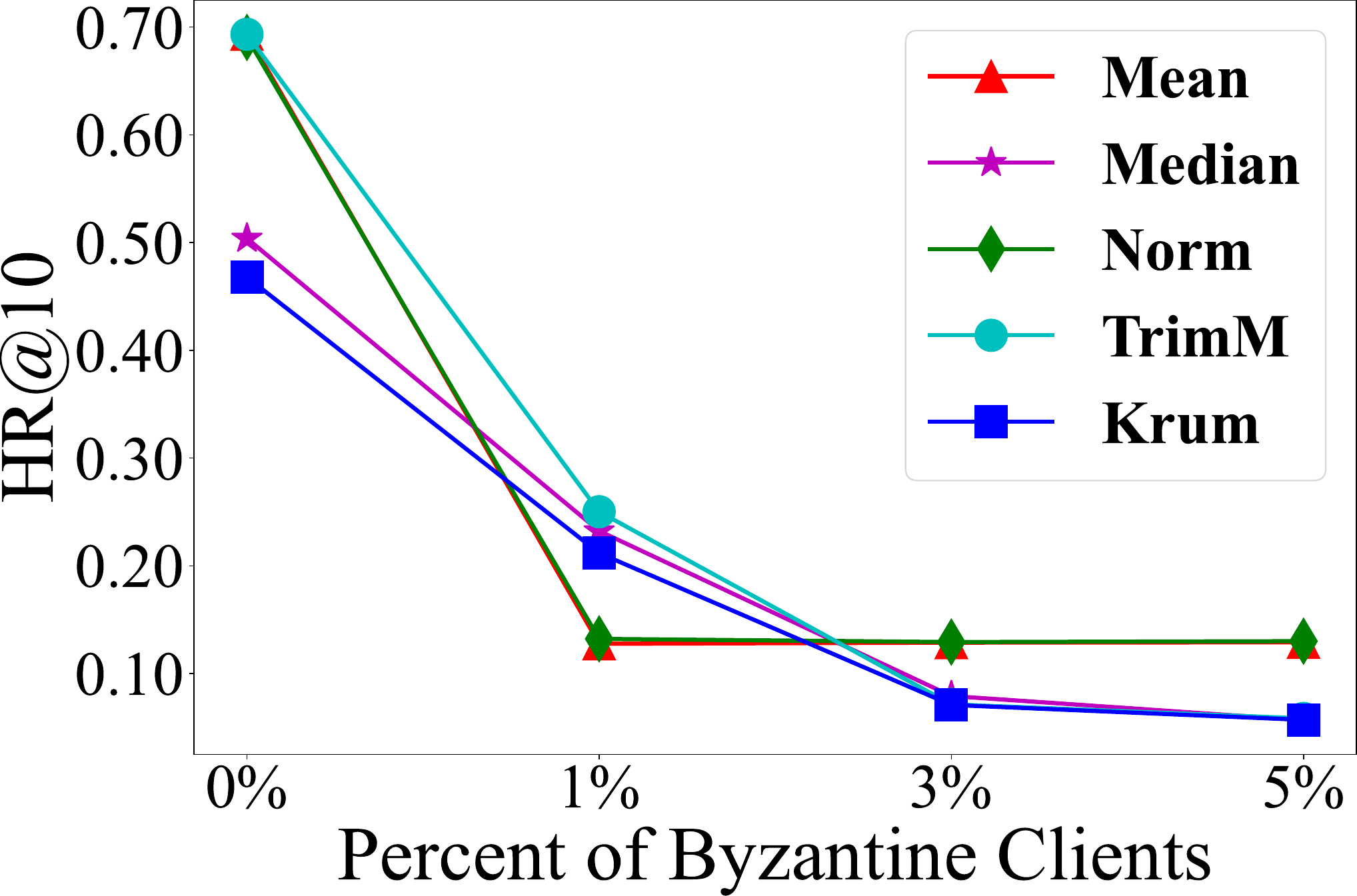}}
\endminipage\hfill
\minipage{0.23\textwidth}
\centering
 \subfigure[Steam (Spattack-O-S)]{\includegraphics[ width=\textwidth]{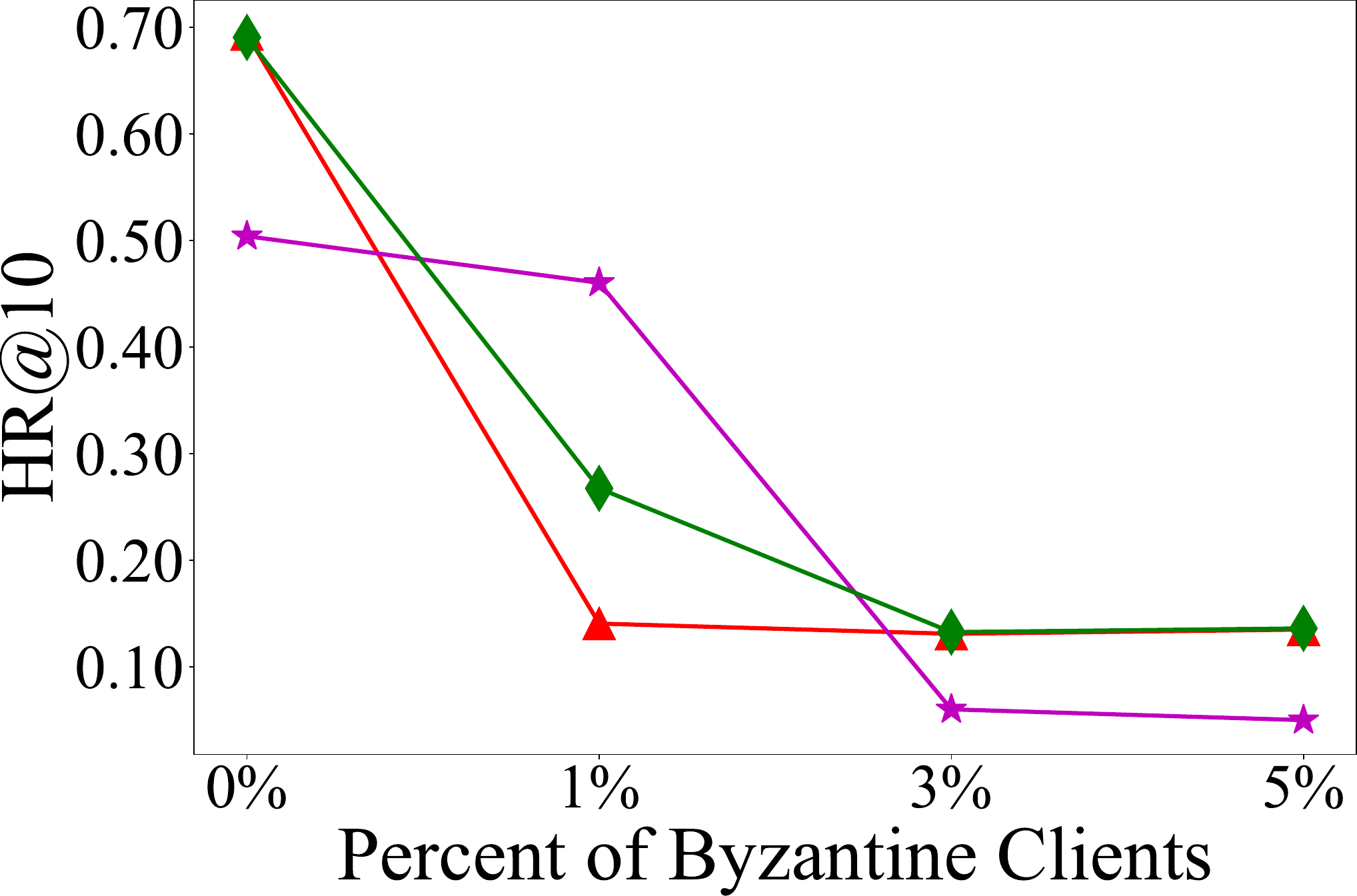}}
\endminipage\hfill
\minipage{0.23\textwidth}
\centering
 \subfigure[Steam (Spattack-L-D)]{\includegraphics[ width=\textwidth]{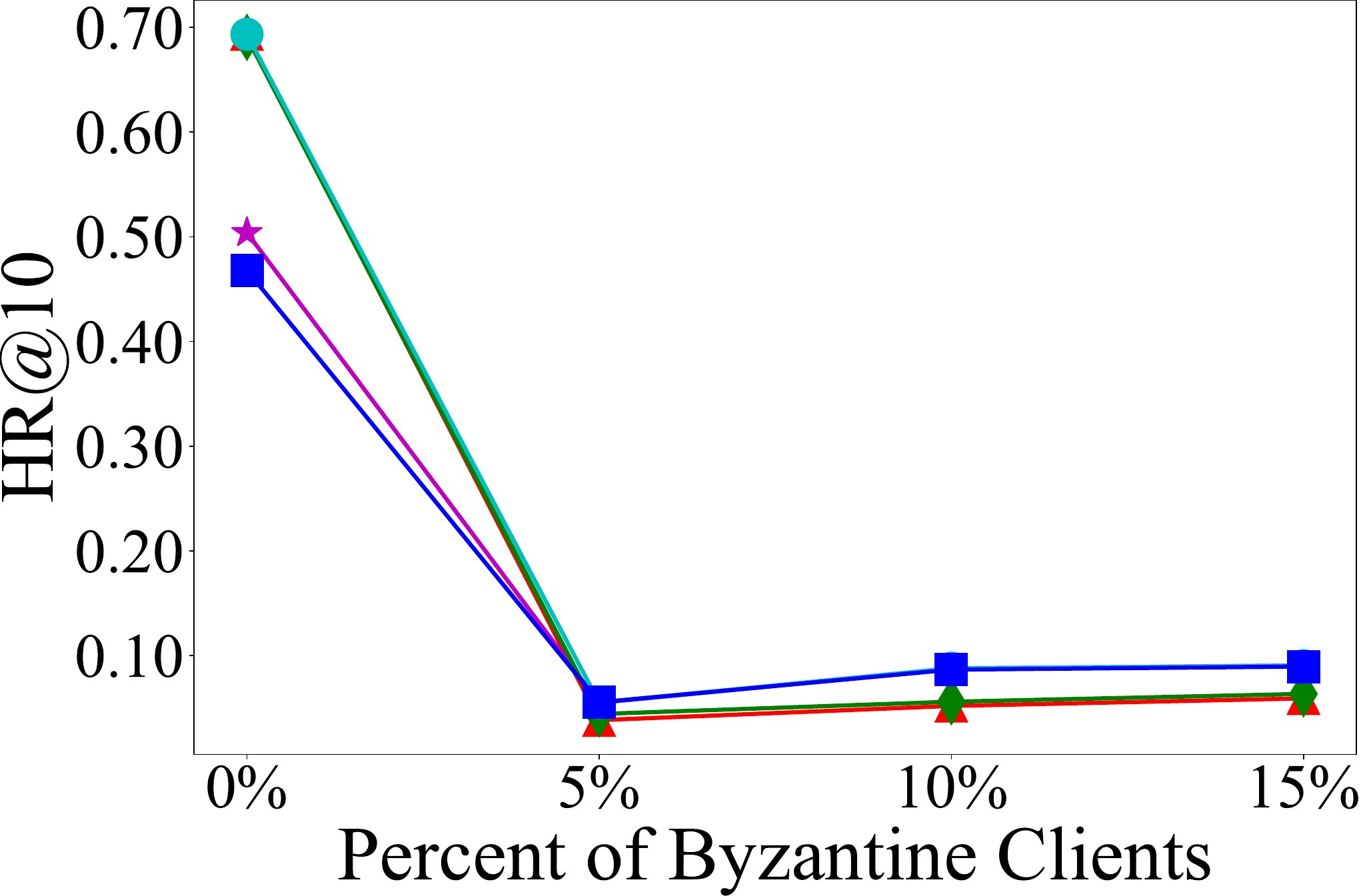}}
\endminipage\hfill
\minipage{0.23\textwidth}
\centering
 \subfigure[Steam (Spattack-L-S)]{ \includegraphics[ width=\textwidth]{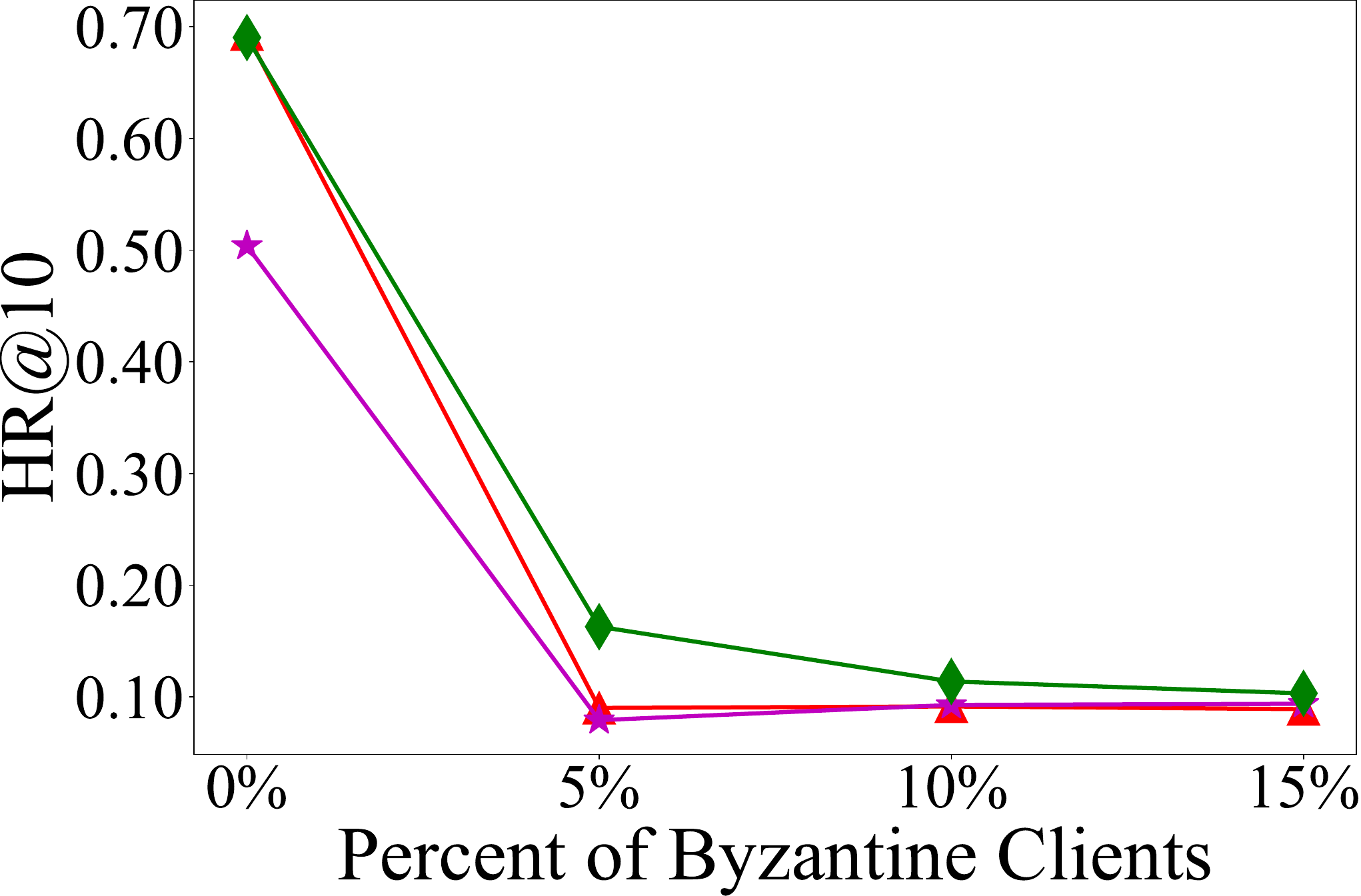}}
\endminipage\hfill

\centering
\minipage{0.23\textwidth}
\centering
 \subfigure[ML100K (Spattack-O-D)]{\includegraphics[ width=\textwidth]{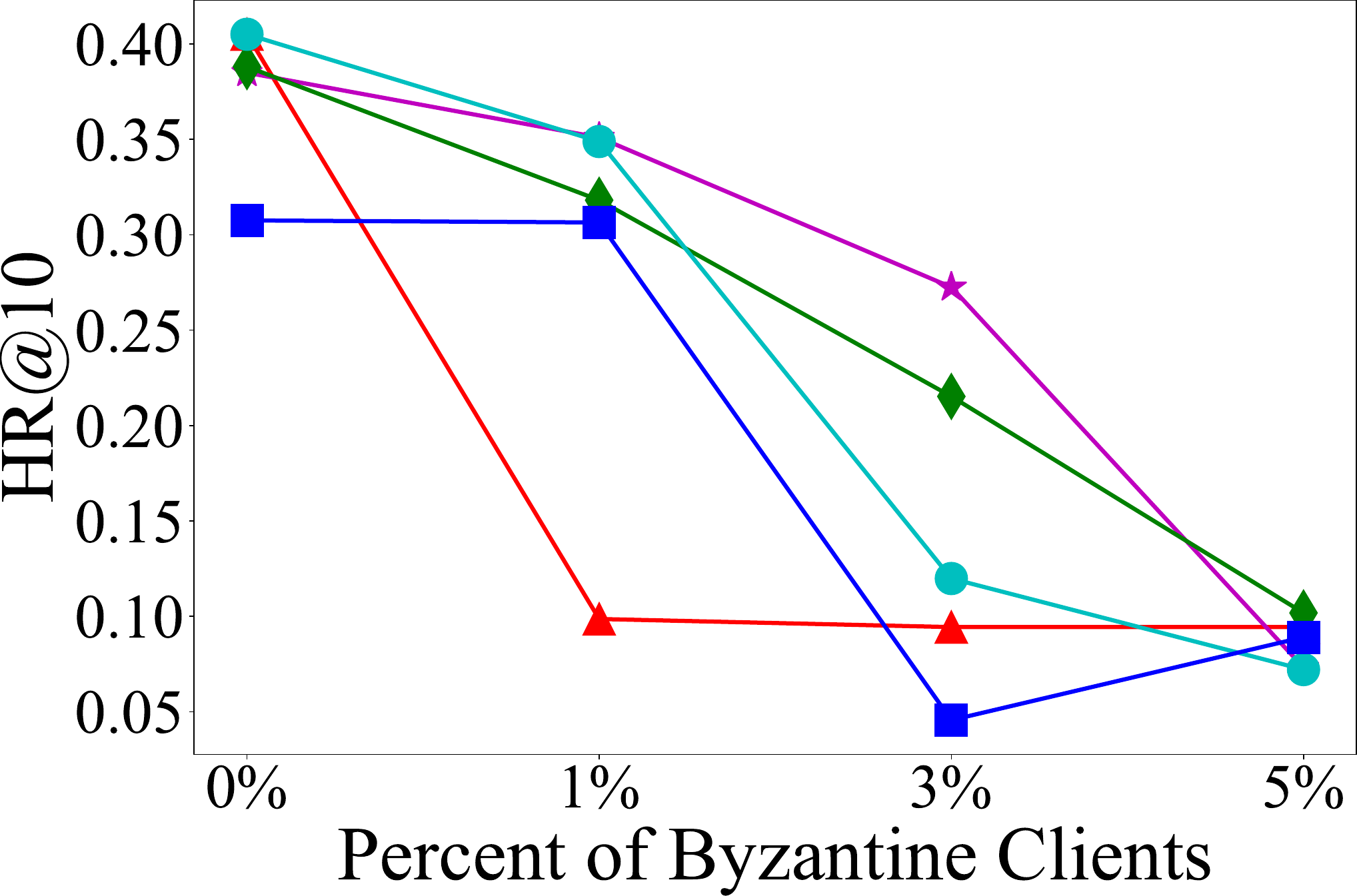}}
\endminipage\hfill
\minipage{0.23\textwidth}
\centering
 \subfigure[ML100K (Spattack-O-S)]{\includegraphics[ width=\textwidth]{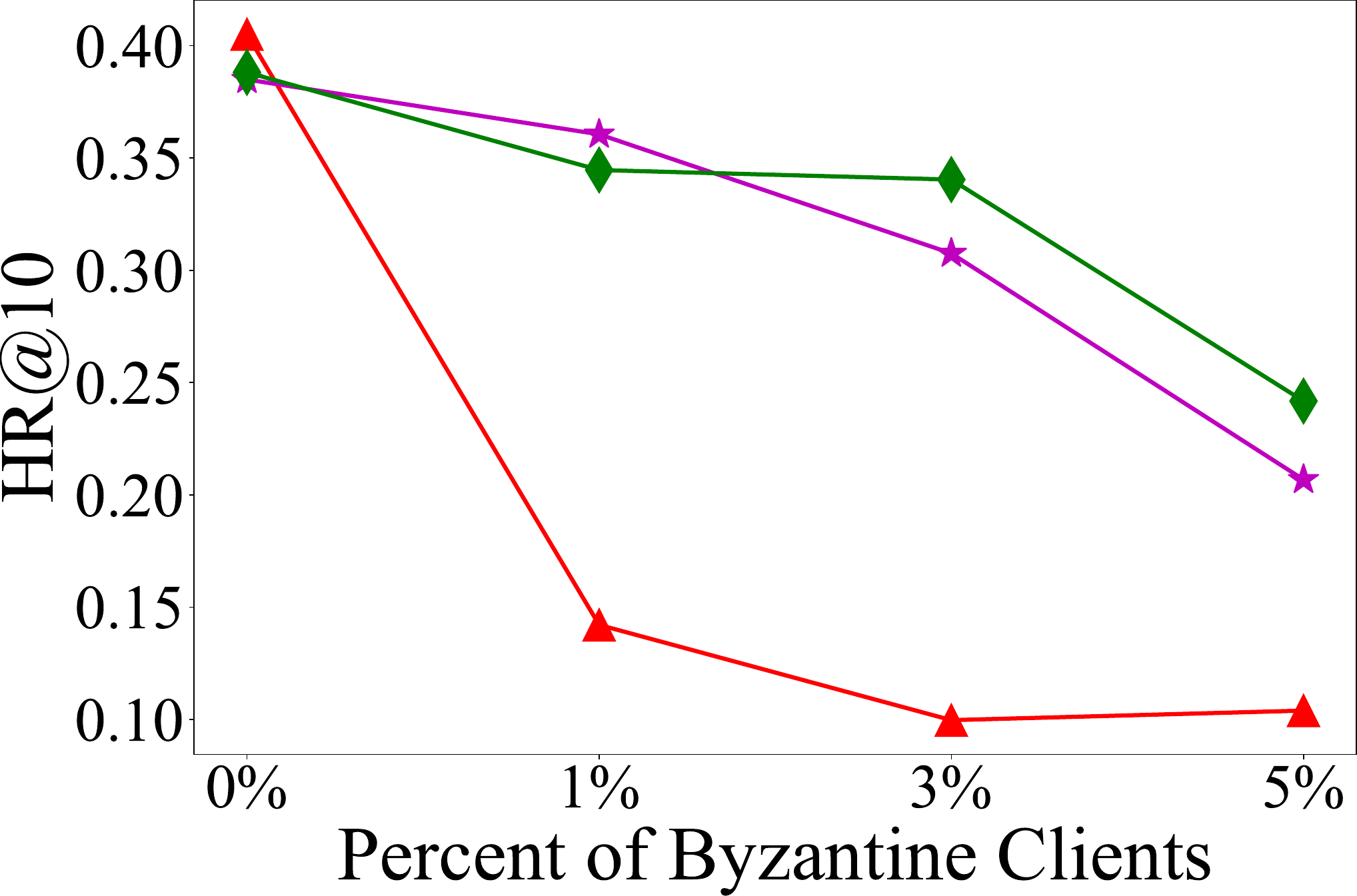}}
\endminipage\hfill
\minipage{0.23\textwidth}
\centering
 \subfigure[ML100K (Spattack-L-D)]{\includegraphics[ width=\textwidth]{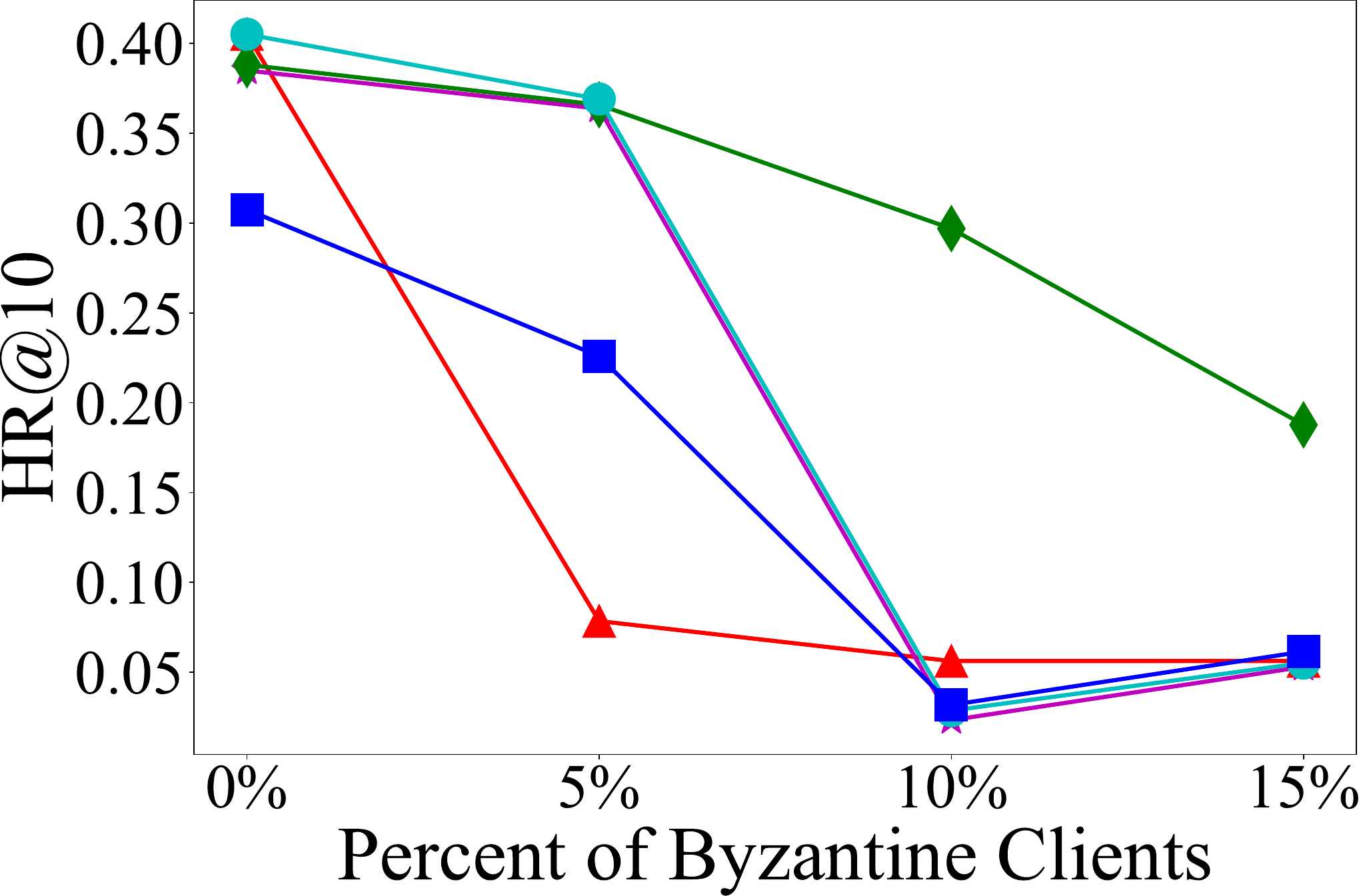}}
\endminipage\hfill
\minipage{0.23\textwidth}
\centering
 \subfigure[ML100K (Spattack-L-S)]{ \includegraphics[ width=\textwidth]{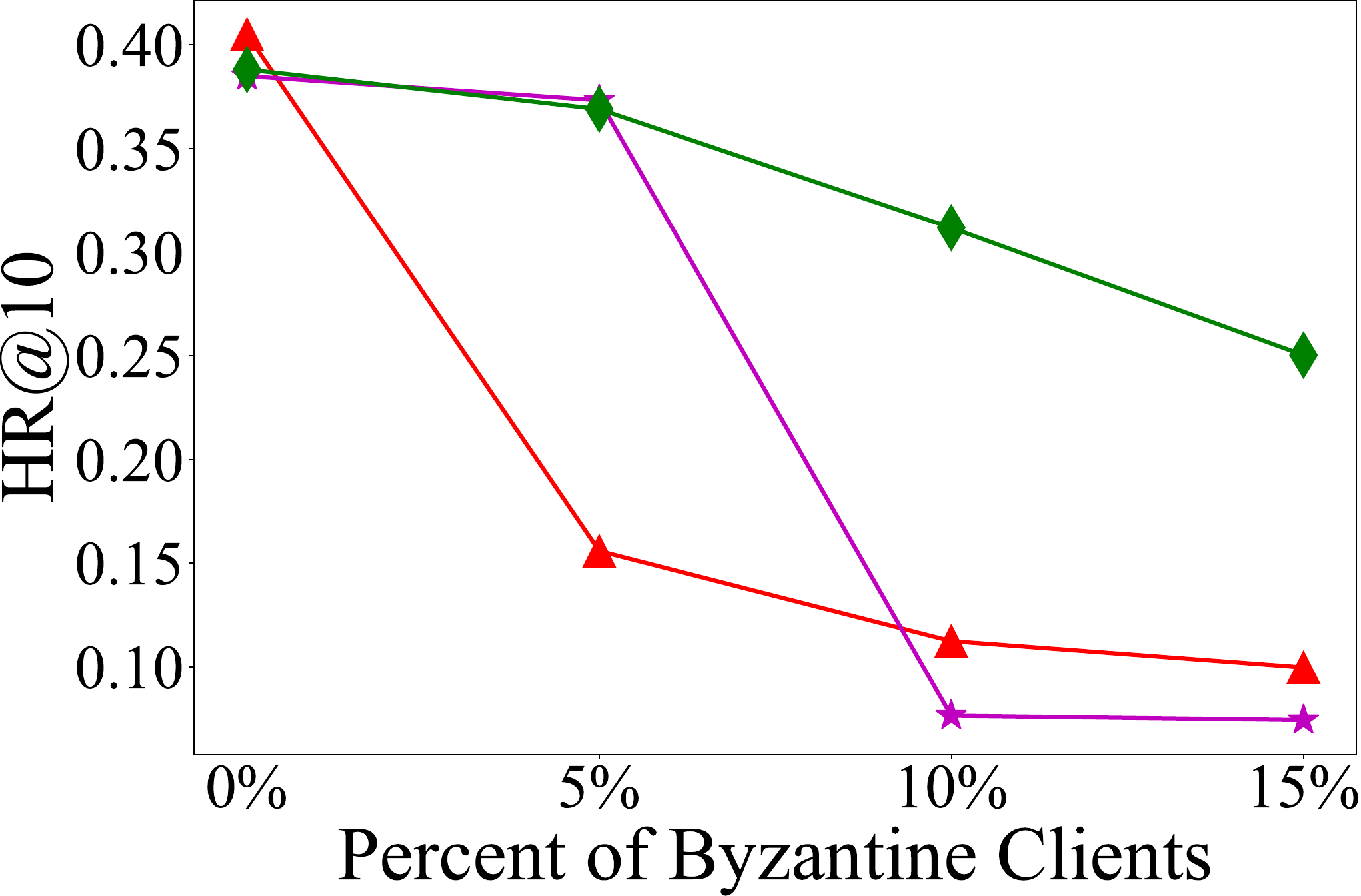}}
\endminipage\hfill

\centering
\minipage{0.23\textwidth}
\centering
 \subfigure[ML1M (Spattack-O-D)]{\includegraphics[ width=\textwidth]{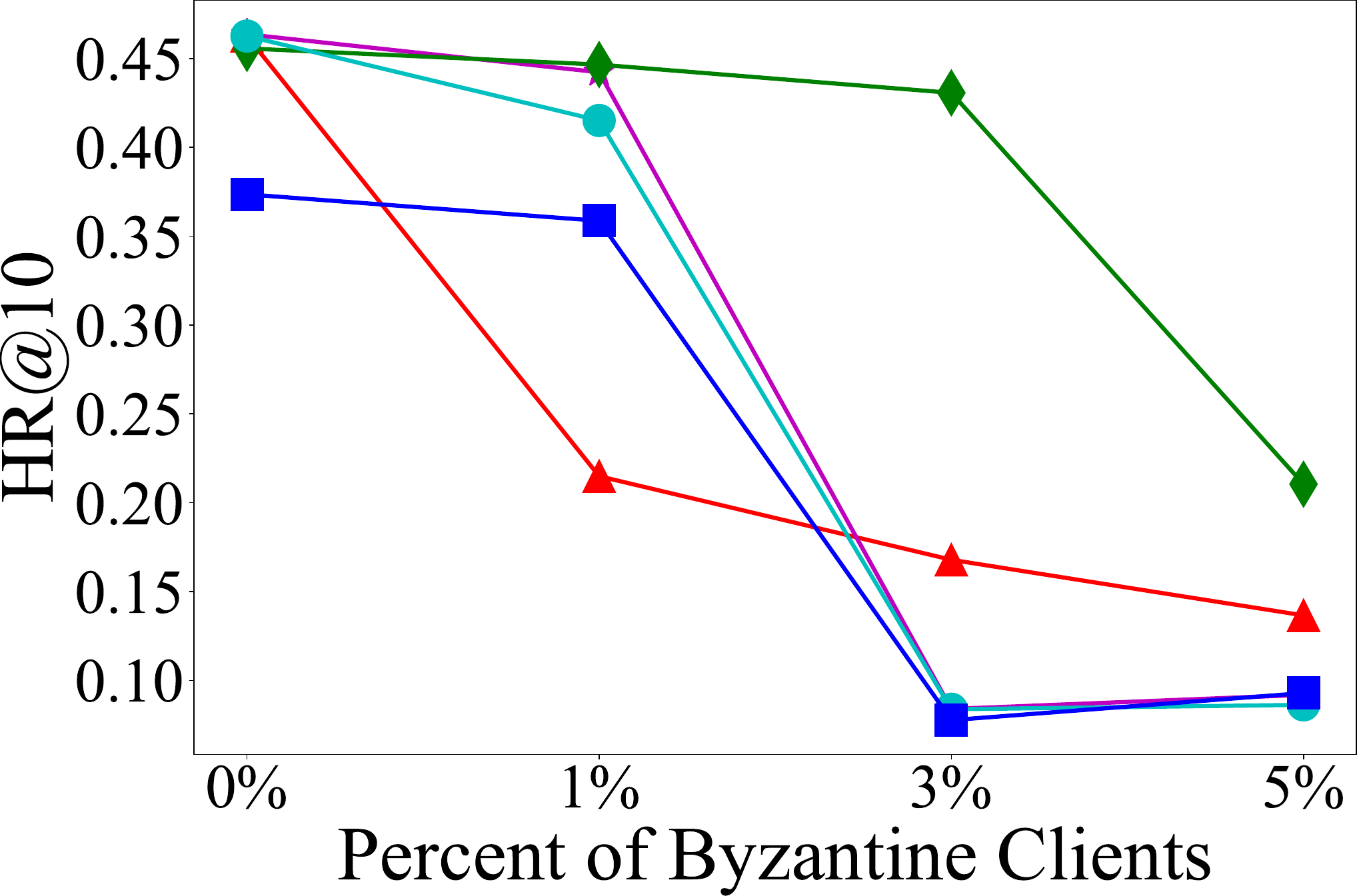}}
\endminipage\hfill
\minipage{0.23\textwidth}
\centering
 \subfigure[ML1M (Spattack-O-S)]{\includegraphics[ width=\textwidth]{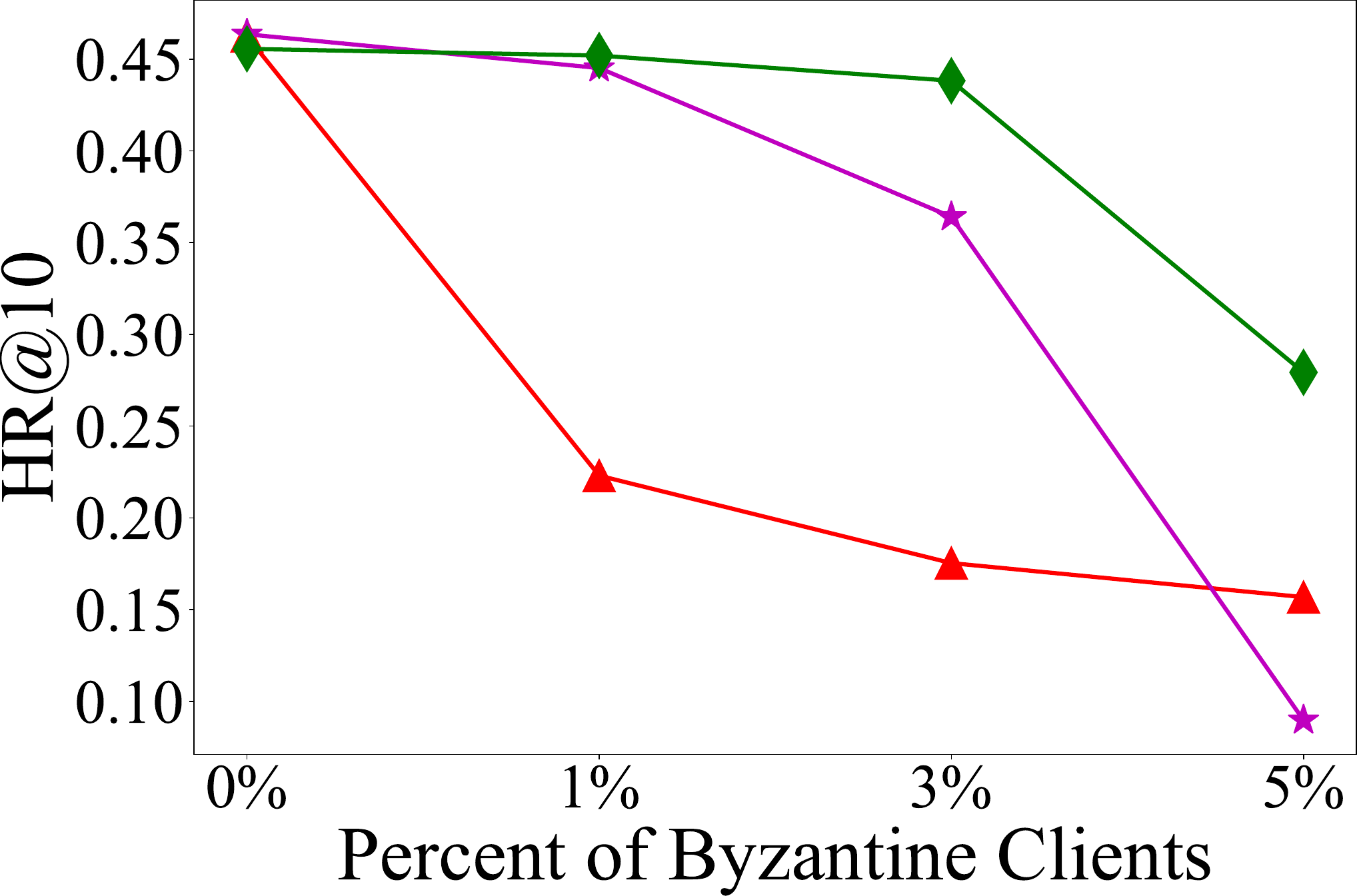}}
\endminipage\hfill
\minipage{0.23\textwidth}
\centering
 \subfigure[ML1M (Spattack-L-D)]{\includegraphics[ width=\textwidth]{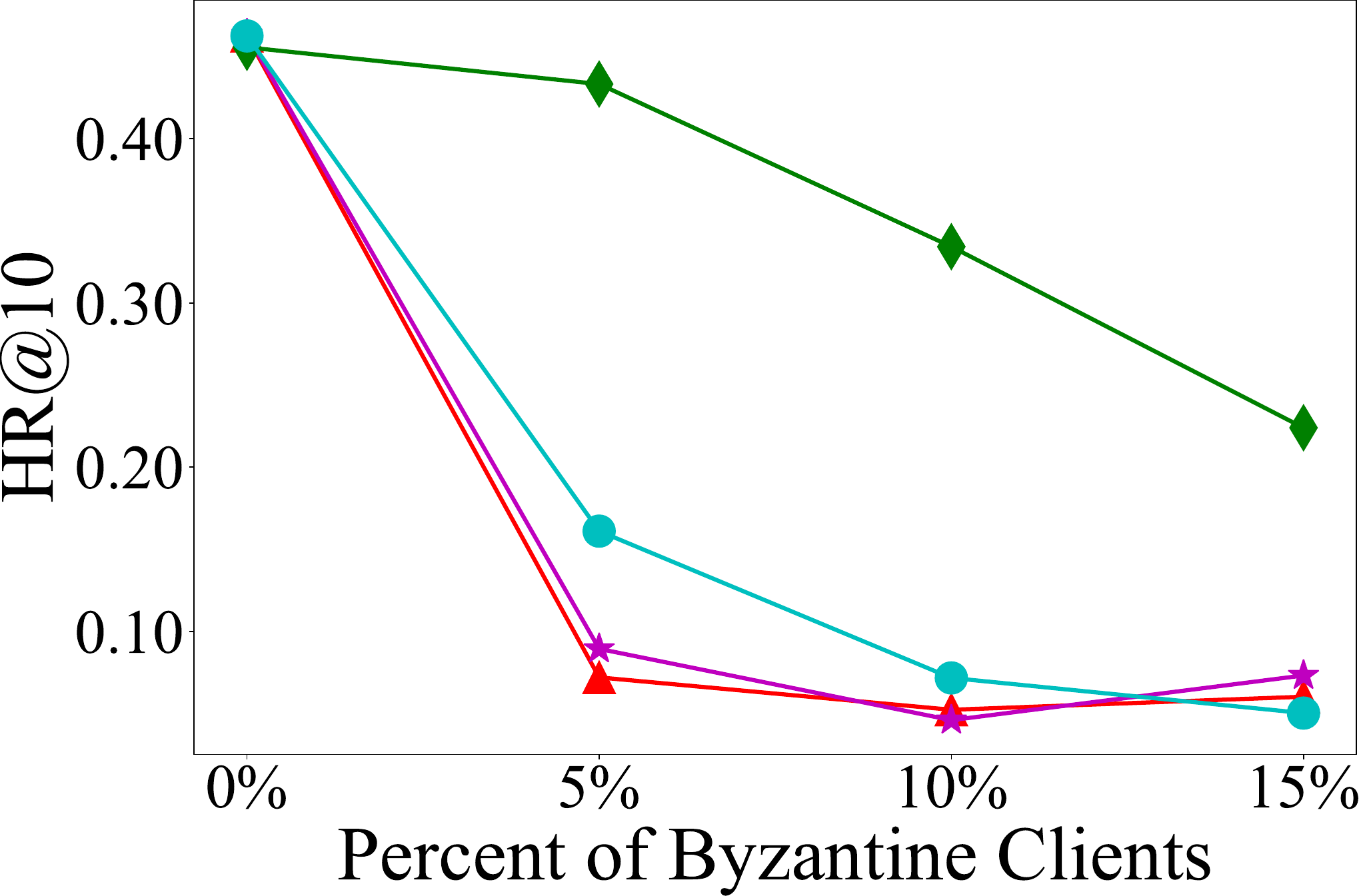}}
\endminipage\hfill
\minipage{0.23\textwidth}
\centering
 \subfigure[ML1M (Spattack-L-S)]{ \includegraphics[ width=\textwidth]{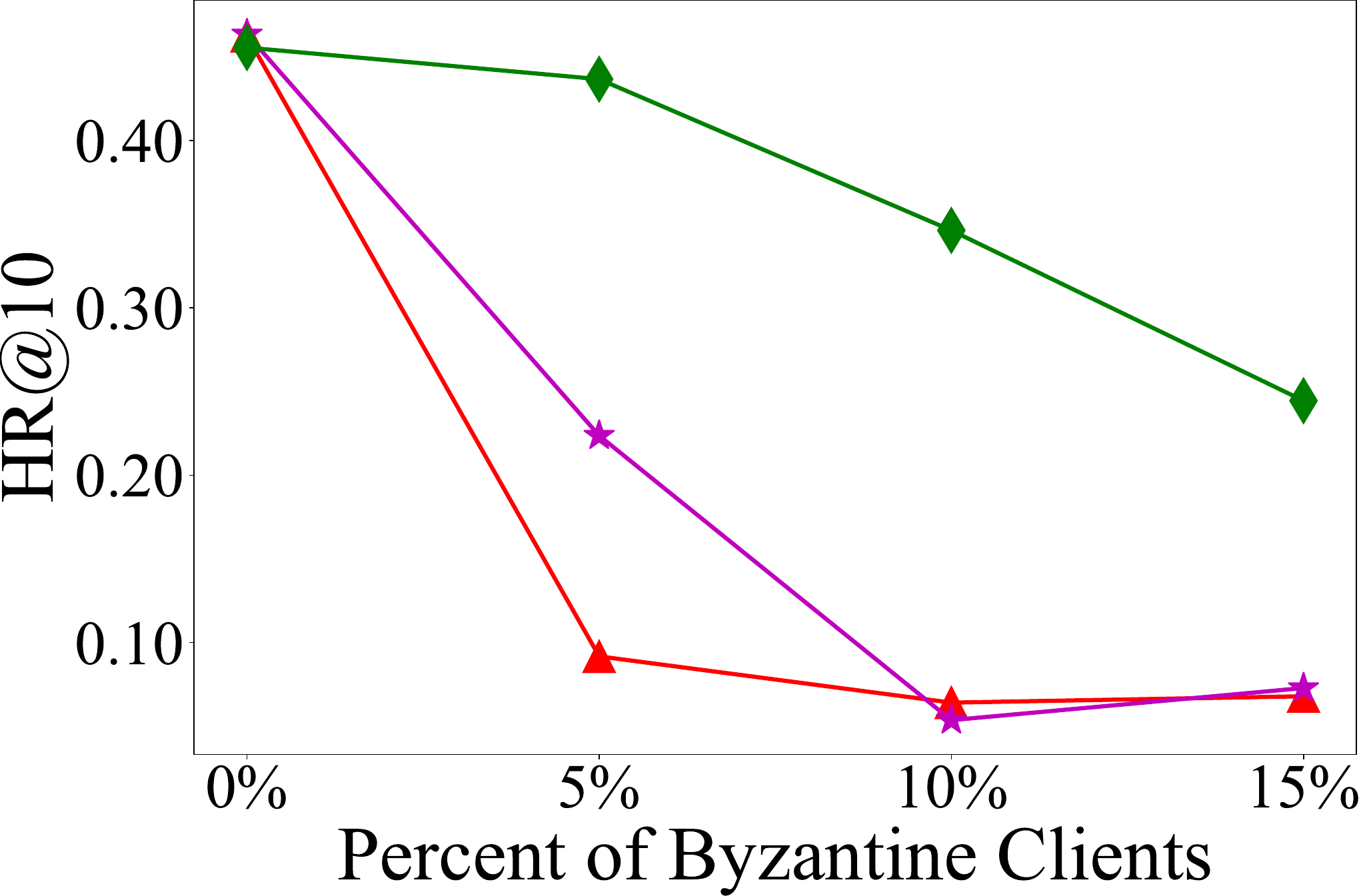}}
\endminipage\hfill
\vskip -0.15in
\caption{Performance of Spattack against multiple defense strategies under different ratios of Byzantine clients.}
\vskip -0.2in
\label{fig:full_defenses}
\end{figure*}

\subsection{Attack Performance Evaluation (RQ1)}
We compare the proposed Spattack against existing SOTA attack baselines under $3\%$ malicious ratio. The experimental results are reported in Tab.~\ref{tab:sota}, we find:\\
$\bullet$ Spattack can prevent FR convergence by controlling a few malicious clients. For example, Spattack can achieve a 47\%-98\% performance drop under $3\%$ malicious clients, demonstrating that FR is extremely vulnerable to Spattack.\\
$\bullet$ Spattack significantly outperforms other baselines. The explanation is that Spattack fully utilizes the sparse aggregation vulnerability by greedily breaking more items. Specifically, LabelFlip and FedAttack only indirectly manipulate the gradient by modifying data, while LIE, Cluster and Fang directly manipulate the gradients and thus can achieve higher attack impacts. Although Fang also perturbs in opposite directions of benign gradients, the malicious gradients are skewed to zero vector without considering the sparse aggregation of FR, leading to less effective attacks. \\
$\bullet$ The results on Steam overall drop more than ML100K and ML1M. A possible reason is that Steam involves fewer interactions on average (referring to the sparsity in Tab.~\ref{tab:dataset}), meaning there are more tailed items, which makes the model more susceptible to attacks.

\subsection{Attack Effectiveness under Defense (RQ2)}
We also evaluate the effectiveness of Spattack under different defenses. We set malicious ratio $\rho$ as 1\%, 3\% and 5\% for omniscient Spattack-O, and set the higher 5\%, 10\% and 15\% for the harder non-omniscient Spattack-L. We equip FR with Mean, Median and Norm aggregators for all attacks. Since TrimM and Krum assume the number of malicious updates is fixed for each item, but Spattack-O/L-S uploads different numbers of updates for each item, making them cannot be applied. 
As shown in Fig.~\ref{fig:full_defenses}, more results and analysis are in the Appendix. We have observations as follows: \\
$\bullet$ With only $5\%$ malicious clients, Spattack can dramatically degrade recommendation performance and even prevent convergence. The explanation is that different items have varied amounts of updates, and the defense of tailed items can be more easily broken than head items. \\
$\bullet$ When the attacker's knowledge and capability are limited, with the increasing malicious ratio $\rho$, the performance of defense FR consistently decreases even reaching an untrained model. The results also demonstrate that hiding the gradients of benign clients cannot protect FR, because the attacker can break the defense using only random noise.
\begin{figure}[tp]
\centering
\minipage{0.23\textwidth}
\centering
\includegraphics[width=\textwidth]{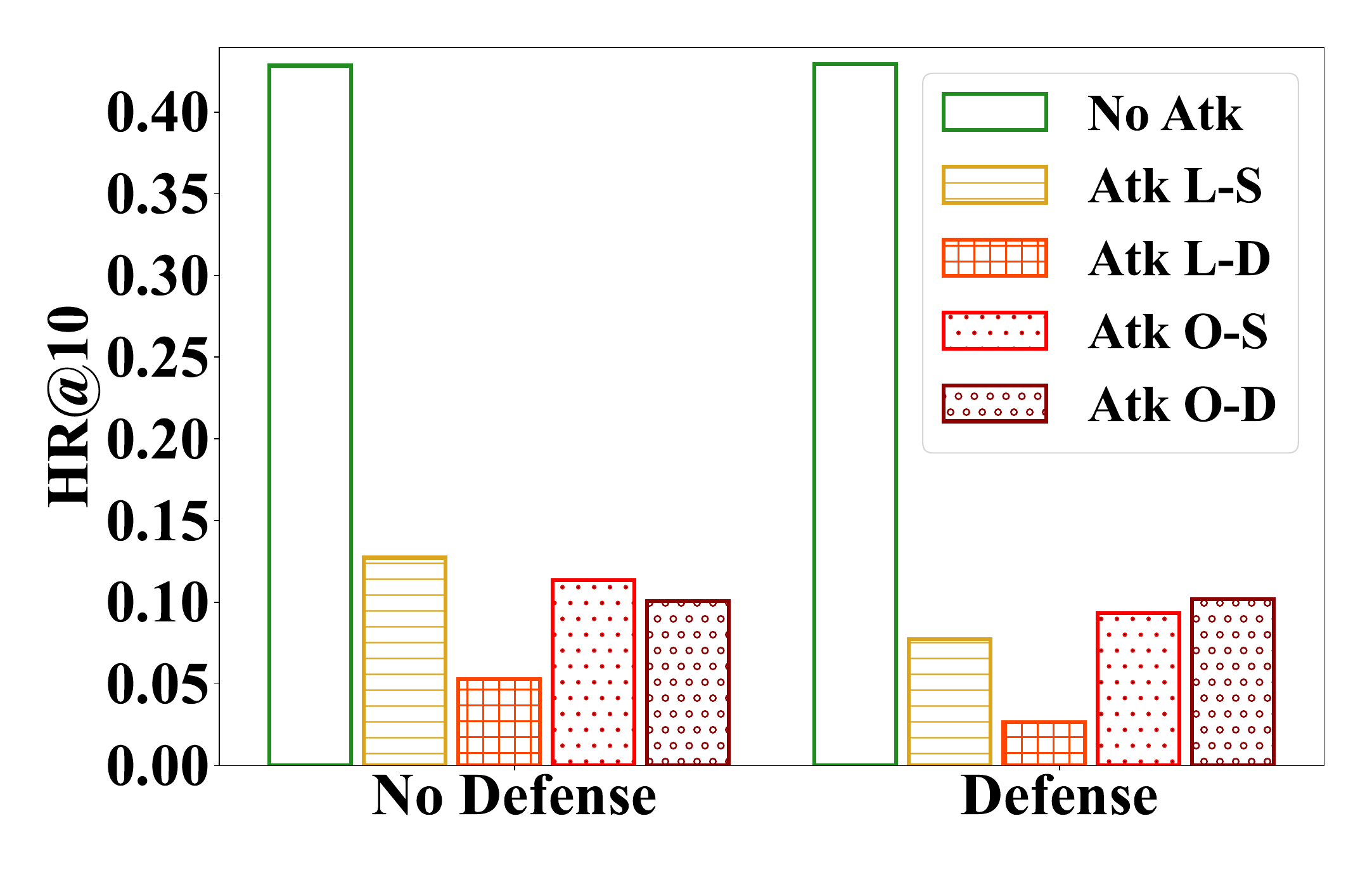}
\vspace{-22pt} 
\caption*{\small (a) ML100K}
\endminipage\hfill
\minipage{0.23\textwidth}
\centering
\includegraphics[width=\textwidth]{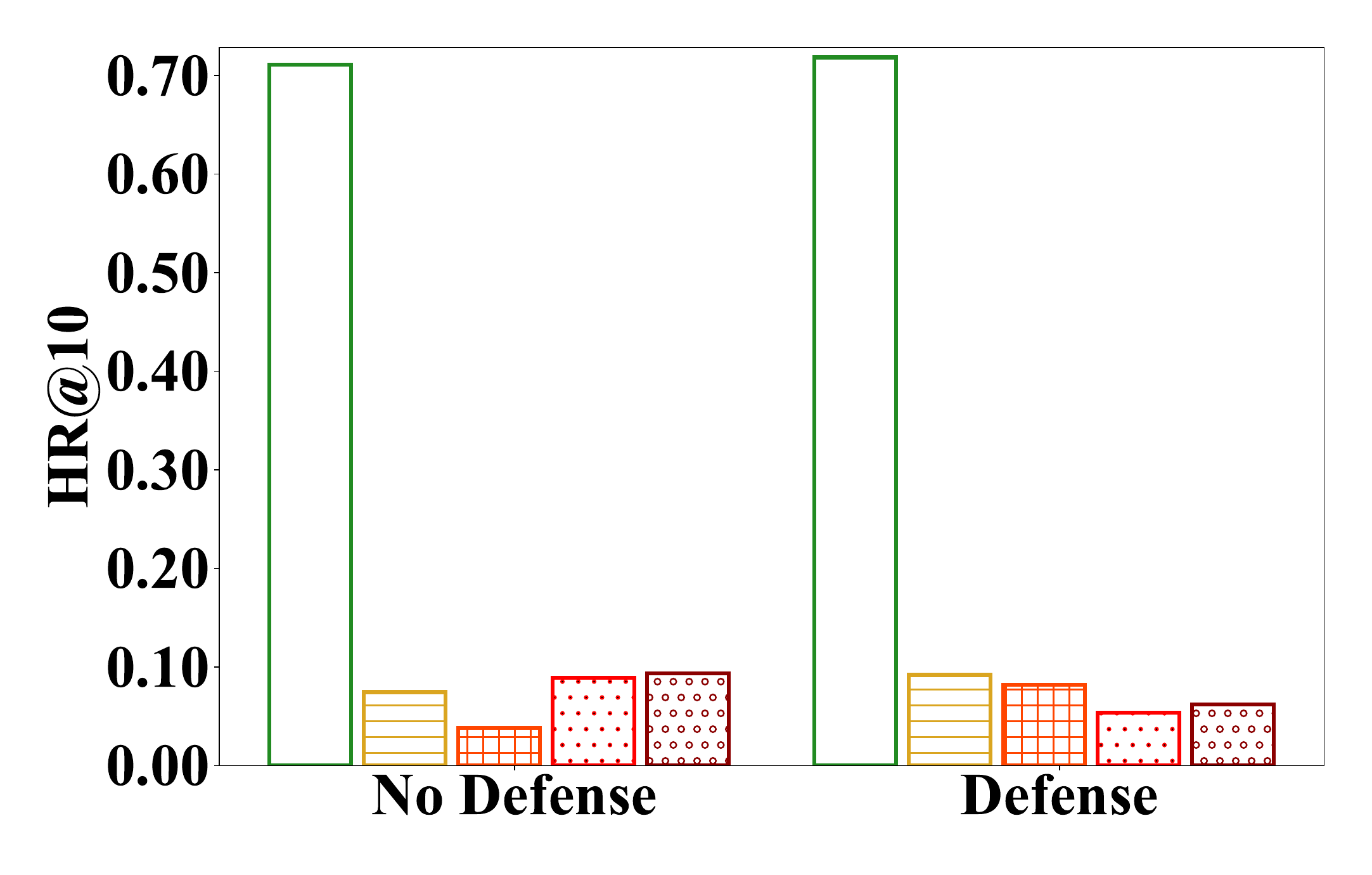}
\vspace{-22pt} 
\caption*{\small (b) Steam}
\endminipage\hfill
\vskip -0.05in
\caption{Attack performance on FedGNN model.}
\vskip -0.2in
\label{fig:fedgnn}
\end{figure}
\subsection{Transferability of Attack (RQ3)}
To demonstrate the generalizability of Spattack to other federated recommender systems, we perform Spattack on the SOTA FedGNN~\cite{Wu2021FedGNN} by extra uploading the malicious gradients of the GNN model. No Defense and Defense correspond to mean and median aggregators, respectively. The malicious ratio $\rho$ is set to 10\%. Please refer to the Appendix for more results and analysis. As shown in Fig.~\ref{fig:fedgnn}, we have the following observations: \\
$\bullet$ The performance of FedGNN dramatically drops under Spattack, demonstrating the common vulnerability of FedMF and FedGNN. Even though the parameters of GNN are densely aggregated, the attacker can still prevent convergence of model training by poisoning item embeddings. \\
$\bullet$ Spattack can achieve more effective attacks under defense. A possible reason is that for most items, i.e., tailed item, the malicious gradients can easily be the majority in its aggregation, so the Median AGR tends to pick the malicious gradient as output, while the poisoning in Mean AGR will be in remission by averaging malicious and benign gradients.

\subsection{Hyperparameter Analysis (RQ4)}
Lastly, we investigate the impact of the hyper-parameter on Spattack. In Fig.~\ref{fig:startEpoch}, we show the convergence of FR under mean aggregators on Steam, where the malicious ratio is set to 10\%, and the results correspond to starting attacks at 0, 20, 40 and 60. Please refer to the Appendix for more results and analysis. We have the following observations: \\
$\bullet$ Spattack with a small starting epoch tends to have better attack performance because the model has converged under a large starting epoch. \\
$\bullet$ When Spattack is launched, Spattack-O prevents the model from continuing to converge, while Spattack-L causes the performance dramatically drops. The reason is that Spattack-O uploads malicious gradients with an average equal to the negative of the benign gradients' average, resulting in zero gradients after aggregation.
\begin{figure}[tp]
\centering
\minipage{0.23\textwidth}
\centering
\includegraphics[width=\textwidth]{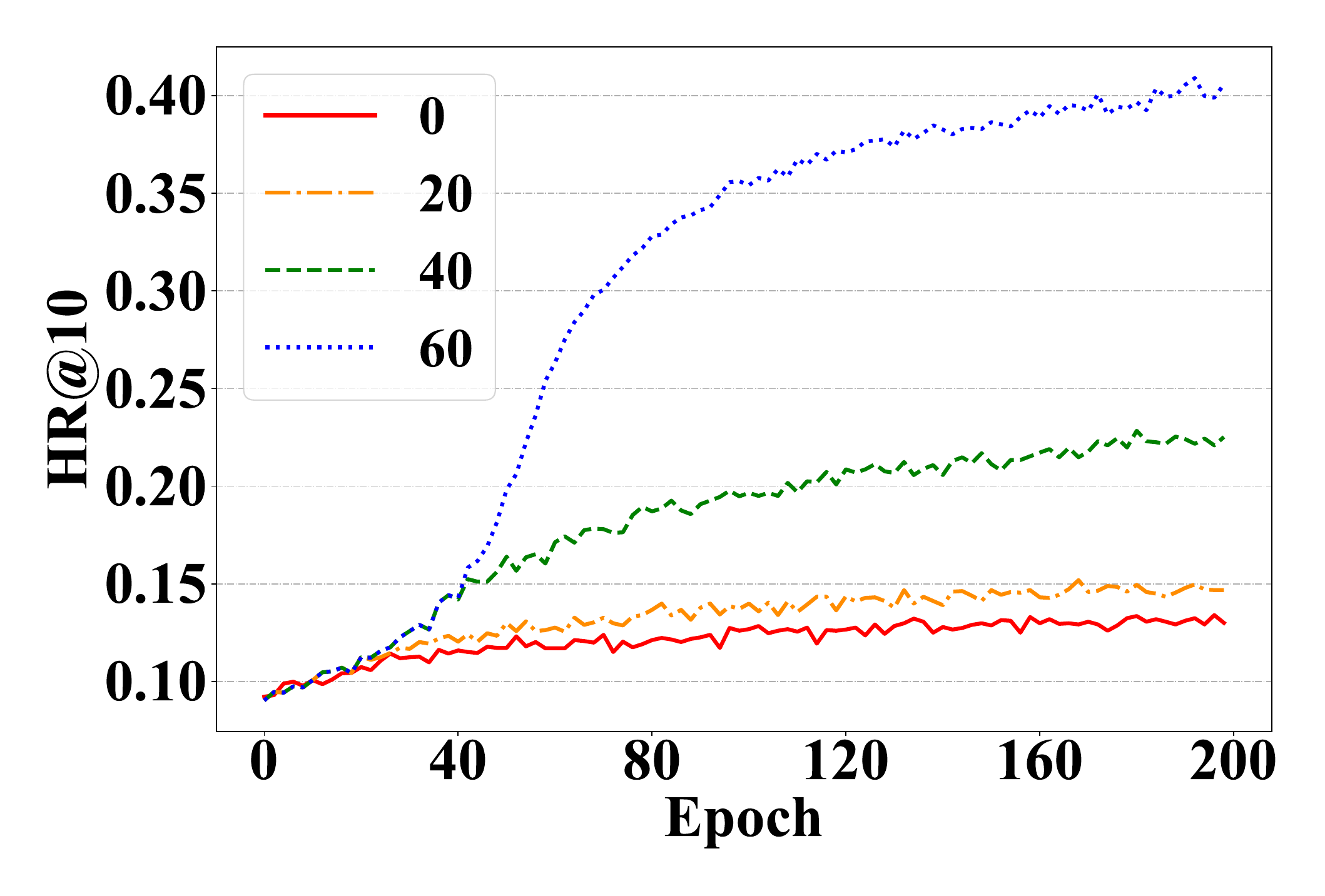}
\vspace{-25pt} 
\caption*{\small (a) Spattack-O-D}
\endminipage\hfill
\minipage{0.23\textwidth}
\centering
\includegraphics[width=\textwidth]{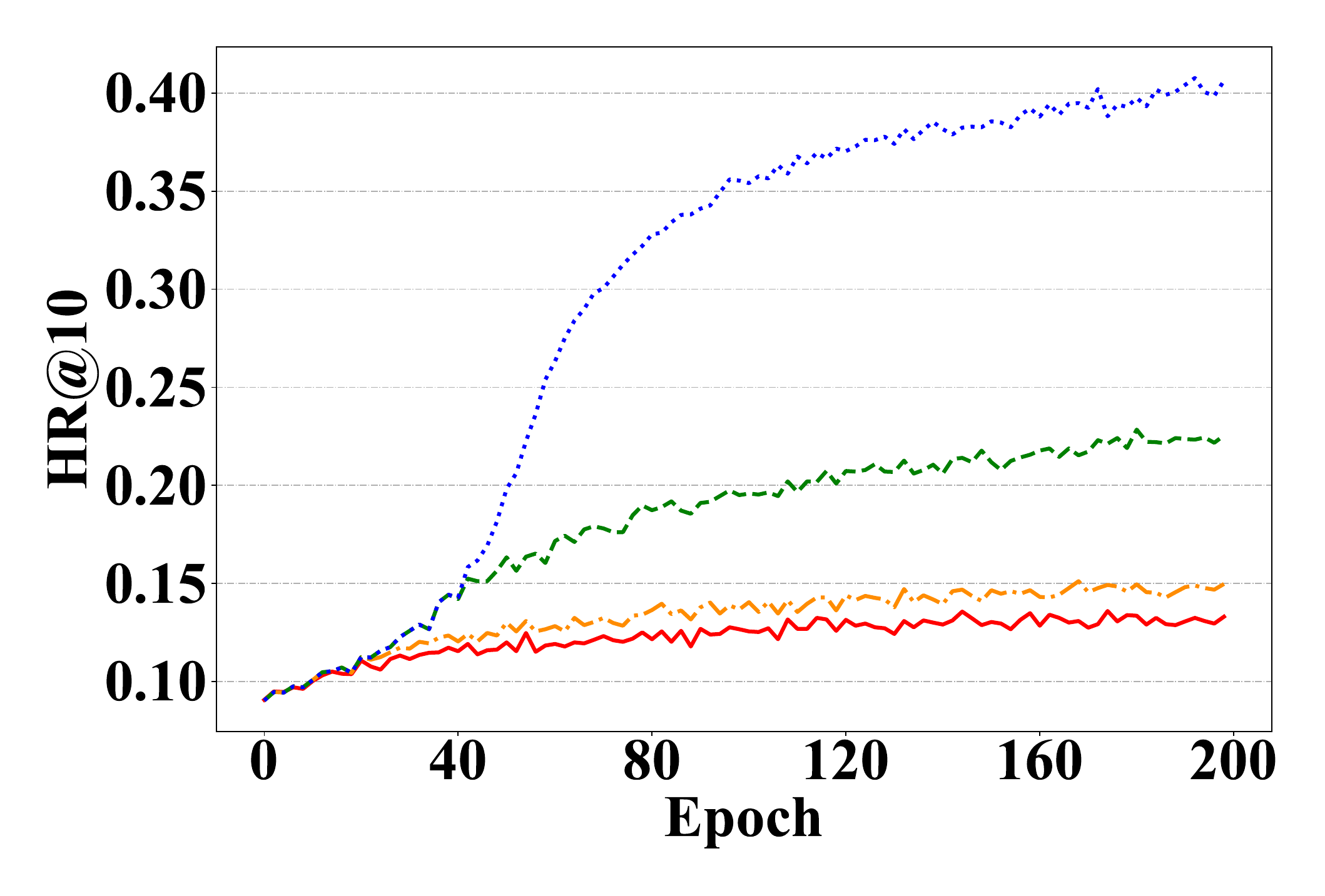}
\vspace{-25pt} 
\caption*{\small (b) Spattack-O-S}
\endminipage\hfill
\centering
\minipage{0.23\textwidth}
\centering
\includegraphics[width=\textwidth]{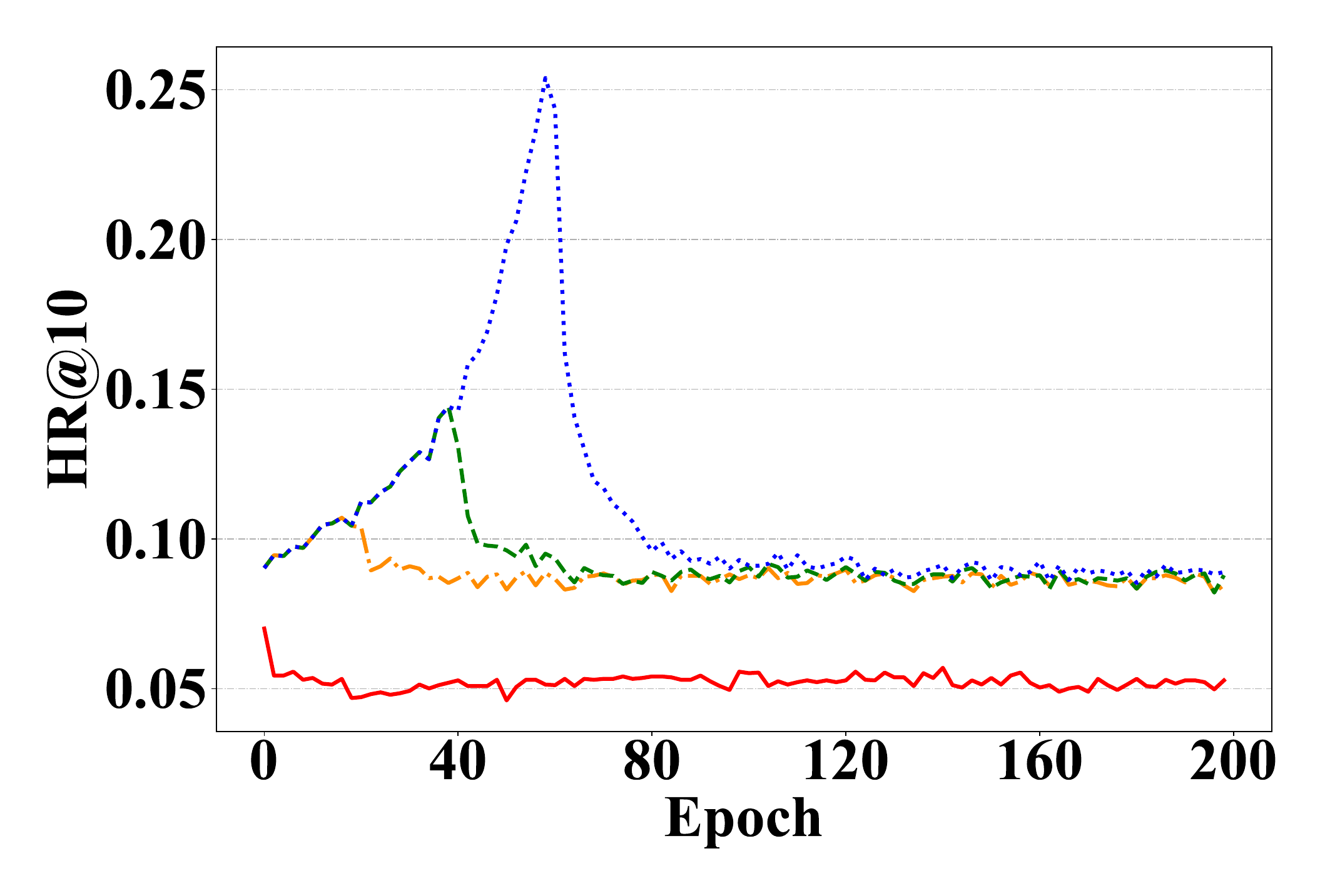}
\vspace{-25pt} 
\caption*{\small (c) Spattack-L-D}
\endminipage\hfill
\minipage{0.23\textwidth}
\centering
\includegraphics[width=\textwidth]{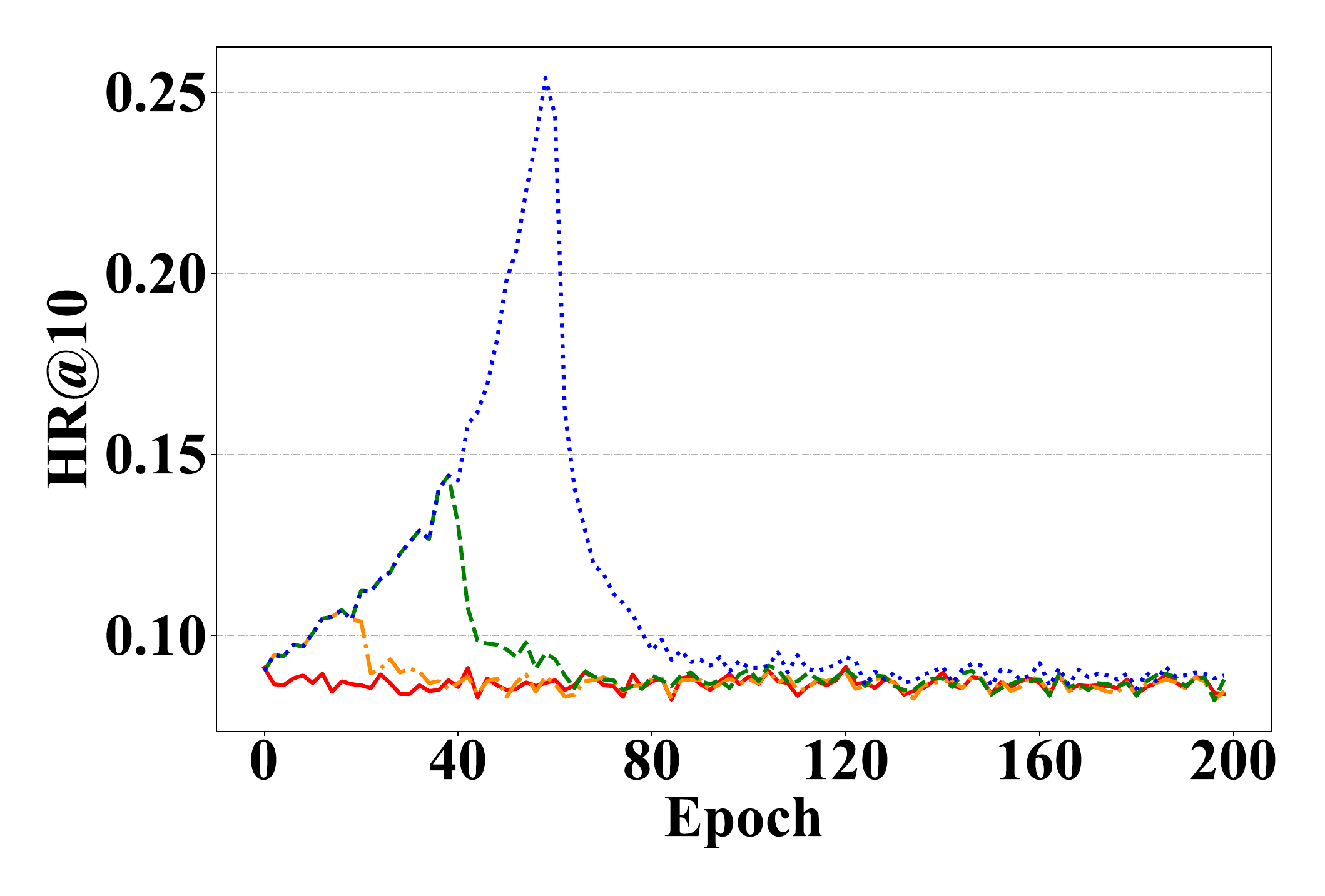}
\vspace{-25pt} 
\caption*{\small (d) Spattack-L-S}
\endminipage\hfill
\vskip -0.1in
\caption{Attack performance on different start epochs.}
\vskip -0.2in
\label{fig:startEpoch}
\end{figure}

\section{Conclusion}
In this paper, we first systematically study the Byzantine robustness of federated recommender from the perspective of sparse aggregation, where the item embedding in FR can only be updated by partial clients, instead of full clients (dense aggregation in general FL). Then we design a series of attack strategies, called Spattack, based on the vulnerability from sparse aggregation in FR. 
Our Spattack can be employed by attackers with different levels of knowledge and capability. Extensive experimental results demonstrated that FR is extremely fragile to Spattack. In the future, we aim to design a more robust aggregator in FR from the perspective of sparse aggregation, which focuses on the robust aggregation for tailed items.

\section*{Acknowledgments}
This work is supported by the National Key Research and Development Program of China (2023YFF0725103), the National Natural Science Foundation of China (U22B2038, 62322203, 62172052, 62192784), and 
the Young Elite Scientists Sponsorship Program (No.2023QNRC001) by CAST.
\bibliography{reference}
\clearpage
\section{Related Work}
\textbf{Federated Recommender.}
Federated learning aims to collaboratively train a shared model based on the distributed data in a privacy-preserving manner~\cite{yan2024federated, McMahan2016CommunicationEfficientLO, zhang2023lightfr}. Accordingly, a federated recommender ensures individual users' historical data is locally stored and only uploads intermediate data to the server for collaborative training. In this process, the user's rating behaviors (the set of interacted items or rating scores) are private information. \citet{MuhammadAmmaduddin2019FederatedCF} first proposed a federated collaborative filter framework for the privacy-preserving recommendation. The user embeddings are stored and updated locally while the gradients of item embeddings are uploaded to the server for aggregation. Moreover, for better privacy, \cite{Chai2019SecureFM} applied homomorphic encryption; \cite{Ying2020SharedMA} further improved the efficiency by utilizing secret sharing instead of homomorphic encryption. Recently, federated recommendations based on graph neural networks have emerged~\cite{Wu2021FedGNN, Liu2021FederatedSR, Luo2022PersonalizedFR, yan2024federated2, zhang2023minimum}, further incorporating high-order user-item interactions into local training data. Overall, most existing federated recommender systems follow the paradigm where the gradients of item embeddings are uploaded to the server for aggregation. So they are all sparse aggregations where each item embedding can only be updated by partial clients, leading to varied robustness of each item. This paper sheds the first light on this unique vulnerability in the view of sparse aggregation.

\textbf{Robustness of Federated Learning.}
The security of federated learning has drawn increasing attention in recent years~\cite{lyu2022privacy,lyu2020threats, zhang2024can, zhang2022robust, zhang2024endowing, wu2024unlearning, wu2024number}, and a large number of attacks against FL have been proposed. Among the attack strategies, the Byzantine attack is one of the most popular attacks~\cite{RodriguezBarroso2022SurveyOF}. The classic mean aggregator can be easily skewed by arbitrary updates from Byzantine clients. Therefore, a lot of Byzantine defense based on statistics has been proposed in recent year~\cite{DongYin2018ByzantineRobustDL, PevaBlanchard2017MachineLW, Mhamdi2018TheHV, Xu2021ByzantinerobustFL, RodriguezBarroso2022SurveyOF, Pillutla2019RobustAF, Wu2019FederatedVS, Fu2019AttackResistantFL, chen2020robust}, which aimed to filter the Byzantine updates and guarantee the convergence of federated learning. Although the Byzantine robustness problem is well-studied in FL, existing Byzantine attacks and defenses of FL are defined based on the dense aggregation and cannot apply to our sparse aggregation in FR. More recently, a few attacks have been proposed against federated recommender, among which, \cite{Rong2022FedRecAttackMP, Zhang2021PipAttackPF} focus on targeted attacks aiming to promote target items by increasing their exposure chances, and~\cite{Wu2022FedAttackEA} employs improper positive/negative samples to manipulate model parameters indirectly. \cite{Yu2022UntargetedAA} uploads poisonous gradients that collapse all item embeddings to several clusters to confuse different items. However, they all neglect the unique sparsity in federated recommender.

\section{Notations}
We present all notations relevant to our paper in Tab.~\ref{tab:notation}.
\begin{table}[ht]
    \centering
    \vskip -0.1in
    \caption{Notations}
    \vskip -0.125in
    \setlength{\extrarowheight}{1pt}
    \scalebox{0.9}{
    \begin{tabular}{c|p{0.9\linewidth}}
    \hline 
    $\mathcal{D}$ & all interactions \\
    $\mathcal{D}_i$ & local interaction of user $u_i$  \\
    $\mathcal{U}$ & set of $n$ benign users\\
    $\tilde{\mathcal{U}}$ & set of $\tilde{n}$ malicious users\\
    $\mathcal{V}$ & set of $m$ items\\
    $\mathcal{V}_{u_i}$ & set of items interacting with $u_{i}$ \\
    $\mathcal{U}_{v_j}$ & set of users interacting with $v_j$ \\
    \hline
    $\boldsymbol{U}$ & the user embeddings where $\boldsymbol{U}=\{\boldsymbol{u}_1,...,\boldsymbol{u}_n\}$\\
    $\boldsymbol{V}$ & the item embeddings where $\boldsymbol{V}=\{\boldsymbol{v}_1,...,\boldsymbol{v}_m\}$\\
    $\nabla \boldsymbol{V}^{i,t}$ & the embedding gradient of $\mathcal{V}$ from benign user $u_i$ at epoch $t$\\
    $\nabla \tilde{\boldsymbol{V}}^{i,t}$ & the embedding gradient of $\mathcal{V}$ from malicious user $\tilde{u}_i$ at epoch $t$\\
    $\nabla \boldsymbol{v}^{i,t}_j$ & the embedding gradient of $v_j$ from benign user $u_i$ at epoch $t$\\
    $\nabla \tilde{\boldsymbol{v}}^{i,t}_j$ & the embedding gradient of $v_j$ from malicious user $\tilde{u}_i$ at epoch $t$\\
    $\boldsymbol{\Theta}$ & the parameters of neural network model\\
    \hline
    $[n]$ & Set of integers $\{1,\cdots,n\}$\\
    $\eta$ & learning rate\\
    $\rho$ & the proportion of malicious users \\
    $\alpha$ & the breaking point of statistically robust aggregator\\
    \hline
    \end{tabular}
    }
    \label{tab:notation}
    \vskip -0.1in
\end{table}

\section{Detailed Proof for Proposition~\ref{prop2}}
\begin{proof}
Consider the robust aggregator with breaking point $\alpha$, the model convergence can be guaranteed when the number of malicious clients meets $\frac{\tilde{n}}{n+\tilde{n}}<\alpha$ in FL. While in FR,  given an item $v$ with degree $d_v$, the aggregator only collects $d_v$ benign gradients in sparse aggregation. In Byzantine attacks, all the malicious clients can easily collude to send consistent malicious gradients to certain item $v$. Once $\tilde{n}/(d_v + \tilde{n}) > \alpha$, namely $v$'s degree meets $d_v < (\tilde{n} - \alpha \tilde{n})/\alpha$, the malicious updates will ultimately dominate and hence the statically robust aggregator will be tricked to pick the malicious updates. 
Let $p(x) = Cx^{-\beta}$ be the power-law distribution of items' degrees. Its cumulative distribution function $P(x)$ is defined as the probability that the quantity of the item's degree is larger than $x$:
\begin{align}
    P(x) &= \operatorname{Pr}(X>x)=C \int_{x}^{+\infty} p(X) dX \notag \\
         &= C \int_{x}^{+\infty} X^{-\beta} dX = \frac{C x^{(1-\beta)}}{\beta-1} .
\end{align}
So the probability that item's degree is smaller than $(\tilde{n} - \alpha\tilde{n})/\alpha$ can be calculated as following:
\begin{equation}
        1-P(X>\frac{1-\alpha}{\alpha} \tilde{n}) = 1-
        \frac{C}{\beta-1} (\frac{1-\alpha}{\alpha} \tilde{n})^{(1-\beta)}.
\end{equation}
\end{proof}

\section{Reproducibility Supplement}\label{app:Reproducibility} 
\subsection{Dataset and Evaluation Metric}
Following~\cite{Rong2022FedRecAttackMP}, Spattack is evaluated on three widely used datasets, including a movie recommendation dataset {MovieLens-1M (ML1M)}~\cite{Harper2016TheMD}, the small version {MovieLens-100K (ML100K)}, and game recommendation dataset {Steam-200K (Steam)}~\cite{Cheuque2019RecommenderSF}. For all datasets, we closely follow the dataset configurations in previous works~\cite{He2017NeuralCF, Rong2022FedRecAttackMP} by unifying interactions as implicit feedback and removing duplicate interactions. To evaluate the ranking quality of the test item recommendation, we adopt two common evaluation protocols, i.e., {hit ratio (HR)} and normalized discounted cumulative gain at rank $K$ (nDCG@$K$). 

Here, we use $K$ as 5 and 10.  For each user, since ranking the test item among all items is time-consuming, following the widely-used strategy~\cite{Elkahky2015AMD,He2017NeuralCF}, we randomly sample 100 items from the items that have not interacted with the user, then rank the test item among the 100 items. The HR@$K$ indicates whether the test item is ranked in the top-$K$, and the nNDCG@$K$ takes position significance into account. The higher HR@$K$ and nNDCG@$K$ indicate better recommendation performance. Note that these metrics are only calculated on benign clients. 

\subsection{Federated Recommender Systems} 
We use FedMF~\cite{Rong2022FedRecAttackMP} as our target model's architecture in evaluation. We outline the details for reproducibility below.
Following~\cite{Rong2022FedRecAttackMP}, with BPR loss, the FedMF updates user embeddings locally and optimizes item embeddings on the server. We set the unit size of embeddings to 32.
To evaluate the generalization ability of Spattack, we also conduct attacks on the state-of-the-art FedGNN~\cite{Wu2021FedGNN}, which can collaboratively train GNN models meanwhile exploiting high-order user-item interaction information with privacy well protected.
We use BPR loss to train a 2-layer FedGNN model, where both item embeddings and GNN parameters are globally optimized. For the input user and item embedding, the dimension is set to 32. For the hidden layer, we set the hidden unit size to 64. Stochastic gradient descent is selected as the default optimization algorithm, and its learning rate is 0.01.

\subsection{Baselines}
$\bullet$ LabelFlip~\cite{ValeTolpegin2020DataPA} poisons data by flipping the training labels of malicious clients and does not require knowledge of the training data distribution. Each malicious client uses positive samples as negative samples and uses negative samples as positive samples.\\
$\bullet$ FedAttack~\cite{Wu2022FedAttackEA} conducts data poisoning by employing improper positive/negative samples. Each malicious client selects the items that are most similar to the user's interest as negative samples while regarding items that are most dissimilar to the user's interest as positive samples.\\
$\bullet$ Gaussian~\cite{Fang2019LocalMP} estimates a Gaussian distribution of benign gradients and then uploads samples from it.\\
$\bullet$ LIE~\cite{GiladBaruch2019ALI} adds small amounts of noise towards the average of benign gradients. We assign 0.1 as the scaling factor that affects the standard deviation of model parameters.\\
$\bullet$ Fang~\cite{Fang2019LocalMP} adds noise to opposite directions of the average normal gradient. We select the attack scaling factor from randomly uniform samples from [3, 4].\\
$\bullet$ Cluster~\cite{Yu2022UntargetedAA} uploads malicious gradients that aim to make item embeddings collapse into several dense clusters. We set the initial number of clusters as 1. The range of the number of clusters and the threshold is set to [1, 10].
\begin{table}[t]
\centering
\caption{The number of malicious clients under different attack ratios $\rho$.}\label{tab:attacker_num}
\vskip -0.15in
\begin{tabular}{l|c|c|c|c|c}
\hline
\centering
Dataset & 1\% & 3\% & 5\% & 10\% & 15\% \\ \hline
ML100K & 9  & 29  & 49  & 105 & 166 \\ 
ML1M   & 61 & 186 & 317 & 671 & 1066 \\ 
Steam   & 37 & 116 & 197 & 417 & 662 \\ 
\hline
\end{tabular}
\vspace{-15pt}
\end{table}

\begin{table*}[t]
\caption{
Recommendation performance under Spattack-O-D. We also report the performance drop rate w.r.t. clean model. 
}
\vskip -0.125in
\resizebox{\linewidth}{!}{
\begin{tabular}{c|c|cccc|cccc|cccc}
\hline
\multirow{2}{*}{Dataset}   & \multirow{2}{*}{Defense} & \multicolumn{4}{c|}{1\%}                                                                                                                                                                                                                                       & \multicolumn{4}{c|}{3\%}                                                                                                                                                                                                                  & \multicolumn{4}{c}{5\%}                                                                                                                                                                                                                   \\ \cline{3-14} 
                        &                          & HR@5                                                                           & nDCG@5                                                    & HR@10                                                     & nDCG@10                                                   & HR@5                                                      & nDCG@5                                                    & HR@10                                                     & nDCG@10                                                   & HR@5                                                      & nDCG@5                                                    & HR@10                                                     & nDCG@10                                                   \\ \hline
\multirow{10}{*}{ML100K} & Mean                     & \begin{tabular}[c]{@{}c@{}}0.0551\\ (-78\%)\end{tabular}                      & \begin{tabular}[c]{@{}c@{}}0.0354\\ (-78\%)\end{tabular} & \begin{tabular}[c]{@{}c@{}}0.0986\\ (-76\%)\end{tabular} & \begin{tabular}[c]{@{}c@{}}0.0491\\ (-77\%)\end{tabular} & \begin{tabular}[c]{@{}c@{}}0.0530\\ (-79\%)\end{tabular}  & \begin{tabular}[c]{@{}c@{}}0.0339\\ (-79\%)\end{tabular} & \begin{tabular}[c]{@{}c@{}}0.0944\\ (-77\%)\end{tabular} & \begin{tabular}[c]{@{}c@{}}0.0470\\ (-78\%)\end{tabular}  & \begin{tabular}[c]{@{}c@{}}0.0573\\ (-77\%)\end{tabular} & \begin{tabular}[c]{@{}c@{}}0.0352\\ (-79\%)\end{tabular} & \begin{tabular}[c]{@{}c@{}}0.0944\\ (-77\%)\end{tabular} & \begin{tabular}[c]{@{}c@{}}0.0470\\ (-78\%)\end{tabular}  \\
                        & Median                   & \begin{tabular}[c]{@{}c@{}}0.2312\\ (-9\%)\end{tabular}                       & \begin{tabular}[c]{@{}c@{}}0.1545\\ (-10\%)\end{tabular} & \begin{tabular}[c]{@{}c@{}}0.3510\\ (-9\%)\end{tabular}   & \begin{tabular}[c]{@{}c@{}}0.1924\\ (-10\%)\end{tabular} & \begin{tabular}[c]{@{}c@{}}0.1485\\ (-42\%)\end{tabular} & \begin{tabular}[c]{@{}c@{}}0.0943\\ (-45\%)\end{tabular} & \begin{tabular}[c]{@{}c@{}}0.2725\\ (-29\%)\end{tabular} & \begin{tabular}[c]{@{}c@{}}0.1339\\ (-37\%)\end{tabular} & \begin{tabular}[c]{@{}c@{}}0.0371\\ (-85\%)\end{tabular} & \begin{tabular}[c]{@{}c@{}}0.0233\\ (-87\%)\end{tabular} & \begin{tabular}[c]{@{}c@{}}0.0732\\ (-81\%)\end{tabular} & \begin{tabular}[c]{@{}c@{}}0.0346\\ (-84\%)\end{tabular} \\
                        & Norm                     & \begin{tabular}[c]{@{}c@{}}0.1972\\ (-17\%)\end{tabular}                      & \begin{tabular}[c]{@{}c@{}}0.1305\\ (-18\%)\end{tabular} & \begin{tabular}[c]{@{}c@{}}0.3181\\ (-18\%)\end{tabular} & \begin{tabular}[c]{@{}c@{}}0.1691\\ (-18\%)\end{tabular} & \begin{tabular}[c]{@{}c@{}}0.1410\\ (-41\%)\end{tabular}  & \begin{tabular}[c]{@{}c@{}}0.0981\\ (-38\%)\end{tabular} & \begin{tabular}[c]{@{}c@{}}0.2153\\ (-45\%)\end{tabular} & \begin{tabular}[c]{@{}c@{}}0.1216\\ (-41\%)\end{tabular} & \begin{tabular}[c]{@{}c@{}}0.0530\\ (-78\%)\end{tabular}  & \begin{tabular}[c]{@{}c@{}}0.0340\\ (-79\%)\end{tabular}  & \begin{tabular}[c]{@{}c@{}}0.1018\\ (-74\%)\end{tabular} & \begin{tabular}[c]{@{}c@{}}0.0496\\ (-76\%)\end{tabular} \\
                        & TrimM                    & \begin{tabular}[c]{@{}c@{}}0.2269\\ (-10\%)\end{tabular}                      & \begin{tabular}[c]{@{}c@{}}0.1488\\ (-9\%)\end{tabular}  & \begin{tabular}[c]{@{}c@{}}0.3489\\ (-14\%)\end{tabular} & \begin{tabular}[c]{@{}c@{}}0.1876\\ (-12\%)\end{tabular} & \begin{tabular}[c]{@{}c@{}}0.0647\\ (-74\%)\end{tabular} & \begin{tabular}[c]{@{}c@{}}0.0407\\ (-75\%)\end{tabular} & \begin{tabular}[c]{@{}c@{}}0.1198\\ (-70\%)\end{tabular} & \begin{tabular}[c]{@{}c@{}}0.0582\\ (-73\%)\end{tabular} & \begin{tabular}[c]{@{}c@{}}0.0361\\ (-86\%)\end{tabular} & \begin{tabular}[c]{@{}c@{}}0.0213\\ (-87\%)\end{tabular} & \begin{tabular}[c]{@{}c@{}}0.0721\\ (-82\%)\end{tabular} & \begin{tabular}[c]{@{}c@{}}0.0328\\ (-85\%)\end{tabular} \\
                        & Krum                     & \begin{tabular}[c]{@{}c@{}}0.1941\\ (+1\%)\end{tabular}                       & \begin{tabular}[c]{@{}c@{}}0.1235\\ (+3\%)\end{tabular}  & \begin{tabular}[c]{@{}c@{}}0.3065\\ (0\%)\end{tabular}   & \begin{tabular}[c]{@{}c@{}}0.1596\\ (+1\%)\end{tabular}  & \begin{tabular}[c]{@{}c@{}}0.0255\\ (-87\%)\end{tabular} & \begin{tabular}[c]{@{}c@{}}0.0134\\ (-89\%)\end{tabular} & \begin{tabular}[c]{@{}c@{}}0.0456\\ (-85\%)\end{tabular} & \begin{tabular}[c]{@{}c@{}}0.0199\\ (-87\%)\end{tabular} & \begin{tabular}[c]{@{}c@{}}0.0509\\ (-73\%)\end{tabular} & \begin{tabular}[c]{@{}c@{}}0.0295\\ (-75\%)\end{tabular} & \begin{tabular}[c]{@{}c@{}}0.0891\\ (-71\%)\end{tabular} & \begin{tabular}[c]{@{}c@{}}0.0418\\ (-73\%)\end{tabular} \\ \hline
\multirow{10}{*}{ML1M}   & Mean                     & \begin{tabular}[c]{@{}c@{}}0.1151\\ (-63\%)\end{tabular}                      & \begin{tabular}[c]{@{}c@{}}0.0702\\ (-66\%)\end{tabular} & \begin{tabular}[c]{@{}c@{}}0.2149\\ (-54\%)\end{tabular} & \begin{tabular}[c]{@{}c@{}}0.1022\\ (-60\%)\end{tabular} & \begin{tabular}[c]{@{}c@{}}0.0907\\ (-71\%)\end{tabular} & \begin{tabular}[c]{@{}c@{}}0.0549\\ (-73\%)\end{tabular} & \begin{tabular}[c]{@{}c@{}}0.1679\\ (-64\%)\end{tabular} & \begin{tabular}[c]{@{}c@{}}0.0793\\ (-69\%)\end{tabular} & \begin{tabular}[c]{@{}c@{}}0.0730\\ (-77\%)\end{tabular}  & \begin{tabular}[c]{@{}c@{}}0.0437\\ (-79\%)\end{tabular} & \begin{tabular}[c]{@{}c@{}}0.1366\\ (-70\%)\end{tabular} & \begin{tabular}[c]{@{}c@{}}0.0640\\ (-75\%)\end{tabular}  \\
                        & Median                   & \begin{tabular}[c]{@{}c@{}}0.2955\\ (-5\%)\end{tabular}                       & \begin{tabular}[c]{@{}c@{}}0.1975\\ (-4\%)\end{tabular}  & \begin{tabular}[c]{@{}c@{}}0.4422\\ (-5\%)\end{tabular}  & \begin{tabular}[c]{@{}c@{}}0.2446\\ (-4\%)\end{tabular}  & \begin{tabular}[c]{@{}c@{}}0.0394\\ (-87\%)\end{tabular} & \begin{tabular}[c]{@{}c@{}}0.0228\\ (-89\%)\end{tabular} & \begin{tabular}[c]{@{}c@{}}0.0839\\ (-82\%)\end{tabular} & \begin{tabular}[c]{@{}c@{}}0.0370\\ (-85\%)\end{tabular}  & \begin{tabular}[c]{@{}c@{}}0.0457\\ (-85\%)\end{tabular} & \begin{tabular}[c]{@{}c@{}}0.0270\\ (-87\%)\end{tabular}  & \begin{tabular}[c]{@{}c@{}}0.0919\\ (-80\%)\end{tabular} & \begin{tabular}[c]{@{}c@{}}0.0418\\ (-84\%)\end{tabular} \\
                        & Norm                     & \begin{tabular}[c]{@{}c@{}}0.3000\\ (-2\%)\end{tabular}                          & \begin{tabular}[c]{@{}c@{}}0.1981\\ (-2\%)\end{tabular}  & \begin{tabular}[c]{@{}c@{}}0.4465\\ (-2\%)\end{tabular}  & \begin{tabular}[c]{@{}c@{}}0.2453\\ (-2\%)\end{tabular}  & \begin{tabular}[c]{@{}c@{}}0.2901\\ (-5\%)\end{tabular}  & \begin{tabular}[c]{@{}c@{}}0.1893\\ (-7\%)\end{tabular}  & \begin{tabular}[c]{@{}c@{}}0.4306\\ (-5\%)\end{tabular}  & \begin{tabular}[c]{@{}c@{}}0.2347\\ (-6\%)\end{tabular}  & \begin{tabular}[c]{@{}c@{}}0.1442\\ (-53\%)\end{tabular} & \begin{tabular}[c]{@{}c@{}}0.0989\\ (-51\%)\end{tabular} & \begin{tabular}[c]{@{}c@{}}0.2104\\ (-54\%)\end{tabular} & \begin{tabular}[c]{@{}c@{}}0.1202\\ (-52\%)\end{tabular} \\
                        & TrimM                    & \begin{tabular}[c]{@{}c@{}}0.2593\\ (-17\%)\end{tabular}                      & \begin{tabular}[c]{@{}c@{}}0.1765\\ (-14\%)\end{tabular} & \begin{tabular}[c]{@{}c@{}}0.4151\\ (-10\%)\end{tabular} & \begin{tabular}[c]{@{}c@{}}0.2262\\ (-11\%)\end{tabular} & \begin{tabular}[c]{@{}c@{}}0.0391\\ (-87\%)\end{tabular} & \begin{tabular}[c]{@{}c@{}}0.0222\\ (-89\%)\end{tabular} & \begin{tabular}[c]{@{}c@{}}0.0838\\ (-82\%)\end{tabular} & \begin{tabular}[c]{@{}c@{}}0.0364\\ (-86\%)\end{tabular} & \begin{tabular}[c]{@{}c@{}}0.0445\\ (-86\%)\end{tabular} & \begin{tabular}[c]{@{}c@{}}0.0255\\ (-88\%)\end{tabular} & \begin{tabular}[c]{@{}c@{}}0.0863\\ (-81\%)\end{tabular} & \begin{tabular}[c]{@{}c@{}}0.0390\\ (-85\%)\end{tabular}  \\
                        & Krum                     & \begin{tabular}[c]{@{}c@{}}0.2361\\ (0\%)\end{tabular}                        & \begin{tabular}[c]{@{}c@{}}0.1504\\ (0\%)\end{tabular}   & \begin{tabular}[c]{@{}c@{}}0.3586\\ (-4\%)\end{tabular}  & \begin{tabular}[c]{@{}c@{}}0.1899\\ (-3\%)\end{tabular}  & \begin{tabular}[c]{@{}c@{}}0.0368\\ (-84\%)\end{tabular} & \begin{tabular}[c]{@{}c@{}}0.0216\\ (-86\%)\end{tabular} & \begin{tabular}[c]{@{}c@{}}0.0776\\ (-79\%)\end{tabular} & \begin{tabular}[c]{@{}c@{}}0.0346\\ (-82\%)\end{tabular} & \begin{tabular}[c]{@{}c@{}}0.0462\\ (-80\%)\end{tabular} & \begin{tabular}[c]{@{}c@{}}0.0268\\ (-82\%)\end{tabular} & \begin{tabular}[c]{@{}c@{}}0.0929\\ (-75\%)\end{tabular} & \begin{tabular}[c]{@{}c@{}}0.0418\\ (-79\%)\end{tabular} \\ \hline
\multirow{10}{*}{Steam}  & Mean                     & \multicolumn{1}{c}{\begin{tabular}[c]{@{}c@{}}0.0677\\ (-88\%)\end{tabular}} & \begin{tabular}[c]{@{}c@{}}0.0403\\ (-89\%)\end{tabular} & \begin{tabular}[c]{@{}c@{}}0.1276\\ (-82\%)\end{tabular} & \begin{tabular}[c]{@{}c@{}}0.0596\\ (-86\%)\end{tabular} & \begin{tabular}[c]{@{}c@{}}0.0685\\ (-88\%)\end{tabular} & \begin{tabular}[c]{@{}c@{}}0.0408\\ (-89\%)\end{tabular} & \begin{tabular}[c]{@{}c@{}}0.1287\\ (-81\%)\end{tabular} & \begin{tabular}[c]{@{}c@{}}0.0601\\ (-86\%)\end{tabular} & \begin{tabular}[c]{@{}c@{}}0.069\\ (-88\%)\end{tabular}  & \begin{tabular}[c]{@{}c@{}}0.0411\\ (-89\%)\end{tabular} & \begin{tabular}[c]{@{}c@{}}0.129\\ (-81\%)\end{tabular}  & \begin{tabular}[c]{@{}c@{}}0.0603\\ (-86\%)\end{tabular} \\
                        & Median                   & \multicolumn{1}{c}{\begin{tabular}[c]{@{}c@{}}0.1719\\ (-38\%)\end{tabular}} & \begin{tabular}[c]{@{}c@{}}0.1205\\ (-38\%)\end{tabular} & \begin{tabular}[c]{@{}c@{}}0.2323\\ (-54\%)\end{tabular} & \begin{tabular}[c]{@{}c@{}}0.1400\\ (-47\%)\end{tabular}   & \begin{tabular}[c]{@{}c@{}}0.0442\\ (-84\%)\end{tabular} & \begin{tabular}[c]{@{}c@{}}0.0265\\ (-86\%)\end{tabular} & \begin{tabular}[c]{@{}c@{}}0.0791\\ (-84\%)\end{tabular} & \begin{tabular}[c]{@{}c@{}}0.0376\\ (-86\%)\end{tabular} & \begin{tabular}[c]{@{}c@{}}0.0290\\ (-90\%)\end{tabular}  & \begin{tabular}[c]{@{}c@{}}0.0175\\ (-91\%)\end{tabular} & \begin{tabular}[c]{@{}c@{}}0.0568\\ (-89\%)\end{tabular} & \begin{tabular}[c]{@{}c@{}}0.0262\\ (-90\%)\end{tabular} \\
                        & Norm                     & \multicolumn{1}{c}{\begin{tabular}[c]{@{}c@{}}0.0717\\ (-87\%)\end{tabular}} & \begin{tabular}[c]{@{}c@{}}0.0428\\ (-88\%)\end{tabular} & \begin{tabular}[c]{@{}c@{}}0.1322\\ (-81\%)\end{tabular} & \begin{tabular}[c]{@{}c@{}}0.0622\\ (-84\%)\end{tabular} & \begin{tabular}[c]{@{}c@{}}0.0690\\ (-87\%)\end{tabular}  & \begin{tabular}[c]{@{}c@{}}0.0409\\ (-88\%)\end{tabular} & \begin{tabular}[c]{@{}c@{}}0.1292\\ (-81\%)\end{tabular} & \begin{tabular}[c]{@{}c@{}}0.0602\\ (-85\%)\end{tabular} & \begin{tabular}[c]{@{}c@{}}0.0682\\ (-87\%)\end{tabular} & \begin{tabular}[c]{@{}c@{}}0.0408\\ (-88\%)\end{tabular} & \begin{tabular}[c]{@{}c@{}}0.1300\\ (-81\%)\end{tabular}   & \begin{tabular}[c]{@{}c@{}}0.0605\\ (-85\%)\end{tabular} \\
                        & TrimM                    & \multicolumn{1}{c}{\begin{tabular}[c]{@{}c@{}}0.2001\\ (-65\%)\end{tabular}} & \begin{tabular}[c]{@{}c@{}}0.1622\\ (-57\%)\end{tabular} & \begin{tabular}[c]{@{}c@{}}0.2502\\ (-64\%)\end{tabular} & \begin{tabular}[c]{@{}c@{}}0.1783\\ (-58\%)\end{tabular} & \begin{tabular}[c]{@{}c@{}}0.0378\\ (-93\%)\end{tabular} & \begin{tabular}[c]{@{}c@{}}0.0228\\ (-94\%)\end{tabular} & \begin{tabular}[c]{@{}c@{}}0.0714\\ (-90\%)\end{tabular} & \begin{tabular}[c]{@{}c@{}}0.0335\\ (-92\%)\end{tabular} & \begin{tabular}[c]{@{}c@{}}0.0288\\ (-95\%)\end{tabular} & \begin{tabular}[c]{@{}c@{}}0.0175\\ (-95\%)\end{tabular} & \begin{tabular}[c]{@{}c@{}}0.0584\\ (-92\%)\end{tabular} & \begin{tabular}[c]{@{}c@{}}0.0269\\ (-94\%)\end{tabular} \\
                        & Krum                     & \multicolumn{1}{c}{\begin{tabular}[c]{@{}c@{}}0.1607\\ (-37\%)\end{tabular}} & \begin{tabular}[c]{@{}c@{}}0.1253\\ (-29\%)\end{tabular} & \begin{tabular}[c]{@{}c@{}}0.2118\\ (-55\%)\end{tabular} & \begin{tabular}[c]{@{}c@{}}0.1416\\ (-42\%)\end{tabular} & \begin{tabular}[c]{@{}c@{}}0.0381\\ (-85\%)\end{tabular} & \begin{tabular}[c]{@{}c@{}}0.0237\\ (-87\%)\end{tabular} & \begin{tabular}[c]{@{}c@{}}0.0709\\ (-85\%)\end{tabular} & \begin{tabular}[c]{@{}c@{}}0.0341\\ (-86\%)\end{tabular} & \begin{tabular}[c]{@{}c@{}}0.0290\\ (-89\%)\end{tabular}  & \begin{tabular}[c]{@{}c@{}}0.0175\\ (-90\%)\end{tabular} & \begin{tabular}[c]{@{}c@{}}0.0570\\ (-88\%)\end{tabular}  & \begin{tabular}[c]{@{}c@{}}0.0264\\ (-89\%)\end{tabular} \\ \hline
\end{tabular}}
\label{tab:main_type1_D}
\vskip -0.05in
\end{table*}

\begin{table*}[t]
\caption{
Recommendation performance under Spattack-O-S. We also report the performance drop rate w.r.t. clean model. 
}
\vskip -0.125in
\resizebox{\linewidth}{!}{
\begin{tabular}{c|c|cccc|cccc|cccc}
\hline
\multirow{2}{*}{Dataset}   & \multirow{2}{*}{Defense} & \multicolumn{4}{c|}{1\%}                                                                                                                                                                                                                  & \multicolumn{4}{c|}{3\%}                                                                                                                                                                                                                  & \multicolumn{4}{c}{5\%}                                                                                                                                                                                                                   \\ \cline{3-14} 
                        &                          & HR@5                                                      & nDCG@5                                                    & HR@10                                                     & nDCG@10                                                   & HR@5                                                      & nDCG@5                                                    & HR@10                                                     & nDCG@10                                                   & HR@5                                                      & nDCG@5                                                    & HR@10                                                     & nDCG@10                                                   \\ \hline
\multirow{5}{*}{ML100K} & Mean                     & \begin{tabular}[c]{@{}c@{}}0.0647\\ (-74\%)\end{tabular} & \begin{tabular}[c]{@{}c@{}}0.0373\\ (-77\%)\end{tabular} & \begin{tabular}[c]{@{}c@{}}0.1421\\ (-65\%)\end{tabular} & \begin{tabular}[c]{@{}c@{}}0.0618\\ (-71\%)\end{tabular} & \begin{tabular}[c]{@{}c@{}}0.0594\\ (-76\%)\end{tabular} & \begin{tabular}[c]{@{}c@{}}0.0362\\ (-78\%)\end{tabular} & \begin{tabular}[c]{@{}c@{}}0.0997\\ (-75\%)\end{tabular} & \begin{tabular}[c]{@{}c@{}}0.0492\\ (-77\%)\end{tabular} & \begin{tabular}[c]{@{}c@{}}0.0541\\ (-78\%)\end{tabular} & \begin{tabular}[c]{@{}c@{}}0.0318\\ (-81\%)\end{tabular} & \begin{tabular}[c]{@{}c@{}}0.1039\\ (-74\%)\end{tabular} & \begin{tabular}[c]{@{}c@{}}0.0479\\ (-78\%)\end{tabular} \\
                        & Median                   & \begin{tabular}[c]{@{}c@{}}0.2365\\ (-7\%)\end{tabular}  & \begin{tabular}[c]{@{}c@{}}0.1591\\ (-8\%)\end{tabular}  & \begin{tabular}[c]{@{}c@{}}0.3606\\ (-6\%)\end{tabular}  & \begin{tabular}[c]{@{}c@{}}0.1989\\ (-7\%)\end{tabular}  & \begin{tabular}[c]{@{}c@{}}0.1994\\ (-22\%)\end{tabular} & \begin{tabular}[c]{@{}c@{}}0.1312\\ (-24\%)\end{tabular} & \begin{tabular}[c]{@{}c@{}}0.3075\\ (-20\%)\end{tabular} & \begin{tabular}[c]{@{}c@{}}0.1660\\ (-22\%)\end{tabular} & \begin{tabular}[c]{@{}c@{}}0.1198\\ (-53\%)\end{tabular} & \begin{tabular}[c]{@{}c@{}}0.0720\\ (-58\%)\end{tabular} & \begin{tabular}[c]{@{}c@{}}0.2068\\ (-46\%)\end{tabular} & \begin{tabular}[c]{@{}c@{}}0.0996\\ (-54\%)\end{tabular} \\
                        & Norm                     & \begin{tabular}[c]{@{}c@{}}0.2216\\ (-7\%)\end{tabular}  & \begin{tabular}[c]{@{}c@{}}0.1417\\ (-11\%)\end{tabular} & \begin{tabular}[c]{@{}c@{}}0.3446\\ (-11\%)\end{tabular} & \begin{tabular}[c]{@{}c@{}}0.1811\\ (-13\%)\end{tabular} & \begin{tabular}[c]{@{}c@{}}0.1994\\ (-16\%)\end{tabular} & \begin{tabular}[c]{@{}c@{}}0.1249\\ (-22\%)\end{tabular} & \begin{tabular}[c]{@{}c@{}}0.3404\\ (-12\%)\end{tabular} & \begin{tabular}[c]{@{}c@{}}0.1698\\ (-18\%)\end{tabular} & \begin{tabular}[c]{@{}c@{}}0.1538\\ (-35\%)\end{tabular} & \begin{tabular}[c]{@{}c@{}}0.1032\\ (-35\%)\end{tabular} & \begin{tabular}[c]{@{}c@{}}0.2418\\ (-38\%)\end{tabular} & \begin{tabular}[c]{@{}c@{}}0.1314\\ (-37\%)\end{tabular} \\ \hline
\multirow{5}{*}{ML1M}   & Mean                     & \begin{tabular}[c]{@{}c@{}}0.1204\\ (-61\%)\end{tabular} & \begin{tabular}[c]{@{}c@{}}0.0738\\ (-64\%)\end{tabular} & \begin{tabular}[c]{@{}c@{}}0.2230\\ (-52\%)\end{tabular} & \begin{tabular}[c]{@{}c@{}}0.1065\\ (-58\%)\end{tabular} & \begin{tabular}[c]{@{}c@{}}0.0925\\ (-70\%)\end{tabular} & \begin{tabular}[c]{@{}c@{}}0.0553\\ (-73\%)\end{tabular} & \begin{tabular}[c]{@{}c@{}}0.1753\\ (-62\%)\end{tabular} & \begin{tabular}[c]{@{}c@{}}0.0817\\ (-68\%)\end{tabular} & \begin{tabular}[c]{@{}c@{}}0.0805\\ (-74\%)\end{tabular} & \begin{tabular}[c]{@{}c@{}}0.0493\\ (-76\%)\end{tabular} & \begin{tabular}[c]{@{}c@{}}0.1568\\ (-66\%)\end{tabular} & \begin{tabular}[c]{@{}c@{}}0.0736\\ (-71\%)\end{tabular} \\
                        & Median                   & \begin{tabular}[c]{@{}c@{}}0.2995\\ (-4\%)\end{tabular}  & \begin{tabular}[c]{@{}c@{}}0.2001\\ (-2\%)\end{tabular}  & \begin{tabular}[c]{@{}c@{}}0.4452\\ (-4\%)\end{tabular}  & \begin{tabular}[c]{@{}c@{}}0.2468\\ (-3\%)\end{tabular}  & \begin{tabular}[c]{@{}c@{}}0.2439\\ (-22\%)\end{tabular} & \begin{tabular}[c]{@{}c@{}}0.1570\\ (-23\%)\end{tabular} & \begin{tabular}[c]{@{}c@{}}0.3641\\ (-21\%)\end{tabular} & \begin{tabular}[c]{@{}c@{}}0.1957\\ (-23\%)\end{tabular} & \begin{tabular}[c]{@{}c@{}}0.0447\\ (-86\%)\end{tabular} & \begin{tabular}[c]{@{}c@{}}0.0264\\ (-87\%)\end{tabular} & \begin{tabular}[c]{@{}c@{}}0.0897\\ (-81\%)\end{tabular} & \begin{tabular}[c]{@{}c@{}}0.0407\\ (-84\%)\end{tabular} \\
                        & Norm                     & \begin{tabular}[c]{@{}c@{}}0.3028\\ (-1\%)\end{tabular}  & \begin{tabular}[c]{@{}c@{}}0.1986\\ (-2\%)\end{tabular}  & \begin{tabular}[c]{@{}c@{}}0.4520\\ (-1\%)\end{tabular}  & \begin{tabular}[c]{@{}c@{}}0.2468\\ (-2\%)\end{tabular}  & \begin{tabular}[c]{@{}c@{}}0.2902\\ (-5\%)\end{tabular}  & \begin{tabular}[c]{@{}c@{}}0.1898\\ (-6\%)\end{tabular}  & \begin{tabular}[c]{@{}c@{}}0.4382\\ (-4\%)\end{tabular}  & \begin{tabular}[c]{@{}c@{}}0.2375\\ (-5\%)\end{tabular}  & \begin{tabular}[c]{@{}c@{}}0.2023\\ (-34\%)\end{tabular} & \begin{tabular}[c]{@{}c@{}}0.1402\\ (-31\%)\end{tabular} & \begin{tabular}[c]{@{}c@{}}0.2793\\ (-39\%)\end{tabular} & \begin{tabular}[c]{@{}c@{}}0.1648\\ (-34\%)\end{tabular} \\ \hline
\multirow{5}{*}{Steam}  & Mean                     & \begin{tabular}[c]{@{}c@{}}0.0701\\ (-88\%)\end{tabular} & \begin{tabular}[c]{@{}c@{}}0.0410\\ (-89\%)\end{tabular} & \begin{tabular}[c]{@{}c@{}}0.1404\\ (-80\%)\end{tabular} & \begin{tabular}[c]{@{}c@{}}0.0635\\ (-85\%)\end{tabular} & \begin{tabular}[c]{@{}c@{}}0.0671\\ (-88\%)\end{tabular} & \begin{tabular}[c]{@{}c@{}}0.0390\\ (-90\%)\end{tabular} & \begin{tabular}[c]{@{}c@{}}0.1308\\ (-81\%)\end{tabular} & \begin{tabular}[c]{@{}c@{}}0.0593\\ (-86\%)\end{tabular} & \begin{tabular}[c]{@{}c@{}}0.0695\\ (-88\%)\end{tabular} & \begin{tabular}[c]{@{}c@{}}0.0410\\ (-89\%)\end{tabular} & \begin{tabular}[c]{@{}c@{}}0.1348\\ (-81\%)\end{tabular} & \begin{tabular}[c]{@{}c@{}}0.0617\\ (-85\%)\end{tabular} \\
                        & Median                   & \begin{tabular}[c]{@{}c@{}}0.3333\\ (+20\%)\end{tabular} & \begin{tabular}[c]{@{}c@{}}0.2448\\ (+27\%)\end{tabular} & \begin{tabular}[c]{@{}c@{}}0.4602\\ (-9\%)\end{tabular}  & \begin{tabular}[c]{@{}c@{}}0.2855\\ (+8\%)\end{tabular}  & \begin{tabular}[c]{@{}c@{}}0.0266\\ (-90\%)\end{tabular} & \begin{tabular}[c]{@{}c@{}}0.0152\\ (-92\%)\end{tabular} & \begin{tabular}[c]{@{}c@{}}0.0600\\ (-88\%)\end{tabular} & \begin{tabular}[c]{@{}c@{}}0.0258\\ (-90\%)\end{tabular} & \begin{tabular}[c]{@{}c@{}}0.0218\\ (-92\%)\end{tabular} & \begin{tabular}[c]{@{}c@{}}0.0115\\ (-94\%)\end{tabular} & \begin{tabular}[c]{@{}c@{}}0.0498\\ (-90\%)\end{tabular} & \begin{tabular}[c]{@{}c@{}}0.0204\\ (-92\%)\end{tabular} \\
                        & Norm                     & \begin{tabular}[c]{@{}c@{}}0.1761\\ (-67\%)\end{tabular} & \begin{tabular}[c]{@{}c@{}}0.1226\\ (-64\%)\end{tabular} & \begin{tabular}[c]{@{}c@{}}0.2673\\ (-61\%)\end{tabular} & \begin{tabular}[c]{@{}c@{}}0.1520\\ (-61\%)\end{tabular} & \begin{tabular}[c]{@{}c@{}}0.0685\\ (-87\%)\end{tabular} & \begin{tabular}[c]{@{}c@{}}0.0398\\ (-88\%)\end{tabular} & \begin{tabular}[c]{@{}c@{}}0.1324\\ (-81\%)\end{tabular} & \begin{tabular}[c]{@{}c@{}}0.0602\\ (-85\%)\end{tabular} & \begin{tabular}[c]{@{}c@{}}0.0703\\ (-87\%)\end{tabular} & \begin{tabular}[c]{@{}c@{}}0.0414\\ (-88\%)\end{tabular} & \begin{tabular}[c]{@{}c@{}}0.1359\\ (-80\%)\end{tabular} & \begin{tabular}[c]{@{}c@{}}0.0623\\ (-84\%)\end{tabular} \\ \hline
\end{tabular}}
\label{tab:main_type1_S}
\vskip -0.2in
\end{table*}
\subsection{Additional Implementation Details}
$\bullet$ \textbf{General settings}. The default optimization algorithm is stochastic gradient descent, and the learning rate is 0.01. We set the epoch number to 200 and the malicious clients conduct attacks from the first epoch. The number of malicious clients under different attack ratios $\rho$ are shown in Tab.~\ref{tab:attacker_num}. \\
$\bullet$ \textbf{FR models}. For the input user and item embedding, the dimension is set to 32. The hidden layer unit size is 64 in FedGNN. In FedMF, following~\cite{Rong2022FedRecAttackMP}, we initialize the representations of users and items with normal distribution with a mean value of 0 and a standard deviation of 0.01. In FedGNN, we initialize the GNN weight matrices with Xavier Glorot's initialization with a gain value of 1. \\
$\bullet$ \textbf{Spattack-O-S/D}. For omniscient attackers, we generate malicious gradients in the opposite direction of the benign ones. Such that the mean aggregator will output zero gradients, preventing the convergence of item embeddings. \\
$\bullet$ \textbf{Spattack-L-S/D}. For non-omniscient attackers, we sample the Gaussian noise with a mean value of 0 and a standard deviation of 1 as the current gradient. And all malicious clients will upload the same malicious gradient.\\
$\bullet$ \textbf{Spattack-O/L-S}. When the maximum number of poisoned items $\tilde{m}_{max}$ is limited, we sample the poisoned item list for each malicious client from the distribution of item degree. The items with more benign updates will receive more malicious updates. And the sampling operation of the malicious clients is non-repeatable. 

\subsection{Experiment Environment and Source Code}
All experiments are conducted on a Linux server with one GPU (NVIDIA GeForce RTX 3090 GPU) and CPU (Intel Xeon Gold 6348), and its operating system is Ubuntu 18.04.5. We implement Spattack with the deep learning library PyTorch. The main source code of Spattack can be found at https://github.com/zhongjian-zhang/Spattack.
\section{Supplemental Experimental Results}\label{app:exp} 
\begin{table*}[t]
\caption{
Recommendation performance under Spattack-L-D. We also report the performance drop ratio w.r.t. clean model. 
}
\vskip -0.125in
\resizebox{\linewidth}{!}{
\begin{tabular}{c|c|cccc|cccc|cccc}
\hline
\multirow{2}{*}{Dataset}   & \multirow{2}{*}{Defense} & \multicolumn{4}{c|}{5\%}                                                                                                                                                                                                                  & \multicolumn{4}{c|}{10\%}                                                                                                                                                                                                                 & \multicolumn{4}{c}{15\%}                                                                                                                                                                                                                  \\ \cline{3-14} 
                        &                          & HR@5                                                      & nDCG@5                                                    & HR@10                                                     & nDCG@10                                                   & HR@5                                                      & nDCG@5                                                    & HR@10                                                     & nDCG@10                                                   & HR@5                                                      & nDCG@5                                                    & HR@10                                                     & nDCG@10                                                   \\ \hline
\multirow{10}{*}{ML100K} & Mean                     & \begin{tabular}[c]{@{}c@{}}0.0318\\ (-87\%)\end{tabular} & \begin{tabular}[c]{@{}c@{}}0.0160\\ (-90\%)\end{tabular}  & \begin{tabular}[c]{@{}c@{}}0.0785\\ (-81\%)\end{tabular} & \begin{tabular}[c]{@{}c@{}}0.0309\\ (-85\%)\end{tabular} & \begin{tabular}[c]{@{}c@{}}0.0233\\ (-91\%)\end{tabular} & \begin{tabular}[c]{@{}c@{}}0.0121\\ (-93\%)\end{tabular} & \begin{tabular}[c]{@{}c@{}}0.0562\\ (-86\%)\end{tabular} & \begin{tabular}[c]{@{}c@{}}0.0225\\ (-89\%)\end{tabular} & \begin{tabular}[c]{@{}c@{}}0.0265\\ (-89\%)\end{tabular} & \begin{tabular}[c]{@{}c@{}}0.0139\\ (-92\%)\end{tabular} & \begin{tabular}[c]{@{}c@{}}0.0562\\ (-86\%)\end{tabular} & \begin{tabular}[c]{@{}c@{}}0.0231\\ (-89\%)\end{tabular} \\
                        & Median                   & \begin{tabular}[c]{@{}c@{}}0.2163\\ (-15\%)\end{tabular} & \begin{tabular}[c]{@{}c@{}}0.1324\\ (-23\%)\end{tabular} & \begin{tabular}[c]{@{}c@{}}0.3637\\ (-6\%)\end{tabular}  & \begin{tabular}[c]{@{}c@{}}0.1797\\ (-16\%)\end{tabular} & \begin{tabular}[c]{@{}c@{}}0.0095\\ (-96\%)\end{tabular} & \begin{tabular}[c]{@{}c@{}}0.0056\\ (-97\%)\end{tabular} & \begin{tabular}[c]{@{}c@{}}0.0233\\ (-94\%)\end{tabular} & \begin{tabular}[c]{@{}c@{}}0.0099\\ (-95\%)\end{tabular} & \begin{tabular}[c]{@{}c@{}}0.0233\\ (-91\%)\end{tabular} & \begin{tabular}[c]{@{}c@{}}0.0139\\ (-92\%)\end{tabular} & \begin{tabular}[c]{@{}c@{}}0.0530\\ (-86\%)\end{tabular}  & \begin{tabular}[c]{@{}c@{}}0.0234\\ (-89\%)\end{tabular} \\
                        & Norm                     & \begin{tabular}[c]{@{}c@{}}0.2418\\ (+2\%)\end{tabular}  & \begin{tabular}[c]{@{}c@{}}0.1480\\ (-7\%)\end{tabular}   & \begin{tabular}[c]{@{}c@{}}0.3659\\ (-6\%)\end{tabular}  & \begin{tabular}[c]{@{}c@{}}0.1876\\ (-10\%)\end{tabular} & \begin{tabular}[c]{@{}c@{}}0.1516\\ (-36\%)\end{tabular} & \begin{tabular}[c]{@{}c@{}}0.0910\\ (-43\%)\end{tabular}  & \begin{tabular}[c]{@{}c@{}}0.2969\\ (-23\%)\end{tabular} & \begin{tabular}[c]{@{}c@{}}0.138\\ (-33\%)\end{tabular}  & \begin{tabular}[c]{@{}c@{}}0.106\\ (-55\%)\end{tabular}  & \begin{tabular}[c]{@{}c@{}}0.0601\\ (-62\%)\end{tabular} & \begin{tabular}[c]{@{}c@{}}0.1877\\ (-52\%)\end{tabular} & \begin{tabular}[c]{@{}c@{}}0.086\\ (-59\%)\end{tabular}  \\
                        & TrimM                    & \begin{tabular}[c]{@{}c@{}}0.2269\\ (-10\%)\end{tabular} & \begin{tabular}[c]{@{}c@{}}0.1511\\ (-8\%)\end{tabular}  & \begin{tabular}[c]{@{}c@{}}0.3690\\ (-9\%)\end{tabular}   & \begin{tabular}[c]{@{}c@{}}0.1968\\ (-8\%)\end{tabular}  & \begin{tabular}[c]{@{}c@{}}0.0074\\ (-97\%)\end{tabular} & \begin{tabular}[c]{@{}c@{}}0.0041\\ (-97\%)\end{tabular} & \begin{tabular}[c]{@{}c@{}}0.0286\\ (-93\%)\end{tabular} & \begin{tabular}[c]{@{}c@{}}0.0110\\ (-95\%)\end{tabular}  & \begin{tabular}[c]{@{}c@{}}0.0276\\ (-89\%)\end{tabular} & \begin{tabular}[c]{@{}c@{}}0.0162\\ (-90\%)\end{tabular} & \begin{tabular}[c]{@{}c@{}}0.0551\\ (-86\%)\end{tabular} & \begin{tabular}[c]{@{}c@{}}0.0249\\ (-88\%)\end{tabular} \\
                        & Krum                     & \begin{tabular}[c]{@{}c@{}}0.1474\\ (-23\%)\end{tabular} & \begin{tabular}[c]{@{}c@{}}0.0879\\ (-27\%)\end{tabular} & \begin{tabular}[c]{@{}c@{}}0.2259\\ (-27\%)\end{tabular} & \begin{tabular}[c]{@{}c@{}}0.1132\\ (-28\%)\end{tabular} & \begin{tabular}[c]{@{}c@{}}0.0159\\ (-92\%)\end{tabular} & \begin{tabular}[c]{@{}c@{}}0.0101\\ (-92\%)\end{tabular} & \begin{tabular}[c]{@{}c@{}}0.0318\\ (-90\%)\end{tabular} & \begin{tabular}[c]{@{}c@{}}0.0152\\ (-90\%)\end{tabular} & \begin{tabular}[c]{@{}c@{}}0.0286\\ (-85\%)\end{tabular} & \begin{tabular}[c]{@{}c@{}}0.0159\\ (-87\%)\end{tabular} & \begin{tabular}[c]{@{}c@{}}0.0615\\ (-80\%)\end{tabular} & \begin{tabular}[c]{@{}c@{}}0.0264\\ (-83\%)\end{tabular} \\ \hline
\multirow{7}{*}{ML1M}   & Mean                     & \begin{tabular}[c]{@{}c@{}}0.0272\\ (-91\%)\end{tabular} & \begin{tabular}[c]{@{}c@{}}0.0153\\ (-93\%)\end{tabular} & \begin{tabular}[c]{@{}c@{}}0.0720\\ (-84\%)\end{tabular}  & \begin{tabular}[c]{@{}c@{}}0.0295\\ (-88\%)\end{tabular} & \begin{tabular}[c]{@{}c@{}}0.0237\\ (-92\%)\end{tabular} & \begin{tabular}[c]{@{}c@{}}0.0128\\ (-94\%)\end{tabular} & \begin{tabular}[c]{@{}c@{}}0.0523\\ (-89\%)\end{tabular} & \begin{tabular}[c]{@{}c@{}}0.0219\\ (-91\%)\end{tabular} & \begin{tabular}[c]{@{}c@{}}0.0260\\ (-92\%)\end{tabular}  & \begin{tabular}[c]{@{}c@{}}0.0146\\ (-93\%)\end{tabular} & \begin{tabular}[c]{@{}c@{}}0.0603\\ (-87\%)\end{tabular} & \begin{tabular}[c]{@{}c@{}}0.0254\\ (-90\%)\end{tabular} \\
                        & Median                   & \begin{tabular}[c]{@{}c@{}}0.0121\\ (-96\%)\end{tabular} & \begin{tabular}[c]{@{}c@{}}0.0055\\ (-97\%)\end{tabular} & \begin{tabular}[c]{@{}c@{}}0.0894\\ (-81\%)\end{tabular} & \begin{tabular}[c]{@{}c@{}}0.0296\\ (-88\%)\end{tabular} & \begin{tabular}[c]{@{}c@{}}0.024\\ (-92\%)\end{tabular}  & \begin{tabular}[c]{@{}c@{}}0.0144\\ (-93\%)\end{tabular} & \begin{tabular}[c]{@{}c@{}}0.046\\ (-90\%)\end{tabular}  & \begin{tabular}[c]{@{}c@{}}0.0215\\ (-92\%)\end{tabular} & \begin{tabular}[c]{@{}c@{}}0.0387\\ (-88\%)\end{tabular} & \begin{tabular}[c]{@{}c@{}}0.0220\\ (-89\%)\end{tabular}  & \begin{tabular}[c]{@{}c@{}}0.0732\\ (-84\%)\end{tabular} & \begin{tabular}[c]{@{}c@{}}0.0330\\ (-87\%)\end{tabular}  \\
                        & Norm                     & \begin{tabular}[c]{@{}c@{}}0.2805\\ (-8\%)\end{tabular}  & \begin{tabular}[c]{@{}c@{}}0.1822\\ (-10\%)\end{tabular} & \begin{tabular}[c]{@{}c@{}}0.4333\\ (-5\%)\end{tabular}  & \begin{tabular}[c]{@{}c@{}}0.2313\\ (-8\%)\end{tabular}  & \begin{tabular}[c]{@{}c@{}}0.1962\\ (-36\%)\end{tabular} & \begin{tabular}[c]{@{}c@{}}0.1185\\ (-42\%)\end{tabular} & \begin{tabular}[c]{@{}c@{}}0.3343\\ (-27\%)\end{tabular} & \begin{tabular}[c]{@{}c@{}}0.1628\\ (-35\%)\end{tabular} & \begin{tabular}[c]{@{}c@{}}0.1119\\ (-63\%)\end{tabular} & \begin{tabular}[c]{@{}c@{}}0.0664\\ (-67\%)\end{tabular} & \begin{tabular}[c]{@{}c@{}}0.2240\\ (-51\%)\end{tabular}  & \begin{tabular}[c]{@{}c@{}}0.1023\\ (-59\%)\end{tabular} \\
                        & TrimM                    & \begin{tabular}[c]{@{}c@{}}0.0232\\ (-93\%)\end{tabular} & \begin{tabular}[c]{@{}c@{}}0.0103\\ (-95\%)\end{tabular} & \begin{tabular}[c]{@{}c@{}}0.1611\\ (-65\%)\end{tabular} & \begin{tabular}[c]{@{}c@{}}0.0537\\ (-79\%)\end{tabular} & \begin{tabular}[c]{@{}c@{}}0.0359\\ (-88\%)\end{tabular} & \begin{tabular}[c]{@{}c@{}}0.0207\\ (-90\%)\end{tabular} & \begin{tabular}[c]{@{}c@{}}0.0717\\ (-85\%)\end{tabular} & \begin{tabular}[c]{@{}c@{}}0.0321\\ (-87\%)\end{tabular} & \begin{tabular}[c]{@{}c@{}}0.0255\\ (-92\%)\end{tabular} & \begin{tabular}[c]{@{}c@{}}0.0152\\ (-93\%)\end{tabular} & \begin{tabular}[c]{@{}c@{}}0.0503\\ (-89\%)\end{tabular} & \begin{tabular}[c]{@{}c@{}}0.0231\\ (-91\%)\end{tabular} \\ \hline
\multirow{10}{*}{Steam}  & Mean                     & \begin{tabular}[c]{@{}c@{}}0.0205\\ (-96\%)\end{tabular} & \begin{tabular}[c]{@{}c@{}}0.0114\\ (-97\%)\end{tabular} & \begin{tabular}[c]{@{}c@{}}0.0384\\ (-94\%)\end{tabular} & \begin{tabular}[c]{@{}c@{}}0.0170\\ (-96\%)\end{tabular}  & \begin{tabular}[c]{@{}c@{}}0.0226\\ (-96\%)\end{tabular} & \begin{tabular}[c]{@{}c@{}}0.0129\\ (-97\%)\end{tabular} & \begin{tabular}[c]{@{}c@{}}0.0520\\ (-93\%)\end{tabular}  & \begin{tabular}[c]{@{}c@{}}0.0221\\ (-95\%)\end{tabular} & \begin{tabular}[c]{@{}c@{}}0.0274\\ (-95\%)\end{tabular} & \begin{tabular}[c]{@{}c@{}}0.0155\\ (-96\%)\end{tabular} & \begin{tabular}[c]{@{}c@{}}0.0592\\ (-91\%)\end{tabular} & \begin{tabular}[c]{@{}c@{}}0.0256\\ (-94\%)\end{tabular} \\
                        & Median                   & \begin{tabular}[c]{@{}c@{}}0.0285\\ (-90\%)\end{tabular} & \begin{tabular}[c]{@{}c@{}}0.0160\\ (-92\%)\end{tabular}  & \begin{tabular}[c]{@{}c@{}}0.0549\\ (-89\%)\end{tabular} & \begin{tabular}[c]{@{}c@{}}0.0245\\ (-91\%)\end{tabular} & \begin{tabular}[c]{@{}c@{}}0.0434\\ (-84\%)\end{tabular} & \begin{tabular}[c]{@{}c@{}}0.0261\\ (-87\%)\end{tabular} & \begin{tabular}[c]{@{}c@{}}0.0874\\ (-83\%)\end{tabular} & \begin{tabular}[c]{@{}c@{}}0.0400\\ (-85\%)\end{tabular}   & \begin{tabular}[c]{@{}c@{}}0.0493\\ (-82\%)\end{tabular} & \begin{tabular}[c]{@{}c@{}}0.0293\\ (-85\%)\end{tabular} & \begin{tabular}[c]{@{}c@{}}0.0906\\ (-82\%)\end{tabular} & \begin{tabular}[c]{@{}c@{}}0.0426\\ (-84\%)\end{tabular} \\
                        & Norm                     & \begin{tabular}[c]{@{}c@{}}0.0171\\ (-97\%)\end{tabular} & \begin{tabular}[c]{@{}c@{}}0.0098\\ (-97\%)\end{tabular} & \begin{tabular}[c]{@{}c@{}}0.0442\\ (-94\%)\end{tabular} & \begin{tabular}[c]{@{}c@{}}0.0184\\ (-95\%)\end{tabular} & \begin{tabular}[c]{@{}c@{}}0.0226\\ (-96\%)\end{tabular} & \begin{tabular}[c]{@{}c@{}}0.0121\\ (-96\%)\end{tabular} & \begin{tabular}[c]{@{}c@{}}0.0560\\ (-92\%)\end{tabular}  & \begin{tabular}[c]{@{}c@{}}0.0227\\ (-94\%)\end{tabular} & \begin{tabular}[c]{@{}c@{}}0.0250\\ (-95\%)\end{tabular}  & \begin{tabular}[c]{@{}c@{}}0.0134\\ (-96\%)\end{tabular} & \begin{tabular}[c]{@{}c@{}}0.0634\\ (-91\%)\end{tabular} & \begin{tabular}[c]{@{}c@{}}0.0256\\ (-94\%)\end{tabular} \\
                        & TrimM                    & \begin{tabular}[c]{@{}c@{}}0.0296\\ (-95\%)\end{tabular} & \begin{tabular}[c]{@{}c@{}}0.0166\\ (-96\%)\end{tabular} & \begin{tabular}[c]{@{}c@{}}0.0552\\ (-92\%)\end{tabular} & \begin{tabular}[c]{@{}c@{}}0.0248\\ (-94\%)\end{tabular} & \begin{tabular}[c]{@{}c@{}}0.0440\\ (-92\%)\end{tabular}  & \begin{tabular}[c]{@{}c@{}}0.0263\\ (-93\%)\end{tabular} & \begin{tabular}[c]{@{}c@{}}0.0879\\ (-87\%)\end{tabular} & \begin{tabular}[c]{@{}c@{}}0.0402\\ (-90\%)\end{tabular} & \begin{tabular}[c]{@{}c@{}}0.0480\\ (-92\%)\end{tabular}  & \begin{tabular}[c]{@{}c@{}}0.0285\\ (-93\%)\end{tabular} & \begin{tabular}[c]{@{}c@{}}0.0903\\ (-87\%)\end{tabular} & \begin{tabular}[c]{@{}c@{}}0.0421\\ (-90\%)\end{tabular} \\
                        & Krum                     & \begin{tabular}[c]{@{}c@{}}0.0288\\ (-89\%)\end{tabular} & \begin{tabular}[c]{@{}c@{}}0.0168\\ (-90\%)\end{tabular} & \begin{tabular}[c]{@{}c@{}}0.0554\\ (-88\%)\end{tabular} & \begin{tabular}[c]{@{}c@{}}0.0253\\ (-90\%)\end{tabular} & \begin{tabular}[c]{@{}c@{}}0.0434\\ (-83\%)\end{tabular} & \begin{tabular}[c]{@{}c@{}}0.0265\\ (-85\%)\end{tabular} & \begin{tabular}[c]{@{}c@{}}0.0866\\ (-81\%)\end{tabular} & \begin{tabular}[c]{@{}c@{}}0.0402\\ (-84\%)\end{tabular} & \begin{tabular}[c]{@{}c@{}}0.0472\\ (-81\%)\end{tabular} & \begin{tabular}[c]{@{}c@{}}0.0281\\ (-84\%)\end{tabular} & \begin{tabular}[c]{@{}c@{}}0.0895\\ (-81\%)\end{tabular} & \begin{tabular}[c]{@{}c@{}}0.0416\\ (-83\%)\end{tabular} \\ \hline
\end{tabular}}
\label{tab:main_type2_D}
\vskip -0.05in
\end{table*}

\begin{table*}[t]
\caption{
Recommendation performance under Spattack-L-S. We also report the performance drop rate w.r.t. clean model. 
}
\vskip -0.125in
\resizebox{\linewidth}{!}{
\begin{tabular}{c|c|cccc|cccc|cccc}
\hline
\multirow{2}{*}{Dataset}   & \multirow{2}{*}{Defense} & \multicolumn{4}{c|}{5\%}                                                                                                                                                                                                                  & \multicolumn{4}{c|}{10\%}                                                                                                                                                                                                                 & \multicolumn{4}{c}{15\%}                                                                                                                                                                                                                  \\ \cline{3-14} 
                        &                          & HR@5                                                      & nDCG@5                                                    & HR@10                                                     & nDCG@10                                                   & HR@5                                                      & nDCG@5                                                    & HR@10                                                     & nDCG@10                                                   & HR@5                                                      & nDCG@5                                                    & HR@10                                                     & nDCG@10                                                   \\ \hline
\multirow{5}{*}{ML100K} & Mean                     & \begin{tabular}[c]{@{}c@{}}0.0742\\ (-70\%)\end{tabular} & \begin{tabular}[c]{@{}c@{}}0.0425\\ (-74\%)\end{tabular} & \begin{tabular}[c]{@{}c@{}}0.1559\\ (-62\%)\end{tabular} & \begin{tabular}[c]{@{}c@{}}0.0683\\ (-68\%)\end{tabular} & \begin{tabular}[c]{@{}c@{}}0.0456\\ (-82\%)\end{tabular} & \begin{tabular}[c]{@{}c@{}}0.0249\\ (-85\%)\end{tabular} & \begin{tabular}[c]{@{}c@{}}0.1124\\ (-72\%)\end{tabular} & \begin{tabular}[c]{@{}c@{}}0.0462\\ (-78\%)\end{tabular} & \begin{tabular}[c]{@{}c@{}}0.0477\\ (-81\%)\end{tabular} & \begin{tabular}[c]{@{}c@{}}0.0264\\ (-84\%)\end{tabular} & \begin{tabular}[c]{@{}c@{}}0.0997\\ (-75\%)\end{tabular} & \begin{tabular}[c]{@{}c@{}}0.0428\\ (-80\%)\end{tabular} \\
                        & Median                   & \begin{tabular}[c]{@{}c@{}}0.2238\\ (-12\%)\end{tabular} & \begin{tabular}[c]{@{}c@{}}0.1433\\ (-17\%)\end{tabular} & \begin{tabular}[c]{@{}c@{}}0.3733\\ (-3\%)\end{tabular}  & \begin{tabular}[c]{@{}c@{}}0.1913\\ (-11\%)\end{tabular} & \begin{tabular}[c]{@{}c@{}}0.0308\\ (-88\%)\end{tabular} & \begin{tabular}[c]{@{}c@{}}0.0179\\ (-90\%)\end{tabular} & \begin{tabular}[c]{@{}c@{}}0.0764\\ (-80\%)\end{tabular} & \begin{tabular}[c]{@{}c@{}}0.0321\\ (-85\%)\end{tabular} & \begin{tabular}[c]{@{}c@{}}0.0339\\ (-87\%)\end{tabular} & \begin{tabular}[c]{@{}c@{}}0.0206\\ (-88\%)\end{tabular} & \begin{tabular}[c]{@{}c@{}}0.0742\\ (-81\%)\end{tabular} & \begin{tabular}[c]{@{}c@{}}0.0336\\ (-84\%)\end{tabular} \\
                        & Norm                     & \begin{tabular}[c]{@{}c@{}}0.2397\\ (+1\%)\end{tabular}  & \begin{tabular}[c]{@{}c@{}}0.1525\\ (-4\%)\end{tabular}  & \begin{tabular}[c]{@{}c@{}}0.3690\\ (-5\%)\end{tabular}   & \begin{tabular}[c]{@{}c@{}}0.1935\\ (-7\%)\end{tabular}  & \begin{tabular}[c]{@{}c@{}}0.2068\\ (-13\%)\end{tabular} & \begin{tabular}[c]{@{}c@{}}0.1243\\ (-22\%)\end{tabular} & \begin{tabular}[c]{@{}c@{}}0.3118\\ (-20\%)\end{tabular} & \begin{tabular}[c]{@{}c@{}}0.1579\\ (-24\%)\end{tabular} & \begin{tabular}[c]{@{}c@{}}0.1315\\ (-45\%)\end{tabular} & \begin{tabular}[c]{@{}c@{}}0.0804\\ (-50\%)\end{tabular} & \begin{tabular}[c]{@{}c@{}}0.2503\\ (-36\%)\end{tabular} & \begin{tabular}[c]{@{}c@{}}0.1187\\ (-43\%)\end{tabular} \\ \hline
\multirow{5}{*}{ML1M}   & Mean                     & \begin{tabular}[c]{@{}c@{}}0.0359\\ (-88\%)\end{tabular} & \begin{tabular}[c]{@{}c@{}}0.0207\\ (-90\%)\end{tabular} & \begin{tabular}[c]{@{}c@{}}0.0917\\ (-80\%)\end{tabular} & \begin{tabular}[c]{@{}c@{}}0.0385\\ (-85\%)\end{tabular} & \begin{tabular}[c]{@{}c@{}}0.0268\\ (-91\%)\end{tabular} & \begin{tabular}[c]{@{}c@{}}0.0156\\ (-92\%)\end{tabular} & \begin{tabular}[c]{@{}c@{}}0.0641\\ (-86\%)\end{tabular} & \begin{tabular}[c]{@{}c@{}}0.0274\\ (-89\%)\end{tabular} & \begin{tabular}[c]{@{}c@{}}0.0288\\ (-91\%)\end{tabular} & \begin{tabular}[c]{@{}c@{}}0.0166\\ (-92\%)\end{tabular} & \begin{tabular}[c]{@{}c@{}}0.0679\\ (-85\%)\end{tabular} & \begin{tabular}[c]{@{}c@{}}0.0290\\ (-89\%)\end{tabular}  \\
                        & Median                   & \begin{tabular}[c]{@{}c@{}}0.0825\\ (-74\%)\end{tabular} & \begin{tabular}[c]{@{}c@{}}0.039\\ (-81\%)\end{tabular}  & \begin{tabular}[c]{@{}c@{}}0.2237\\ (-52\%)\end{tabular} & \begin{tabular}[c]{@{}c@{}}0.0840\\ (-67\%)\end{tabular}  & \begin{tabular}[c]{@{}c@{}}0.0285\\ (-91\%)\end{tabular} & \begin{tabular}[c]{@{}c@{}}0.0169\\ (-92\%)\end{tabular} & \begin{tabular}[c]{@{}c@{}}0.0536\\ (-88\%)\end{tabular} & \begin{tabular}[c]{@{}c@{}}0.0250\\ (-90\%)\end{tabular}  & \begin{tabular}[c]{@{}c@{}}0.0387\\ (-88\%)\end{tabular} & \begin{tabular}[c]{@{}c@{}}0.0229\\ (-89\%)\end{tabular} & \begin{tabular}[c]{@{}c@{}}0.0728\\ (-84\%)\end{tabular} & \begin{tabular}[c]{@{}c@{}}0.0338\\ (-87\%)\end{tabular} \\
                        & Norm                     & \begin{tabular}[c]{@{}c@{}}0.2815\\ (-8\%)\end{tabular}  & \begin{tabular}[c]{@{}c@{}}0.1822\\ (-10\%)\end{tabular} & \begin{tabular}[c]{@{}c@{}}0.4368\\ (-4\%)\end{tabular}  & \begin{tabular}[c]{@{}c@{}}0.2321\\ (-7\%)\end{tabular}  & \begin{tabular}[c]{@{}c@{}}0.2063\\ (-33\%)\end{tabular} & \begin{tabular}[c]{@{}c@{}}0.1246\\ (-38\%)\end{tabular} & \begin{tabular}[c]{@{}c@{}}0.3464\\ (-24\%)\end{tabular} & \begin{tabular}[c]{@{}c@{}}0.1696\\ (-32\%)\end{tabular} & \begin{tabular}[c]{@{}c@{}}0.1281\\ (-58\%)\end{tabular} & \begin{tabular}[c]{@{}c@{}}0.0746\\ (-63\%)\end{tabular} & \begin{tabular}[c]{@{}c@{}}0.2445\\ (-46\%)\end{tabular} & \begin{tabular}[c]{@{}c@{}}0.1118\\ (-55\%)\end{tabular} \\ \hline
\multirow{5}{*}{Steam}  & Mean                     & \begin{tabular}[c]{@{}c@{}}0.0453\\ (-92\%)\end{tabular} & \begin{tabular}[c]{@{}c@{}}0.0259\\ (-93\%)\end{tabular} & \begin{tabular}[c]{@{}c@{}}0.0901\\ (-87\%)\end{tabular} & \begin{tabular}[c]{@{}c@{}}0.0401\\ (-90\%)\end{tabular} & \begin{tabular}[c]{@{}c@{}}0.0440\\ (-92\%)\end{tabular}  & \begin{tabular}[c]{@{}c@{}}0.0258\\ (-93\%)\end{tabular} & \begin{tabular}[c]{@{}c@{}}0.0914\\ (-87\%)\end{tabular} & \begin{tabular}[c]{@{}c@{}}0.0408\\ (-90\%)\end{tabular} & \begin{tabular}[c]{@{}c@{}}0.0410\\ (-93\%)\end{tabular}  & \begin{tabular}[c]{@{}c@{}}0.0239\\ (-94\%)\end{tabular} & \begin{tabular}[c]{@{}c@{}}0.0890\\ (-87\%)\end{tabular}  & \begin{tabular}[c]{@{}c@{}}0.0393\\ (-91\%)\end{tabular} \\
                        & Median                   & \begin{tabular}[c]{@{}c@{}}0.0378\\ (-86\%)\end{tabular} & \begin{tabular}[c]{@{}c@{}}0.0221\\ (-89\%)\end{tabular} & \begin{tabular}[c]{@{}c@{}}0.0791\\ (-84\%)\end{tabular} & \begin{tabular}[c]{@{}c@{}}0.0351\\ (-87\%)\end{tabular} & \begin{tabular}[c]{@{}c@{}}0.0488\\ (-82\%)\end{tabular} & \begin{tabular}[c]{@{}c@{}}0.0285\\ (-85\%)\end{tabular} & \begin{tabular}[c]{@{}c@{}}0.0927\\ (-82\%)\end{tabular} & \begin{tabular}[c]{@{}c@{}}0.0426\\ (-84\%)\end{tabular} & \begin{tabular}[c]{@{}c@{}}0.0450\\ (-84\%)\end{tabular}  & \begin{tabular}[c]{@{}c@{}}0.0264\\ (-86\%)\end{tabular} & \begin{tabular}[c]{@{}c@{}}0.0938\\ (-81\%)\end{tabular} & \begin{tabular}[c]{@{}c@{}}0.0420\\ (-84\%)\end{tabular}  \\
                        & Norm                     & \begin{tabular}[c]{@{}c@{}}0.0679\\ (-87\%)\end{tabular} & \begin{tabular}[c]{@{}c@{}}0.0389\\ (-89\%)\end{tabular} & \begin{tabular}[c]{@{}c@{}}0.1628\\ (-76\%)\end{tabular} & \begin{tabular}[c]{@{}c@{}}0.0692\\ (-82\%)\end{tabular} & \begin{tabular}[c]{@{}c@{}}0.0554\\ (-90\%)\end{tabular} & \begin{tabular}[c]{@{}c@{}}0.0322\\ (-91\%)\end{tabular} & \begin{tabular}[c]{@{}c@{}}0.1138\\ (-84\%)\end{tabular} & \begin{tabular}[c]{@{}c@{}}0.0508\\ (-87\%)\end{tabular} & \begin{tabular}[c]{@{}c@{}}0.0469\\ (-91\%)\end{tabular} & \begin{tabular}[c]{@{}c@{}}0.0265\\ (-92\%)\end{tabular} & \begin{tabular}[c]{@{}c@{}}0.1031\\ (-85\%)\end{tabular} & \begin{tabular}[c]{@{}c@{}}0.0444\\ (-89\%)\end{tabular} \\ \hline
\end{tabular}}
\label{tab:main_type2_S}
\vskip -0.2in
\end{table*}
\subsection{Performance Evaluation of Spattack-O-D}
We report the detailed results of Spattack-O-D in Tab.~\ref{tab:main_type1_D}, where attackers know total benign gradients and with no limitation of maximum number of updating items. We first find that with the increasing malicious ratio $\rho$, the defense performance under Spattack-O-D consistently decreases and even reaches an untrained state, indicating existing defenses are more fragile in FR than expected. When $\rho = 3\%$, the averaged performance drop rate of defenses is about 71\%, and the degradation will further increase to 82\% when $\rho=5\%$. The reason is that these tailed items in FR have lower defense breaking points and can be easily broken. 
 
\subsection{Performance Evaluation of Spattack-O-S}
The detailed results of Spattack-O-S are reported in Tab.~\ref{tab:main_type1_S}, where each malicious client only uploads malicious gradients for partial items to avoid triggering the anomaly detection based on the user's degree. We restrict the $\tilde{m}_{max}$ as the maximum number of uploading items in benign clients. As seen, Spattack-O-S can still significantly degrade recommendation performance under 1\% malicious clients and prevent the model convergence for a 5\% ratio. For example, Spattack-O-S achieves 35\% (ML100K), 31\% (ML1M), and 80\% (Steam) drop rate at least when the malicious client ratio $\rho$ is 5\%. Overall, Spattack-O-S can bypass anomaly detection based on the user's degree and still achieve successful attacks with few malicious clients.
\begin{figure*}[t]
\centering
\minipage{0.23\textwidth}
\centering
\includegraphics[width=\textwidth]{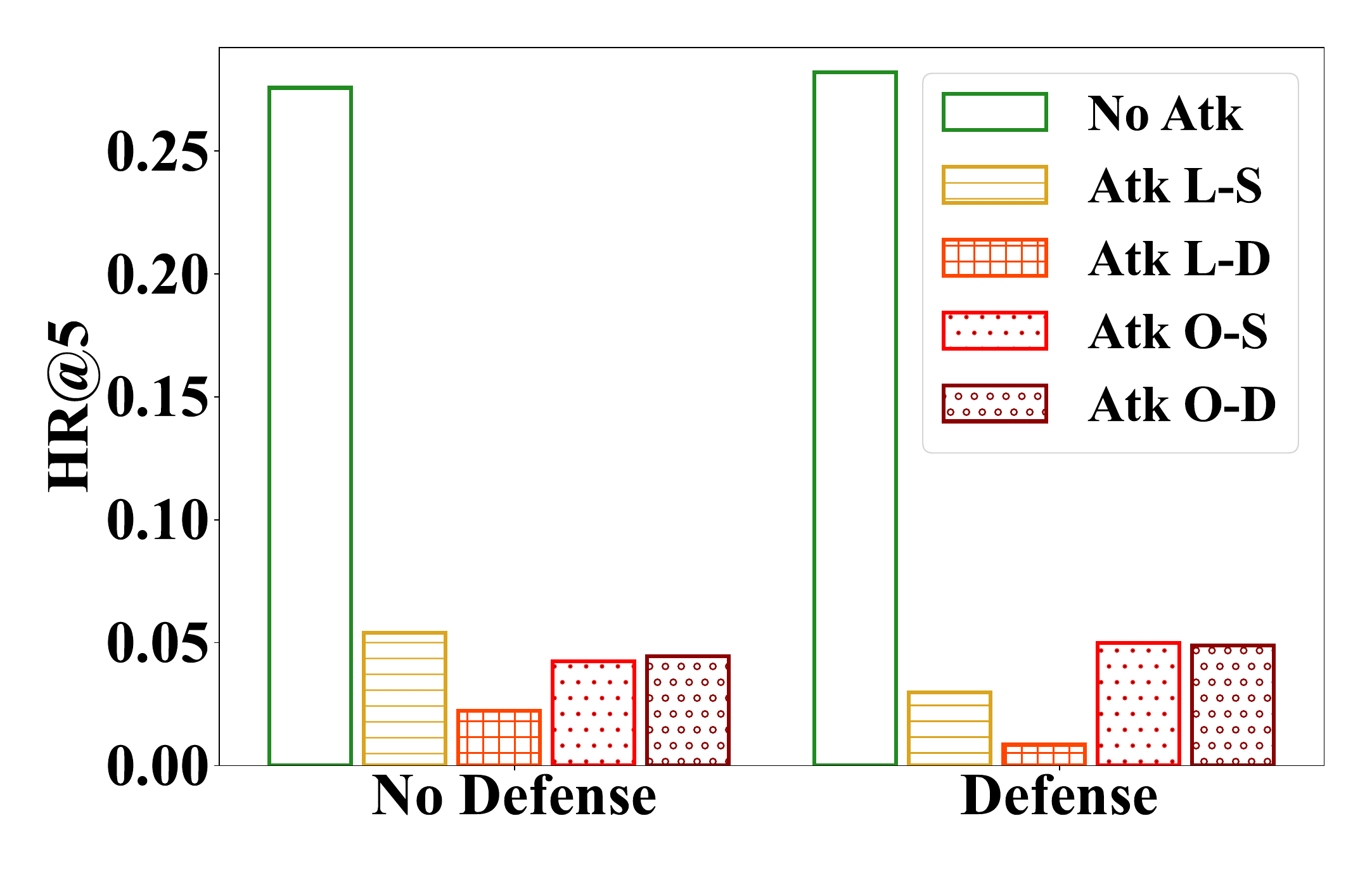}
\vspace{-22pt}
\caption*{\small (a) ML100K HR@5}
\endminipage\hfill
\minipage{0.23\textwidth}
\centering
\includegraphics[width=\textwidth]{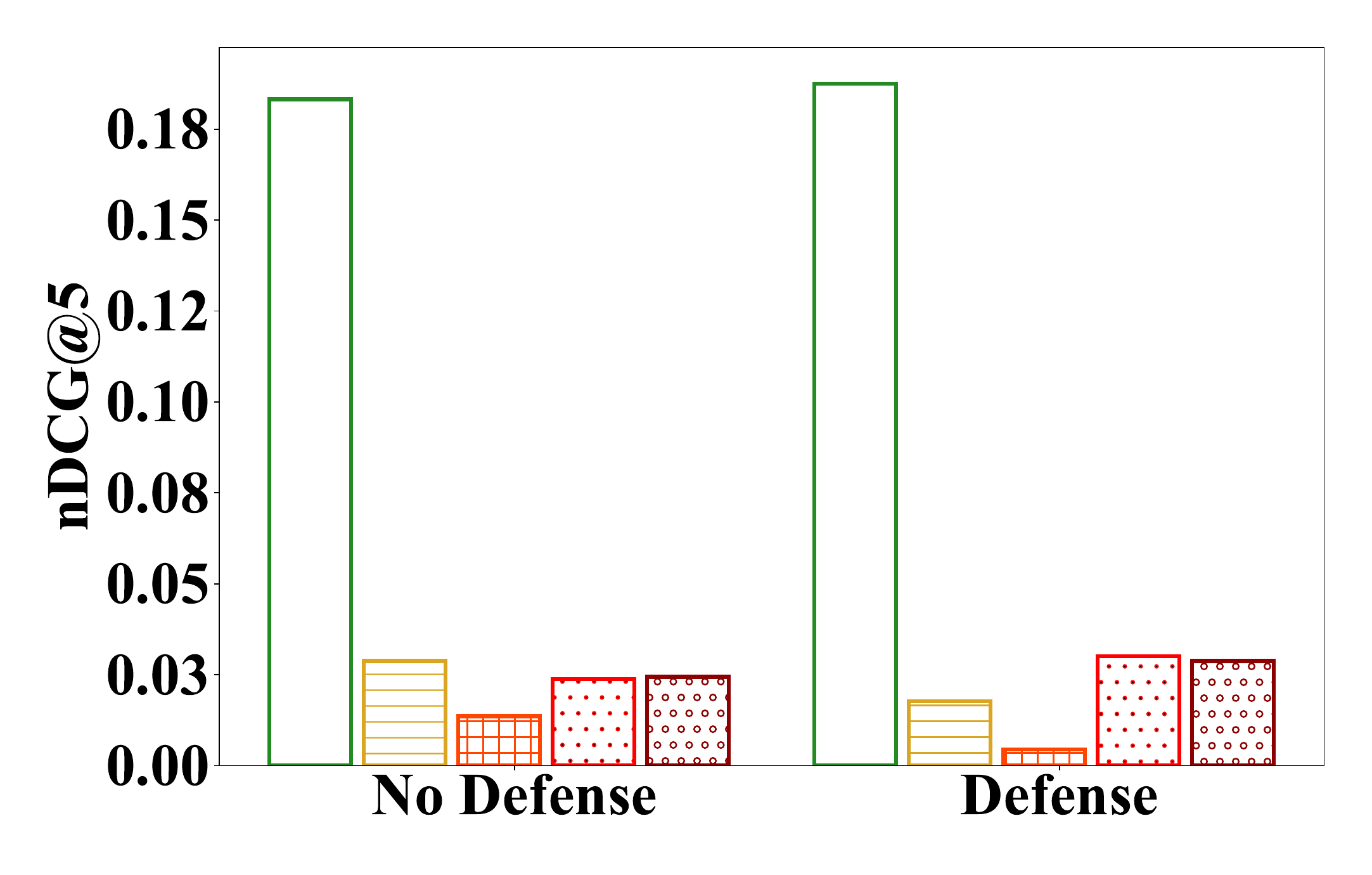}
\vspace{-22pt}
\caption*{\small (b) ML100K nDCG@5}
\endminipage\hfill
\minipage{0.23\textwidth}
\centering
\includegraphics[width=\textwidth]{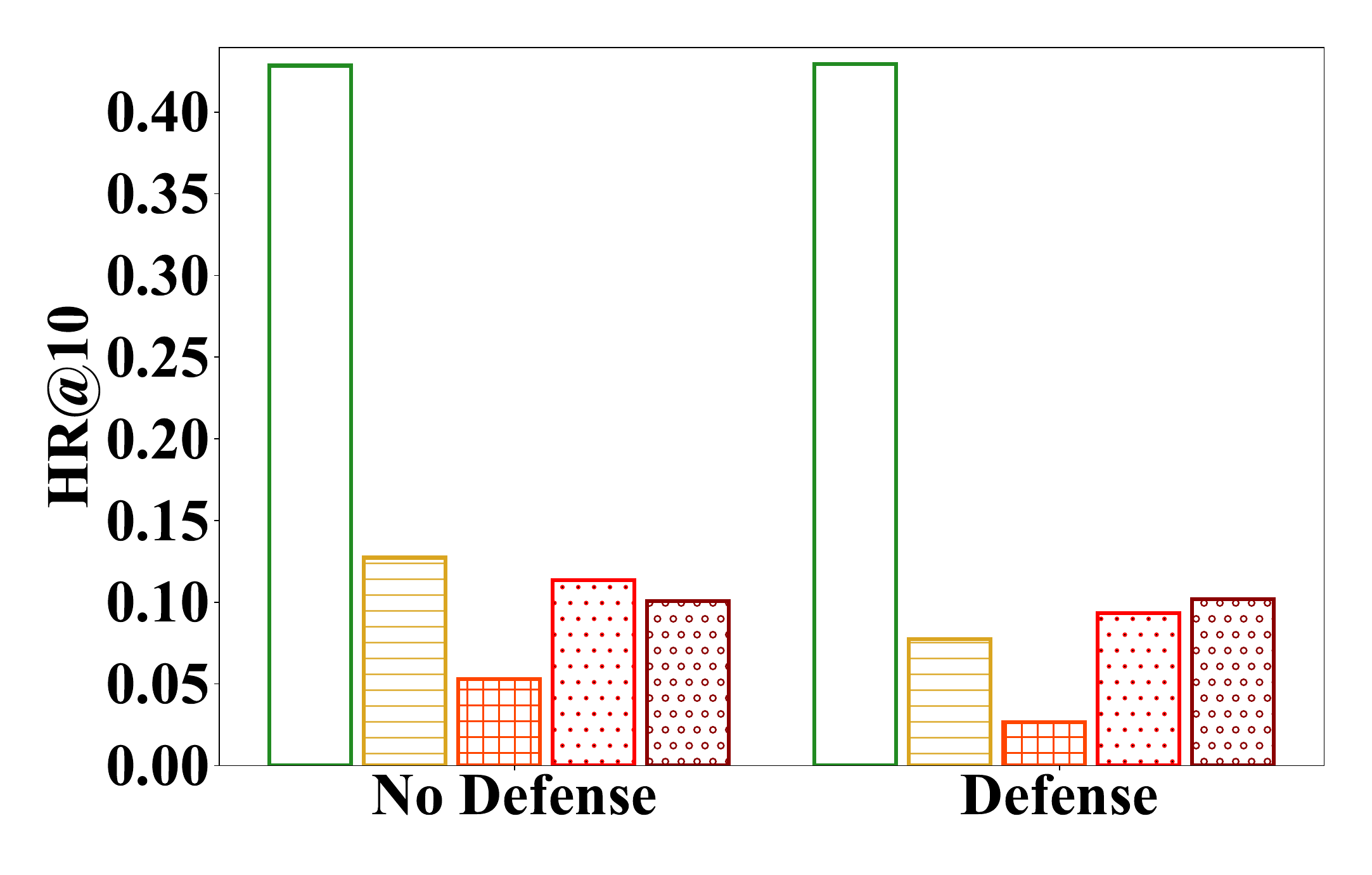}
\vspace{-22pt}
\caption*{\small (c) ML100K HR@10}
\endminipage\hfill
\minipage{0.23\textwidth}
\centering
\includegraphics[width=\textwidth]{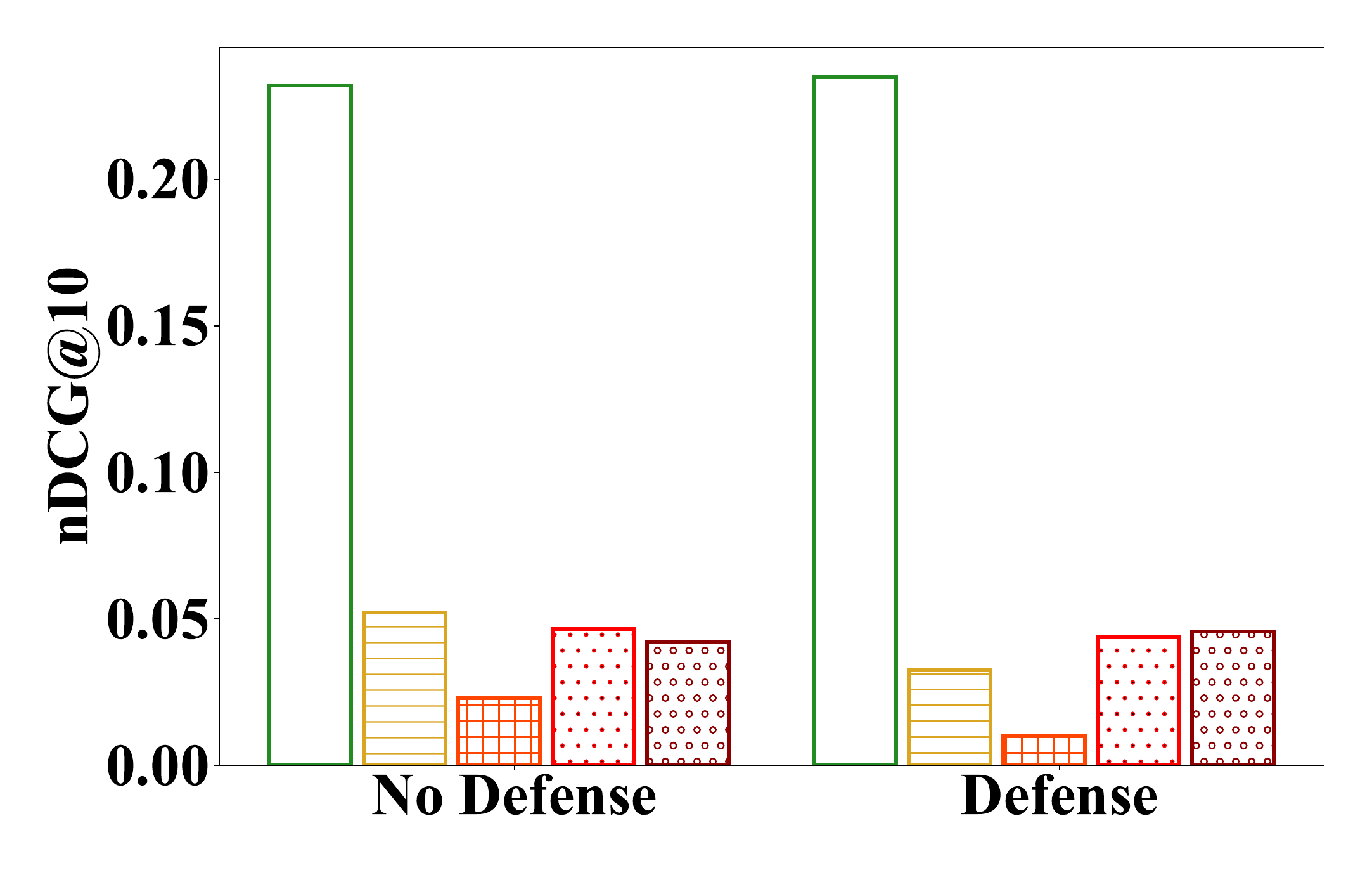}
\vspace{-22pt}
\caption*{\small (d) ML100K nDCG@10}
\endminipage\hfill

\minipage{0.23\textwidth}
\centering
\includegraphics[width=\textwidth]{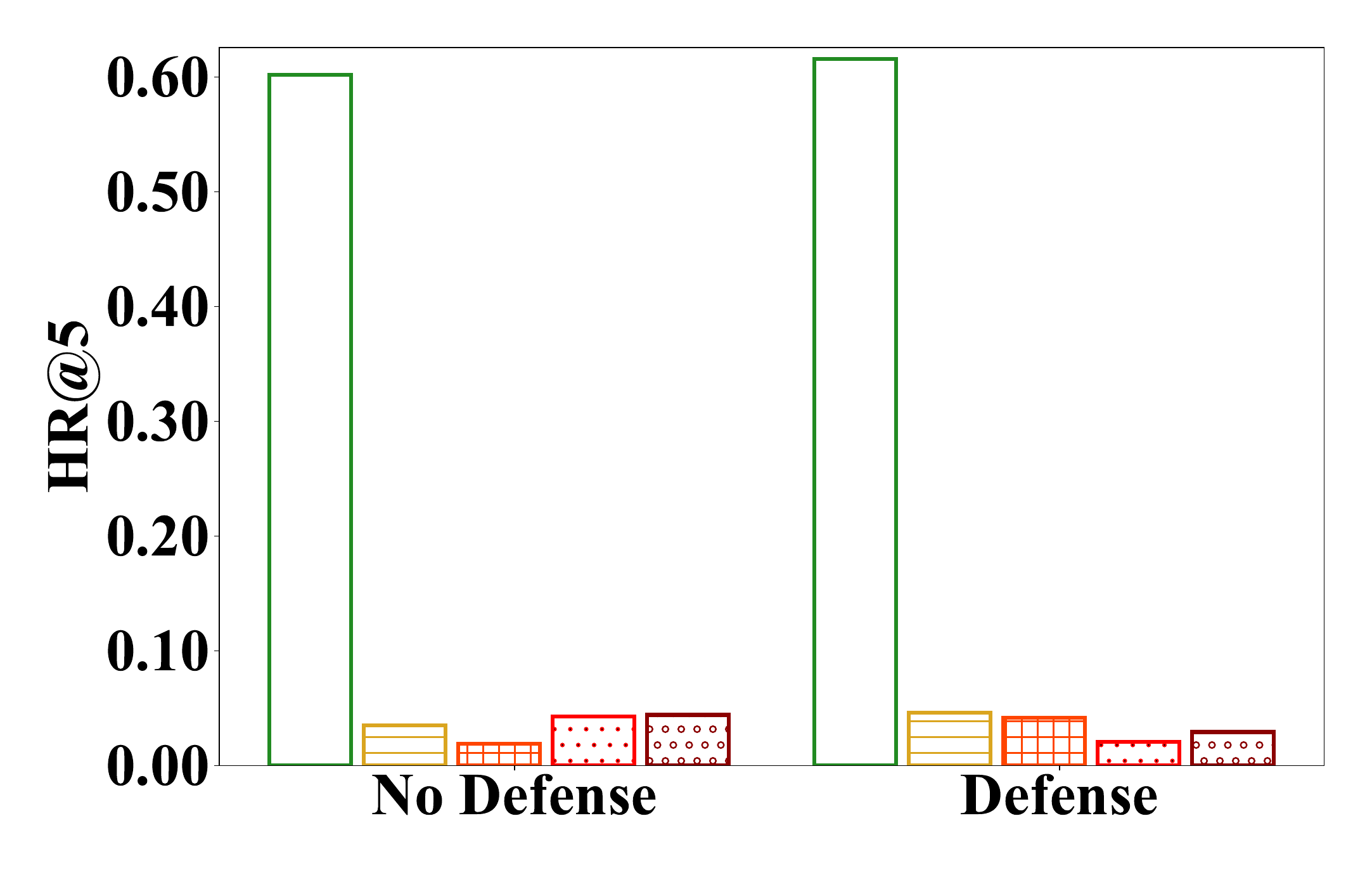}
\vspace{-22pt}
\caption*{\small (e) Steam HR@5}
\endminipage\hfill
\minipage{0.23\textwidth}
\centering
\includegraphics[width=\textwidth]{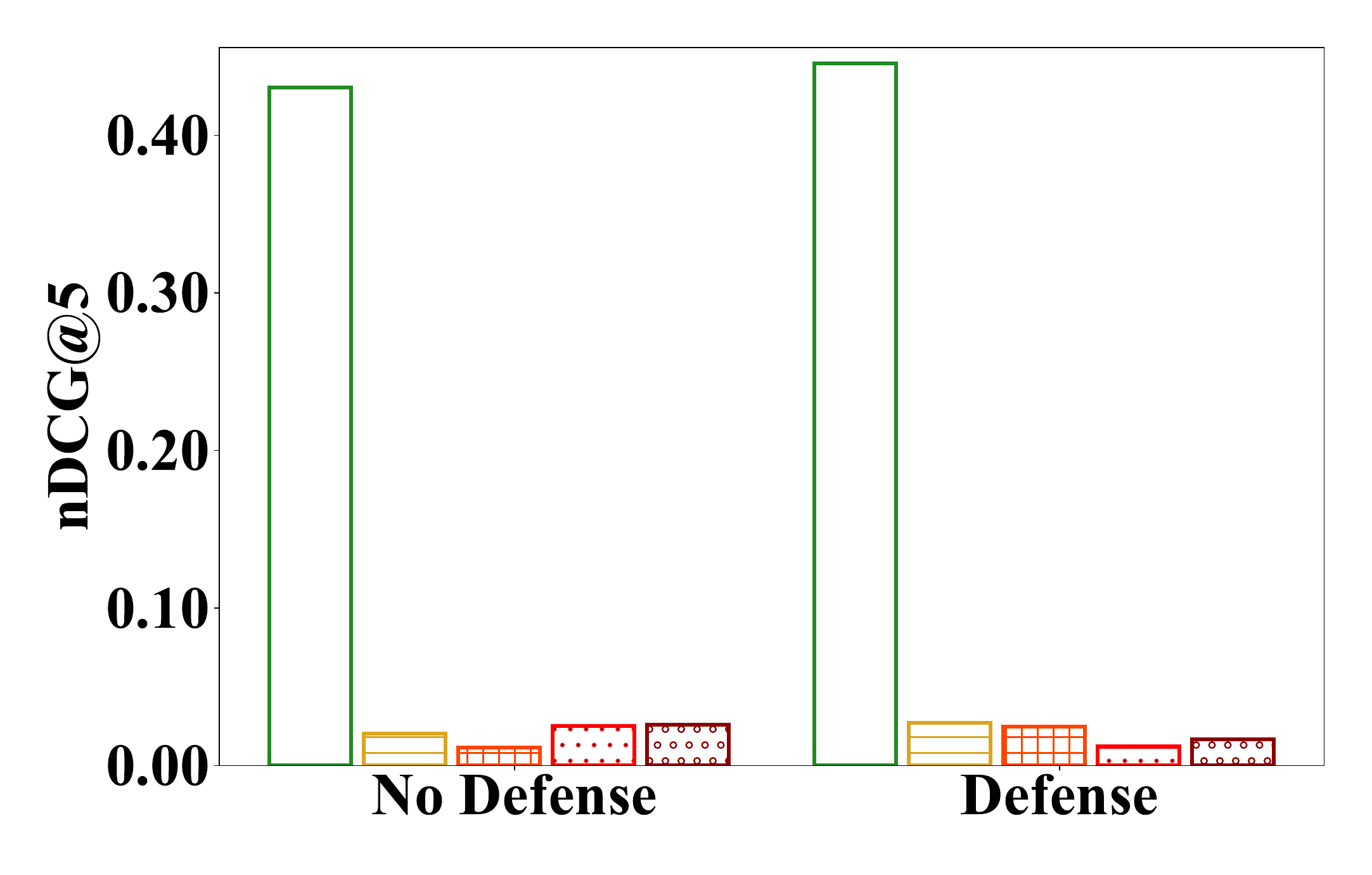}
\vspace{-22pt}
\caption*{\small (f) Steam nDCG@5}
\endminipage\hfill
\minipage{0.23\textwidth}
\centering
\includegraphics[width=\textwidth]{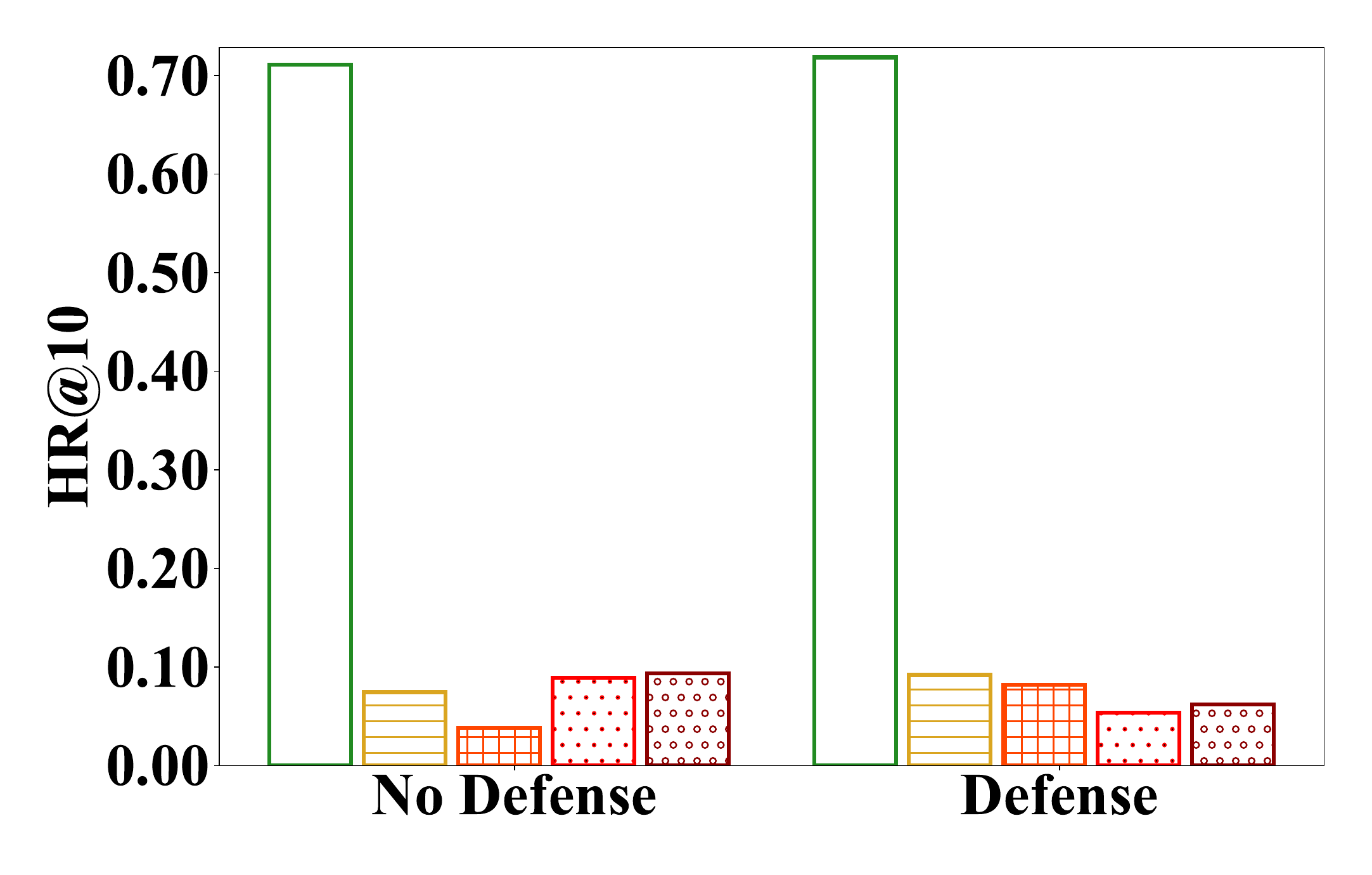}
\vspace{-22pt}
\caption*{\small (g) Steam HR@10}
\endminipage\hfill
\minipage{0.23\textwidth}
\centering
\includegraphics[width=\textwidth]{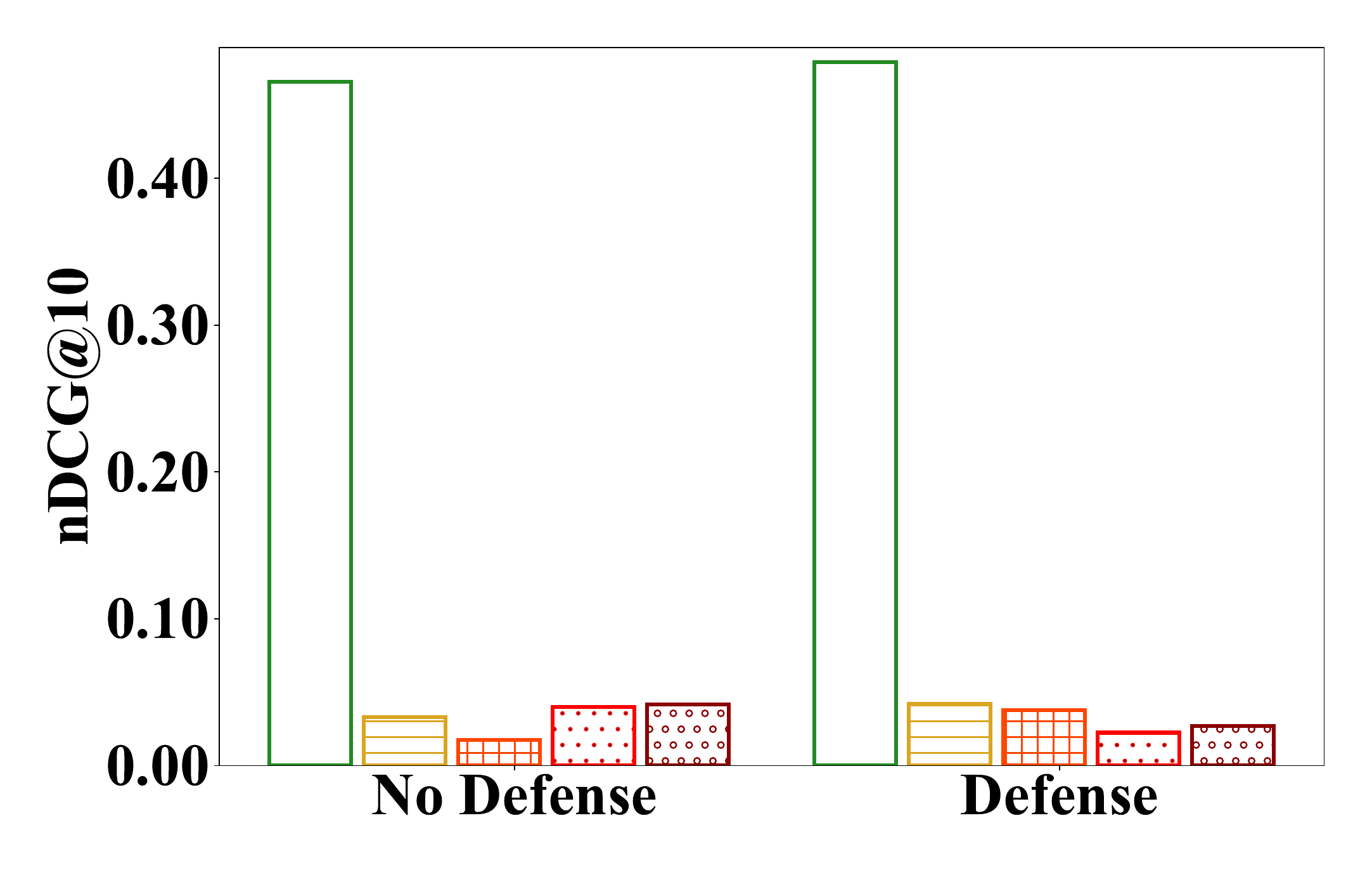}
\vspace{-22pt}
\caption*{\small (h) Steam nDCG@10}
\endminipage\hfill

\minipage{0.23\textwidth}
\centering
\includegraphics[width=\textwidth]{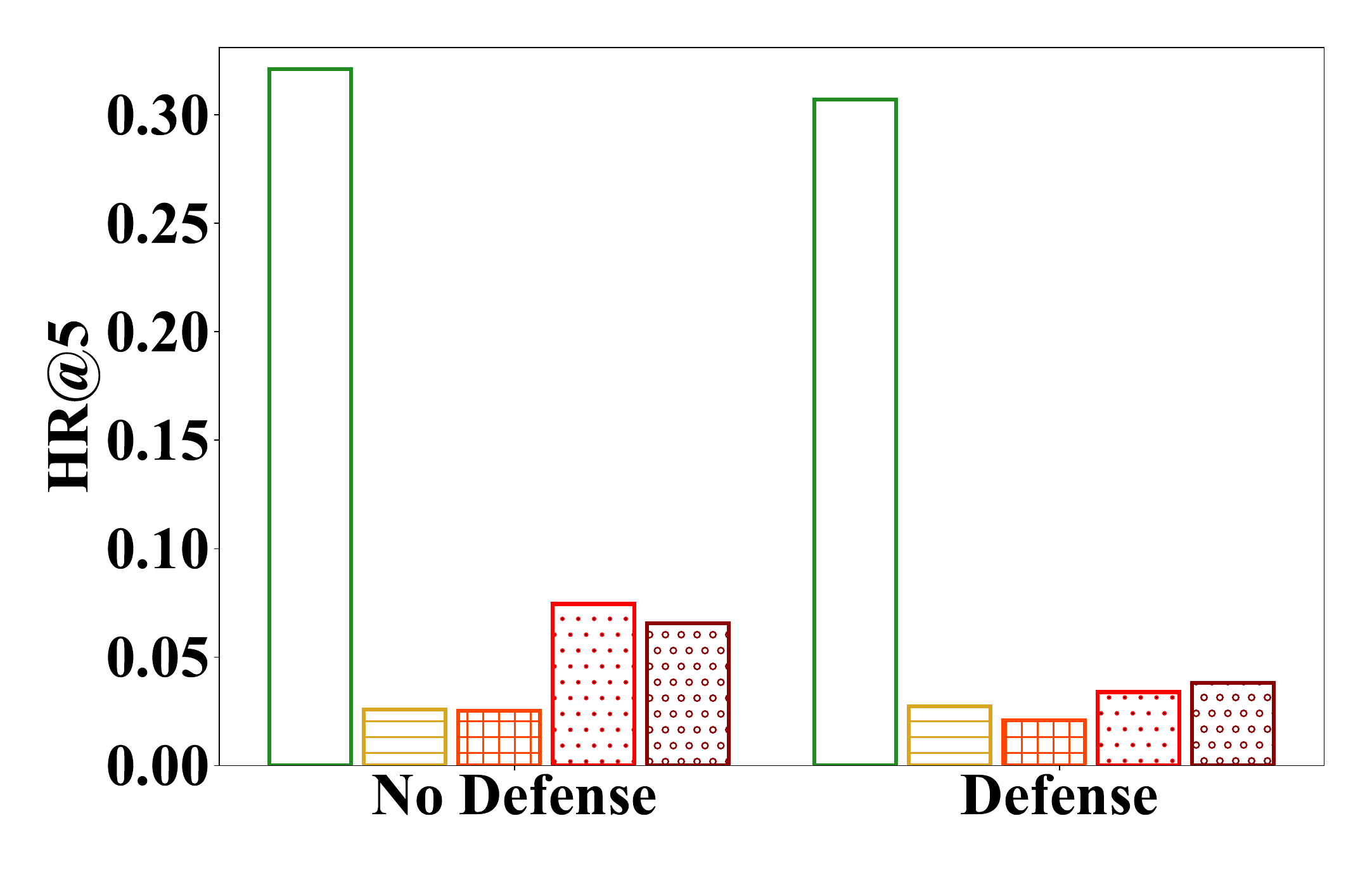}
\vspace{-22pt}
\caption*{\small (i) ML1M HR@5}
\endminipage\hfill
\minipage{0.23\textwidth}
\centering
\includegraphics[width=\textwidth]{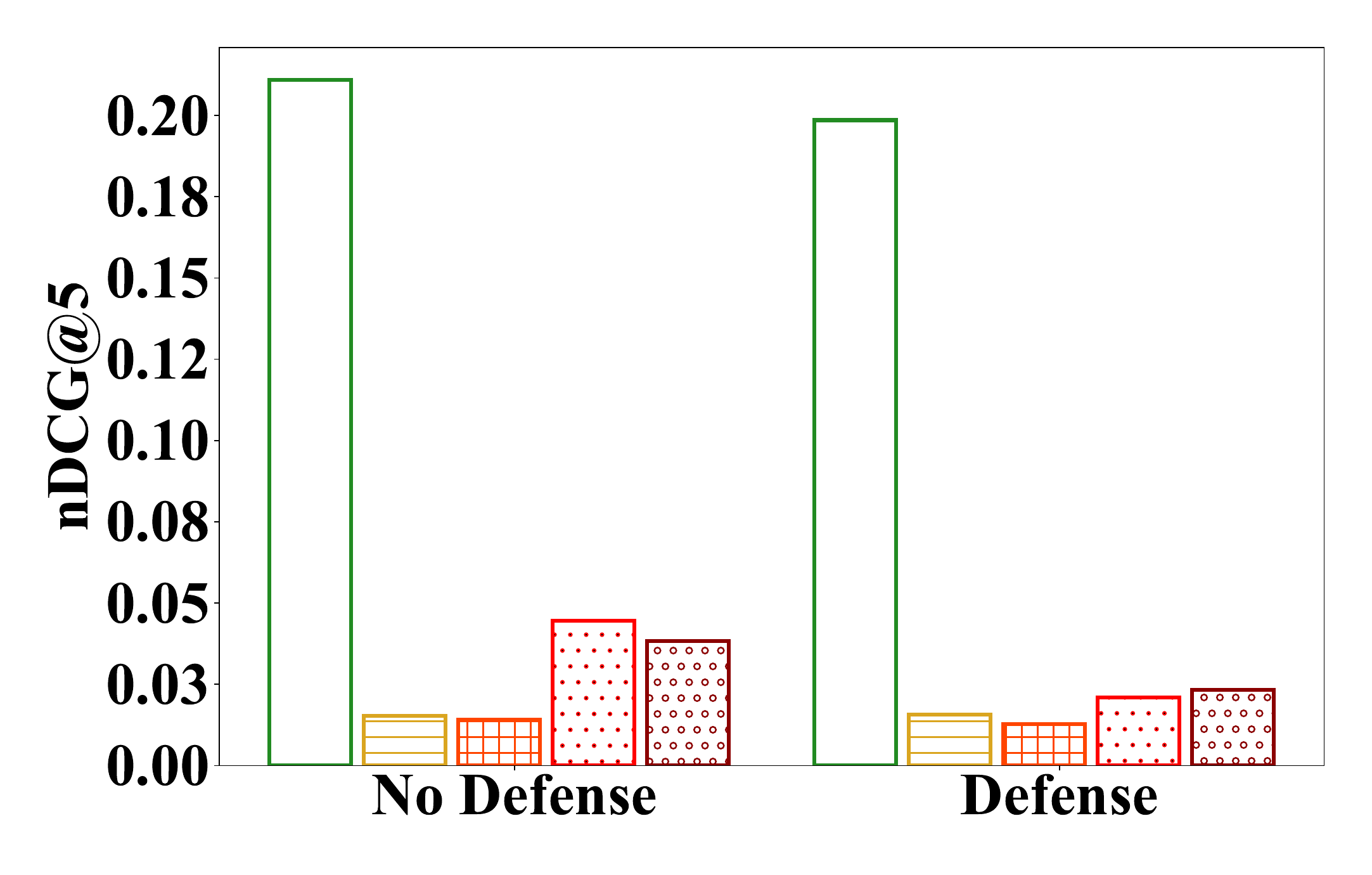}
\vspace{-22pt}
\caption*{\small (j) ML1M nDCG@5}
\endminipage\hfill
\minipage{0.23\textwidth}
\centering
\includegraphics[width=\textwidth]{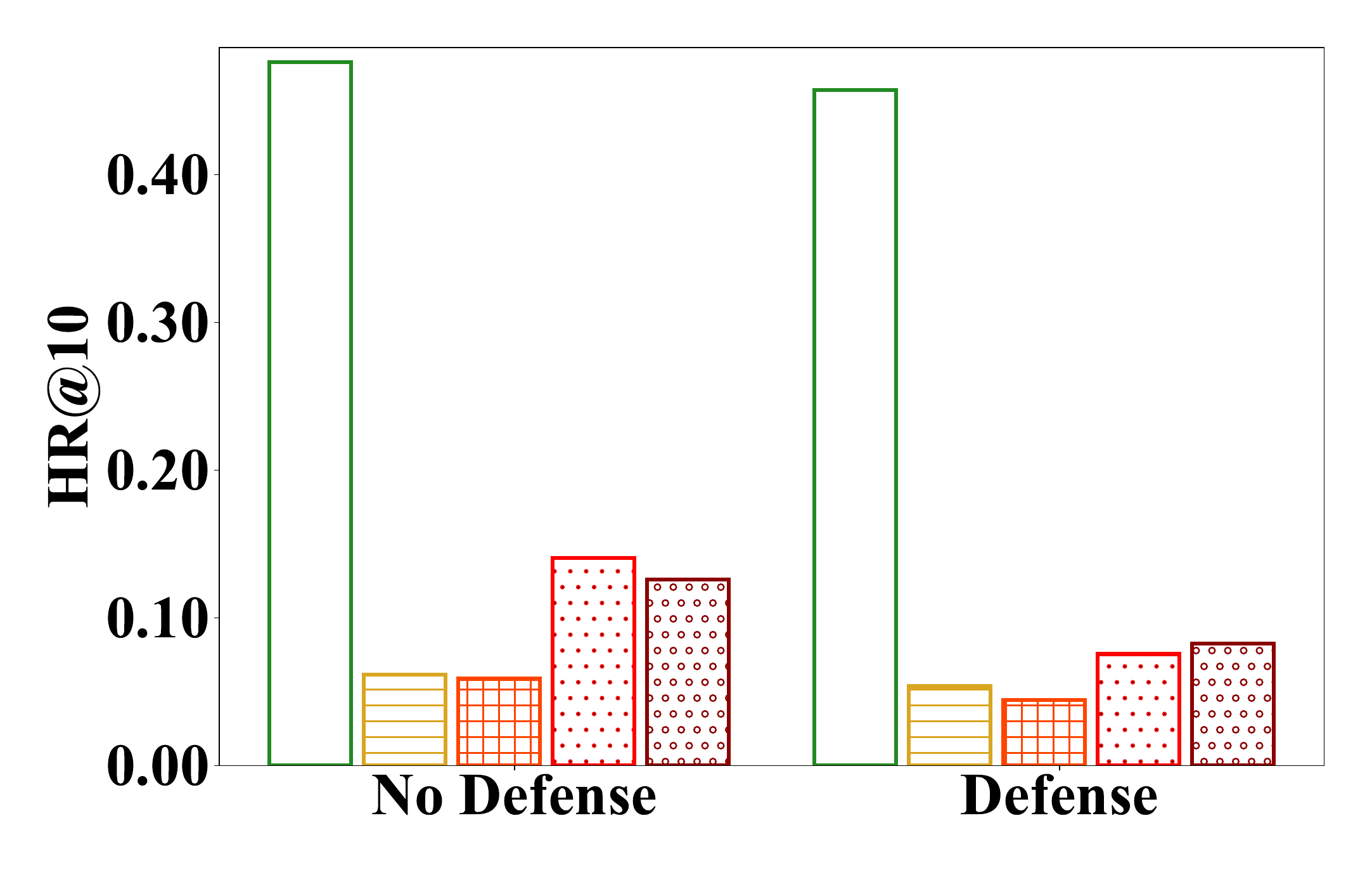}
\vspace{-22pt}
\caption*{\small (k) ML1M HR@10}
\endminipage\hfill
\minipage{0.23\textwidth}
\centering
\includegraphics[width=\textwidth]{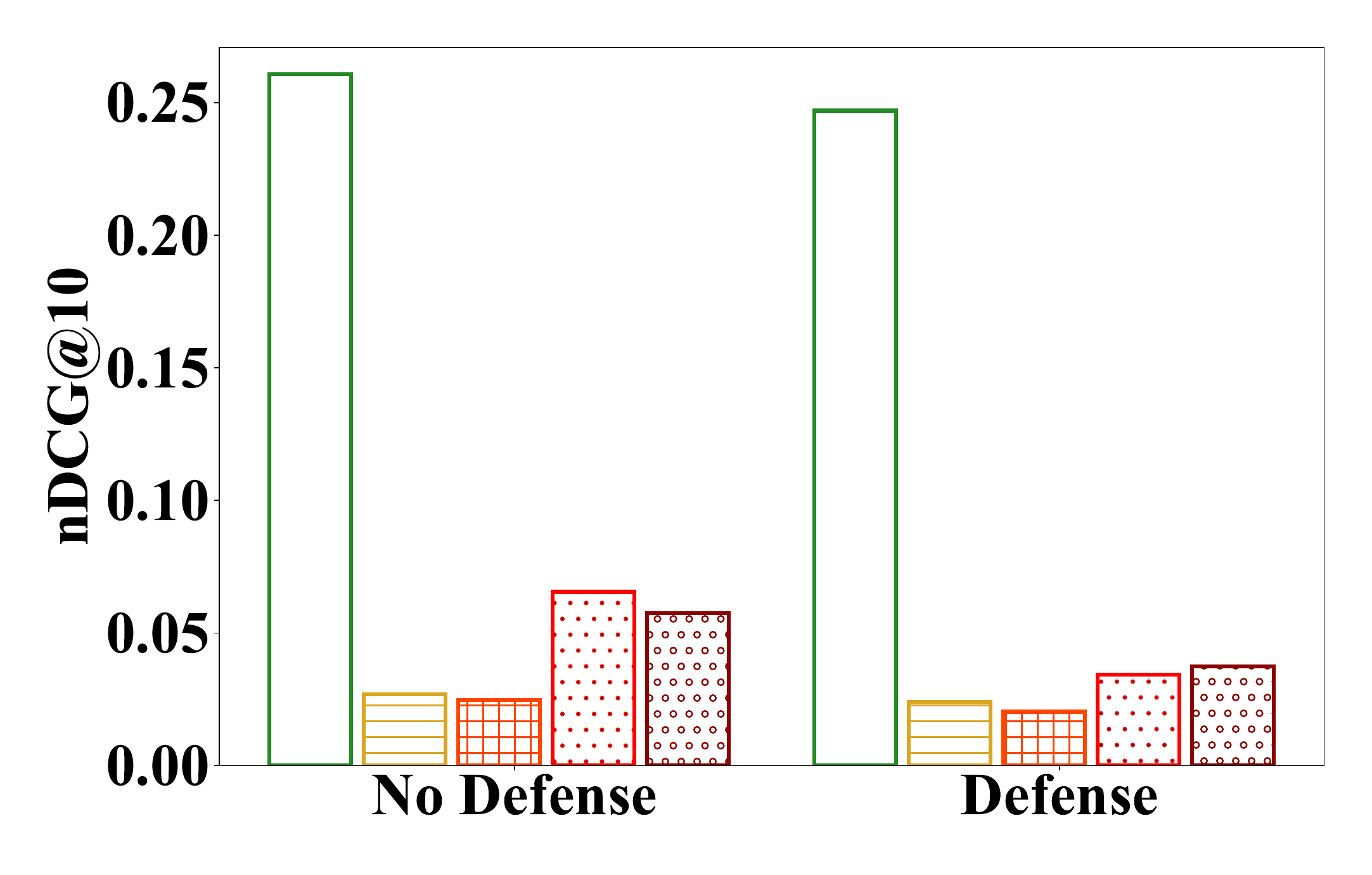}
\vspace{-22pt}
\caption*{\small (l) ML1M nDCG@10}
\endminipage\hfill

\vskip -0.1in
\caption{Experimental Results on FedGNN Model.}
\vskip -0.1in
\label{fig:full_fedgnn}
\end{figure*}

\begin{figure*}[t]
\minipage{0.24\textwidth}
\centering
 \subfigure[Spattack-O-D (Defense)]{\includegraphics[ width=\textwidth]{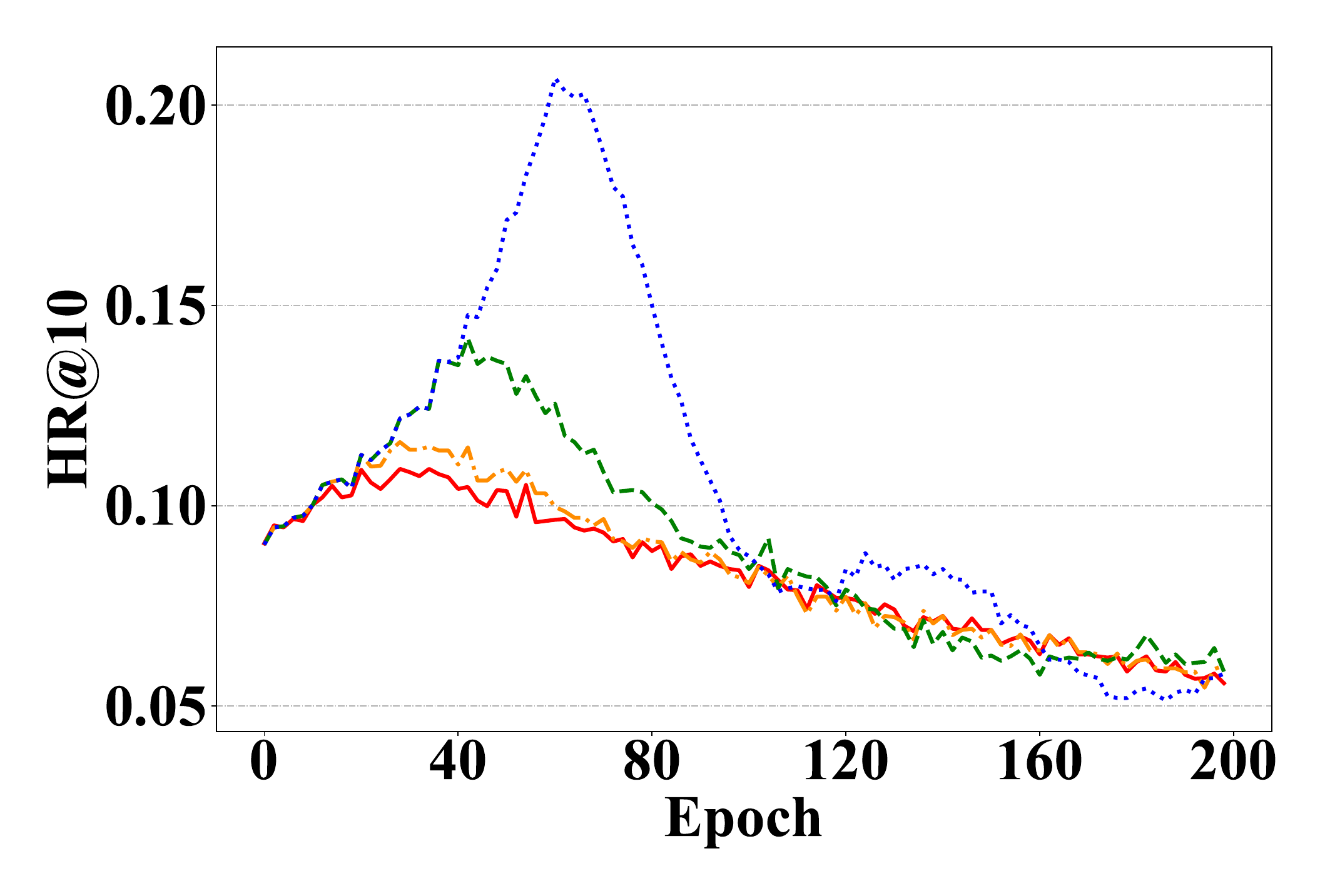}}
\endminipage\hfill
\minipage{0.24\textwidth}
\centering
 \subfigure[Spattack-O-S (Defense)]{\includegraphics[ width=\textwidth]{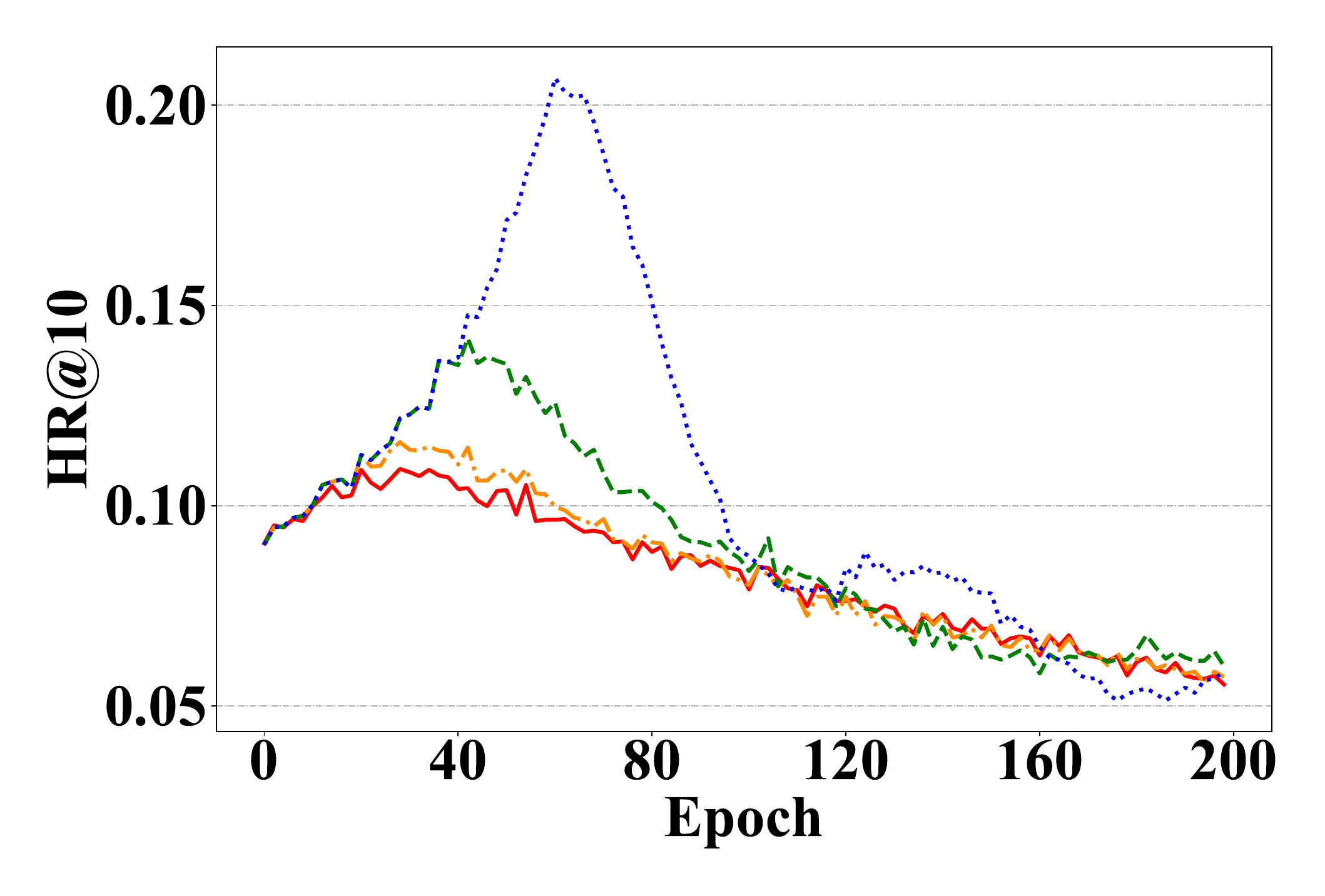}}
\endminipage\hfill
\minipage{0.24\textwidth}
\centering
 \subfigure[Spattack-L-D (Defense)]{\includegraphics[ width=\textwidth]{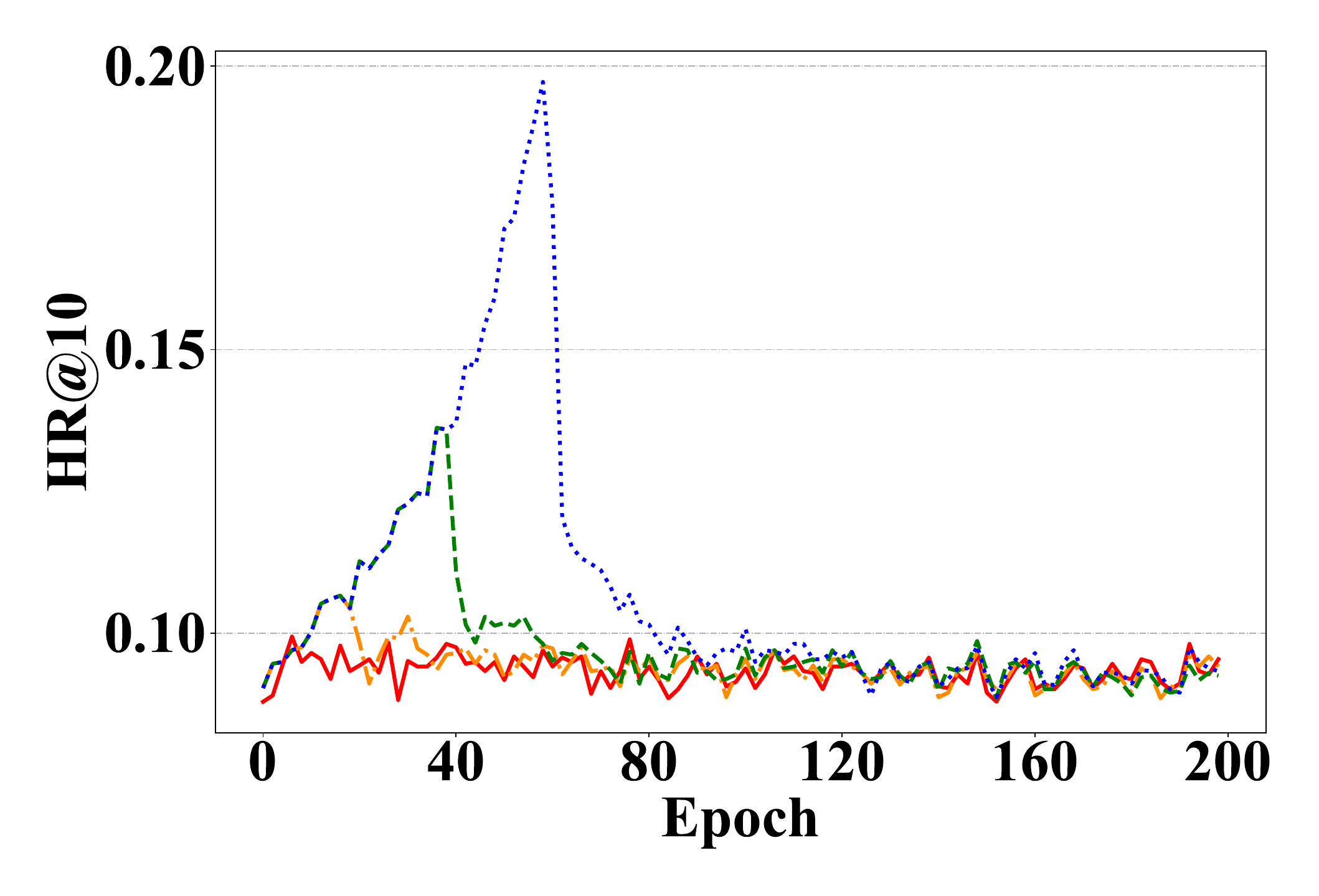}}
\endminipage\hfill
\minipage{0.24\textwidth}
\centering
 \subfigure[Spattack-L-S (Defense)]{ \includegraphics[ width=\textwidth]{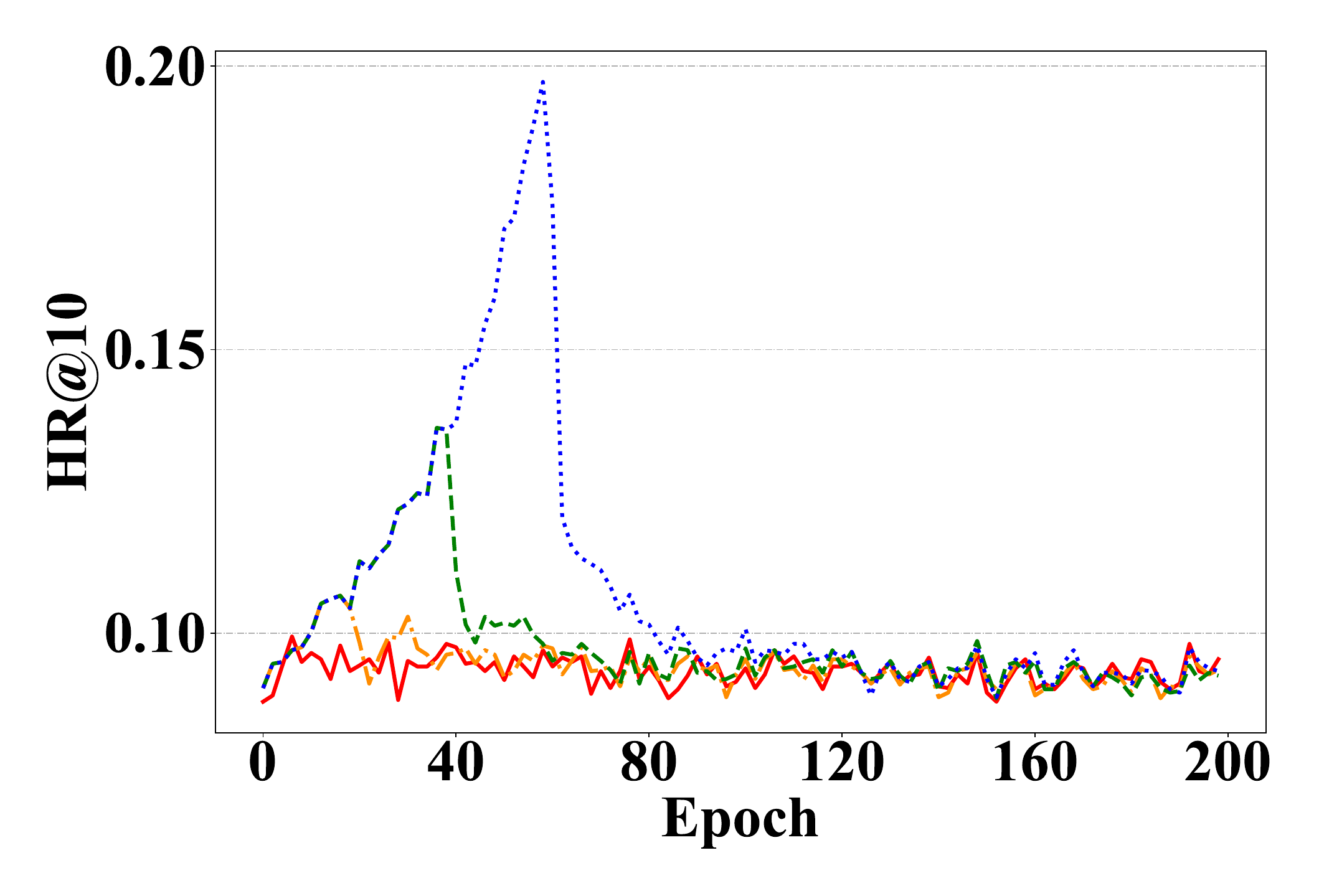}}
\endminipage\hfill
\vskip -0.1in
\caption{Performance of Spattack starting with different epochs on Steam dataset.}
\label{fig:full_startEpoch_steam}
\vskip -0.15in
\end{figure*}
\subsection{Performance Evaluation of Spattack-L-D} 
Considering non-omniscient attackers, where the benign gradients are unavailable, we launch the Spattack-L-D by randomly generating Gaussian noise as malicious gradients. To get comparable results, the ratios of malicious clients in Spattack-L-D are set to higher values of $\{5\%, 10\%, 15\%\}$.
As seen in Tab.~\ref{tab:main_type2_D}, though equipped with statically robust aggregators, Spattack-L-D can still prevent the convergence with degradation of 81\% to 97\% under a 10\% ratio. The results also provide a valuable implication that hiding the gradients of benign clients cannot protect the federated recommender well, because the attackers can break down the model by using random noise as a substitution. Besides, one can observe that Norm aggregators can provide better defense than others. The reason is that the malicious gradients of Spattack-L-D are from a Gaussian noise with the same variance: a large variance will benefit skewing data under statically robust aggregators but can be easily clipped by the norm-based defense. 

\subsection{Performance Evaluation of Spattack-L-S}
When both knowledge and capability are limited, Spattack-L-S still significantly degrades the FR model as shown in Tab.~\ref{tab:main_type2_S}. The performance drops by about 56\%, 73\% and 78\% under 5\%, 10\%, 15\% malicious ratio on average, indicating that the FR system is vulnerable to our attacks, which could hinder its applicability in various domains. 

\subsection{More Results on the Transferability of Attacks}
Here, we evaluate the effectiveness of Spattack on more FR scenarios. First, we perform Spattack with a malicious ratio of 10\% to the SOTA FedGNN~\cite{Wu2021FedGNN}, where the Byzantine clients can only upload malicious gradients of item embeddings. No Defense and Defense correspond to mean and median aggregators, respectively. As shown in Fig.~\ref{fig:full_fedgnn}, we find that FedGNN's performance dramatically drops under Spattack, demonstrating the common vulnerability of FedMF and FedGNN. Even though GNN's parameters are densely aggregated, attackers can still prevent model convergence by only poisoning item embeddings. Moreover, we also show the effectiveness of Spattack under the Adam optimizer in Fig.~\ref{fig:full_adam}. Lastly, we evaluate Spattack's effectiveness when differential privacy is applied to the local gradients in ~\cite{Wu2021FedGNN}. As shown in Fig.~\ref{fig:dp}, Spattack still achieves successful attacks. 
\begin{figure}[t]
\minipage{0.23\textwidth}
\centering
 \subfigure[FedMF.]{\includegraphics[ width=\textwidth]{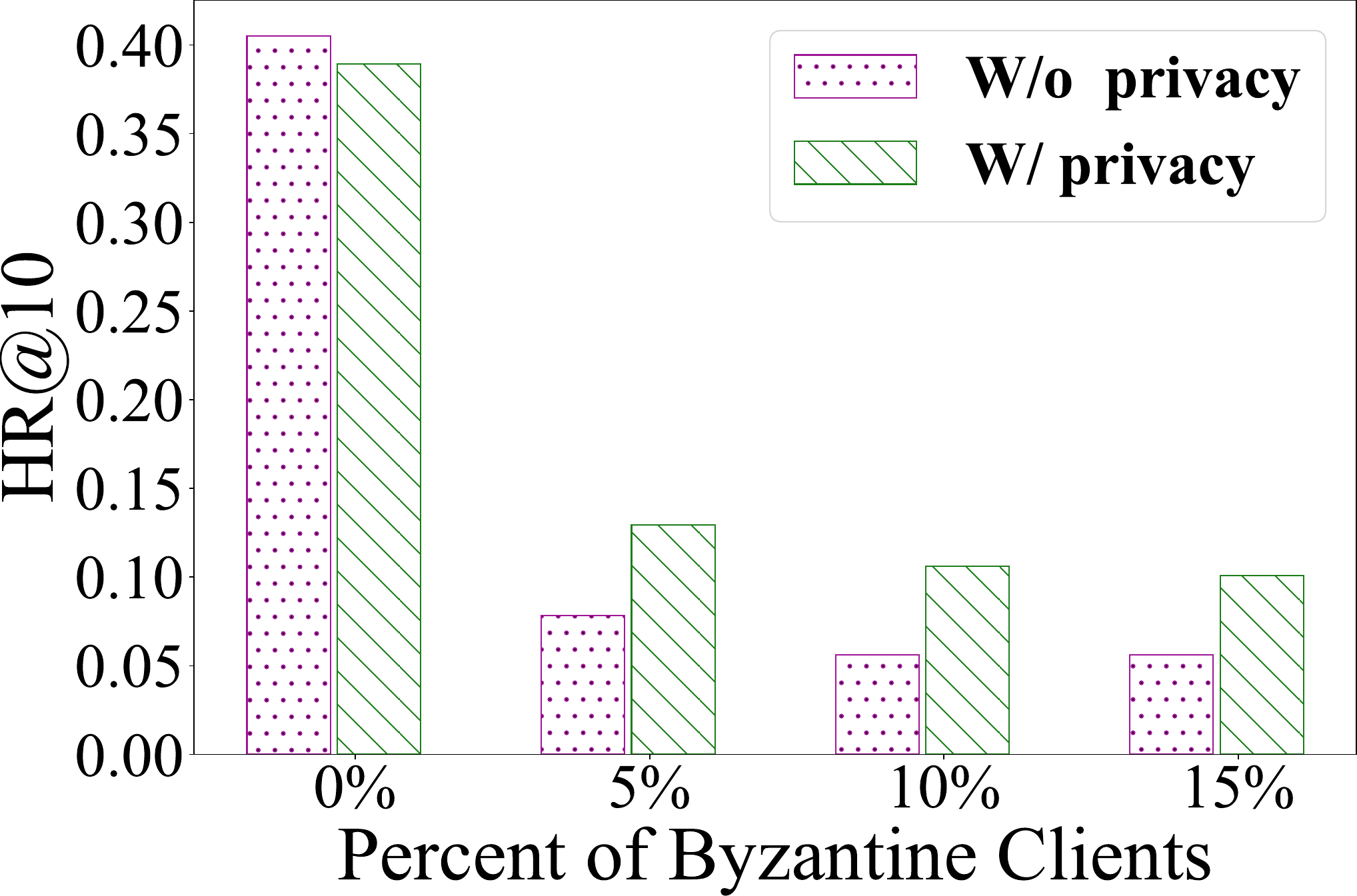}}
\endminipage\hfill
\minipage{0.23\textwidth}
\centering
 \subfigure[FedGNN.]{\includegraphics[ width=\textwidth]{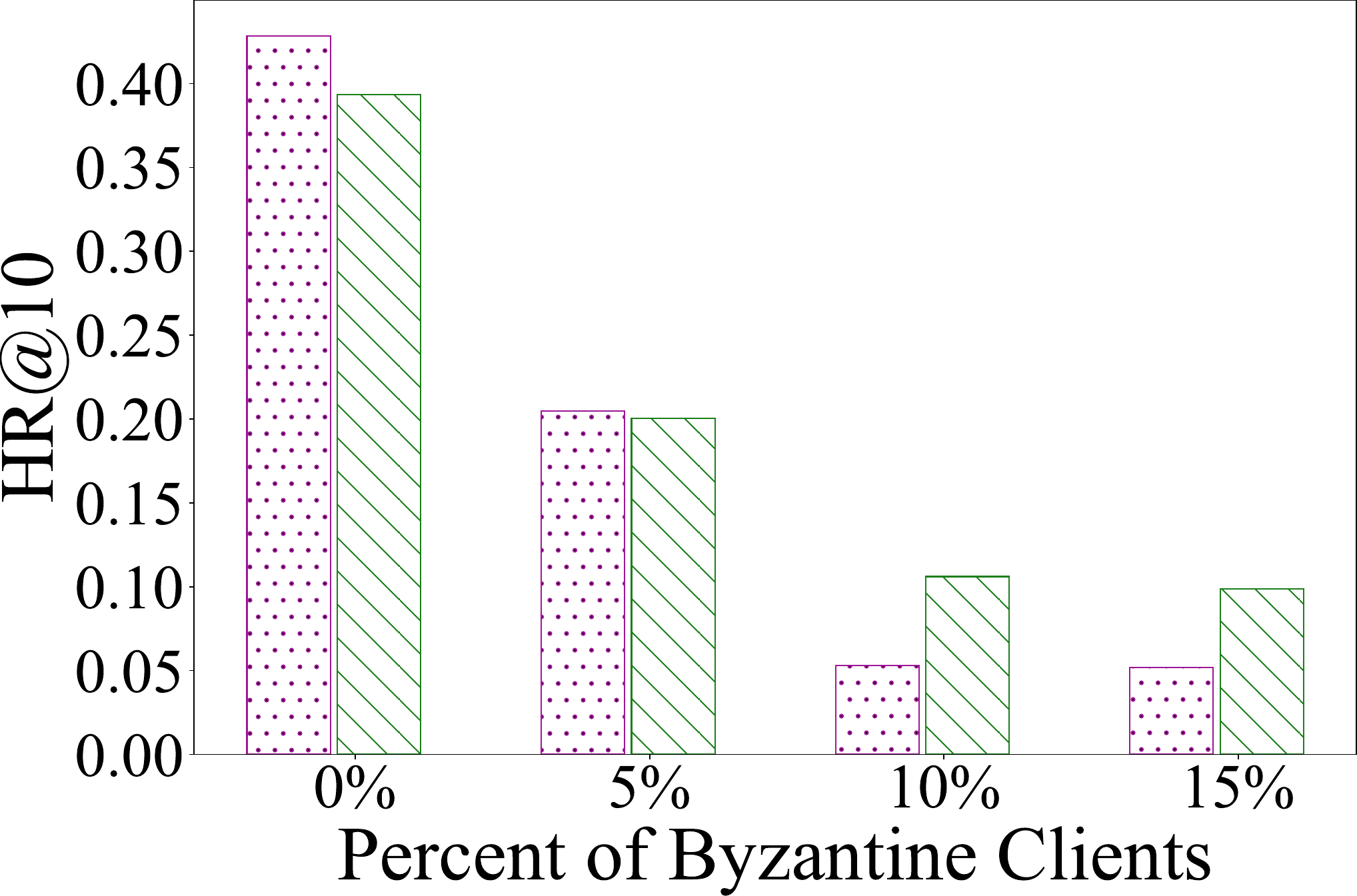}}
\endminipage\hfill
\vskip -0.1in
\caption{Attack performance under differential privacy.}
\label{fig:dp}
\vskip -0.15in
\end{figure}
\subsection{More results on the Hyperparameter Analysis}
In Fig.~\ref{fig:full_startEpoch_steam}, we present the convergence of FR under defense (i.e., Median AGR) against Spattack with a malicious ratio of 10\%. The results correspond to starting attacks at 0, 20, 40, and 60, with a visualization of HR@10 in 200 epochs on Steam. We observe that the Median AGR directly picks the malicious gradient as output and rapidly decreases the model performance along the opposite direction of the true gradient, which once again validates the vulnerability of sparse aggregation on FR.

\begin{table}[h]
\vskip -0.025in
\caption{The performance under lower malicious ratios.} \label{tab:lower_malicios_ratio}
\vskip -0.125in
\Large 
\setlength{\extrarowheight}{2pt}
\resizebox{\linewidth}{!}{
\begin{tabular}{c|c|c|c|c|c|c|c}
\hline
$\text{AGR}(\cdot)$              & Spattack       & Clean      & 1 (0.1\%)        & 5 (0.5\%)        & 9 (1\%)          & 29 (3\%)         & 49 (5\%)         \\ \hline
\multirow{2}{*}{Mean}   & O-D & \begin{tabular}[c]{@{}c@{}}0.4051\end{tabular} & \begin{tabular}[c]{@{}c@{}}0.0997\\(-75.4\%)\end{tabular} & \begin{tabular}[c]{@{}c@{}}0.1029\\(-74.6\%)\end{tabular} & \begin{tabular}[c]{@{}c@{}}0.0986\\(-75.7\%)\end{tabular} & \begin{tabular}[c]{@{}c@{}}0.0944\\(-76.7\%)\end{tabular} & \begin{tabular}[c]{@{}c@{}}\begin{tabular}[c]{@{}c@{}}0.0944\end{tabular}\\(-76.7\%)\end{tabular} \\
                        & O-S & \begin{tabular}[c]{@{}c@{}}0.4051\end{tabular} & \begin{tabular}[c]{@{}c@{}}0.2471\\(-39.0\%)\end{tabular} & \begin{tabular}[c]{@{}c@{}}0.175\\(-56.8\%)\end{tabular}  & \begin{tabular}[c]{@{}c@{}}0.1421\\(-64.9\%)\end{tabular} & \begin{tabular}[c]{@{}c@{}}0.0997\\(-75.4\%)\end{tabular} & \begin{tabular}[c]{@{}c@{}}0.1039\\(-74.4\%)\end{tabular} \\ \hline
\multirow{2}{*}{Median} & O-D & \begin{tabular}[c]{@{}c@{}}0.3849\end{tabular} & \begin{tabular}[c]{@{}c@{}}0.3818\\(-0.8\%)\end{tabular}  & \begin{tabular}[c]{@{}c@{}}0.3743\\(-2.8\%)\end{tabular}  & \begin{tabular}[c]{@{}c@{}}0.3510\\(-8.8\%)\end{tabular}  & \begin{tabular}[c]{@{}c@{}}0.2725\\(-29.2\%)\end{tabular} & \begin{tabular}[c]{@{}c@{}}0.0732\\(-81.0\%)\end{tabular} \\
                        & O-S & \begin{tabular}[c]{@{}c@{}}0.3849\end{tabular} & \begin{tabular}[c]{@{}c@{}}0.3826\\(-0.6\%)\end{tabular}  & \begin{tabular}[c]{@{}c@{}}0.3765\\(-2.2\%)\end{tabular}  & \begin{tabular}[c]{@{}c@{}}0.3606\\(-6.3\%)\end{tabular}  & \begin{tabular}[c]{@{}c@{}}0.3075\\(-20.1\%)\end{tabular} & \begin{tabular}[c]{@{}c@{}}0.2068\\(-46.3\%)\end{tabular} \\ \hline
\end{tabular}}
\vskip -0.175in
\end{table}
\subsection{The effectiveness under lower malicious ratio}
To explore the lowest malicious ratio for the effectiveness of Spattack, we consider the worst-case scenario, Spattack-O, where the attacker have the knowledge of total gradients. We report HR@10 on the ML-100K dataset under varying malicious ratios: 0.1\%, 0.5\%, 1\%, 3\%, and 5\%, corresponding to 1, 5, 9, 29, and 49 malicious clients, respectively.  The experimental results are reported in Tab.~\ref{tab:lower_malicios_ratio}. We have the following observations: \\
$\bullet$ Without defense (i.e., Mean), with only 0.1\% malicious ratio (1 malicious client), Spattack-O-D dramatically degrades the performance by over 75\% and prevents model convergence, which is consistent with Proposition 1. Even if the capability is limited, Spattack-O-S still effectively degrades the performance by 39\% with only 1 malicious client.\\
$\bullet$ For the defense aggregator (i.e., Median), a 3\% malicious ratio (29 malicious clients) can significantly degrade model performance by over 20\%. 

\begin{figure*}[ht]
\minipage{0.24\textwidth}
\centering
 \subfigure[ML100K]{\includegraphics[ width=\textwidth]{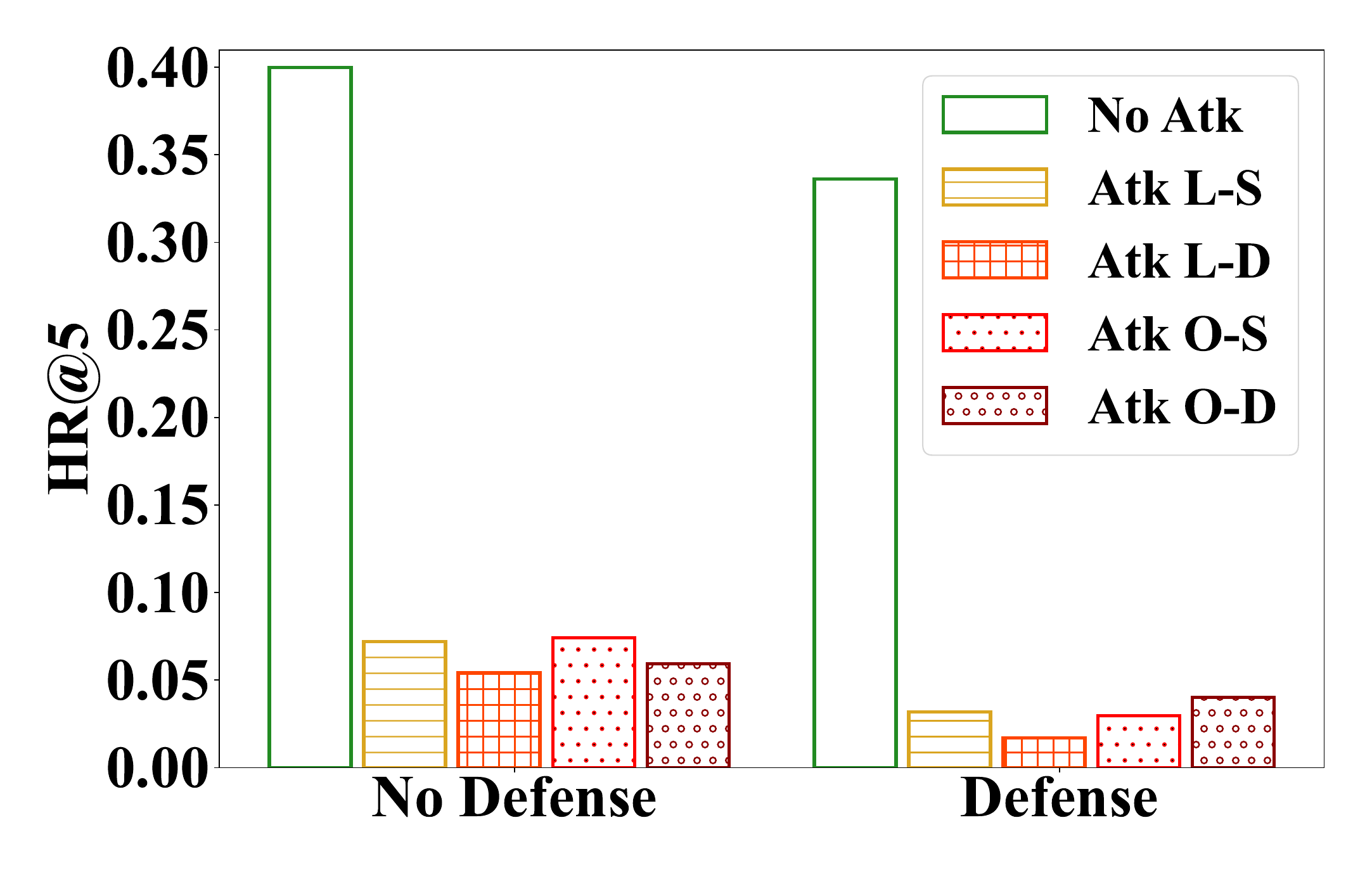}}
\endminipage\hfill
\minipage{0.24\textwidth}
\centering
 \subfigure[ML100K]{\includegraphics[ width=\textwidth]{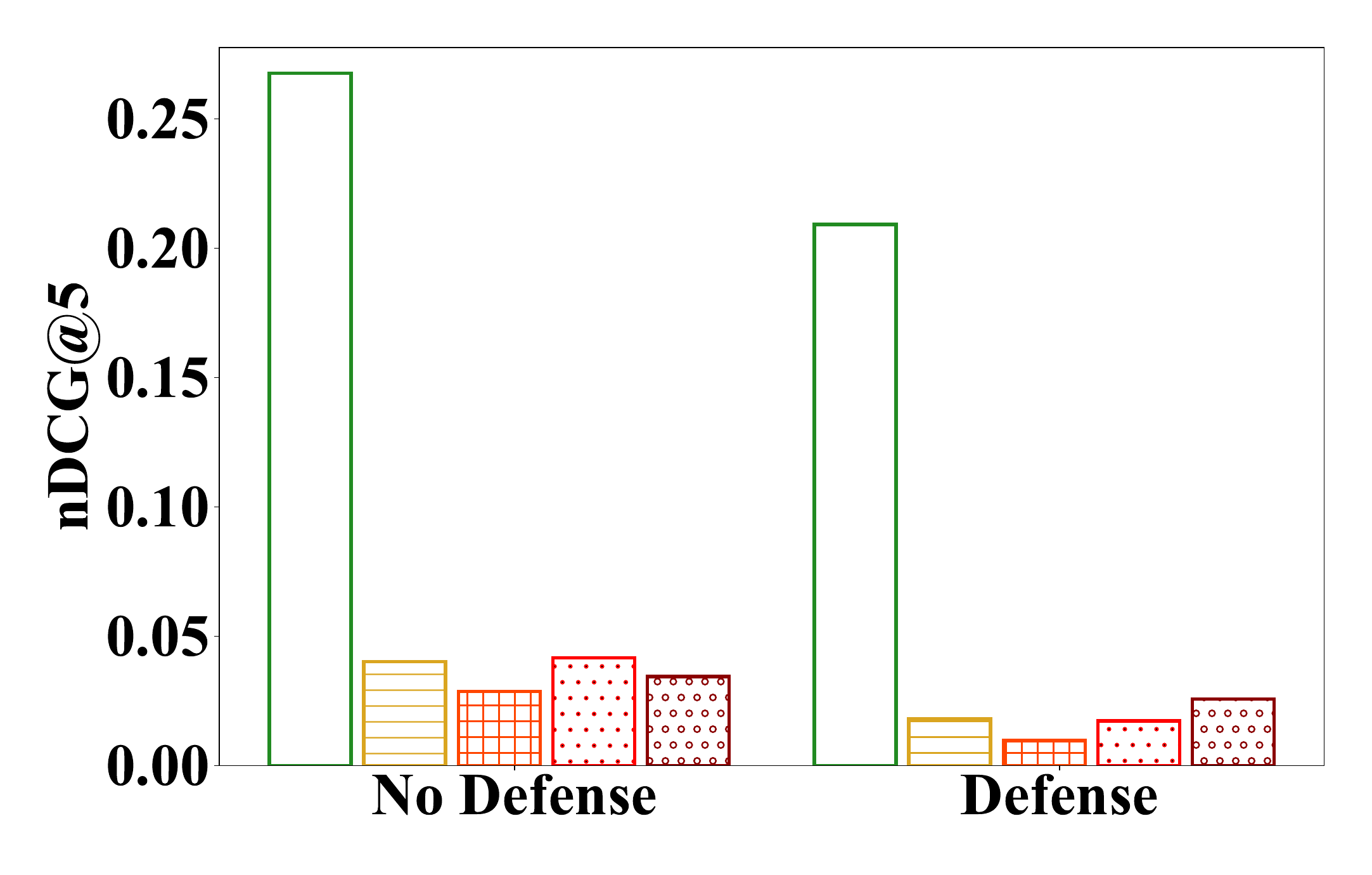}}
\endminipage\hfill
\minipage{0.24\textwidth}
\centering
 \subfigure[ML100K]{\includegraphics[ width=\textwidth]{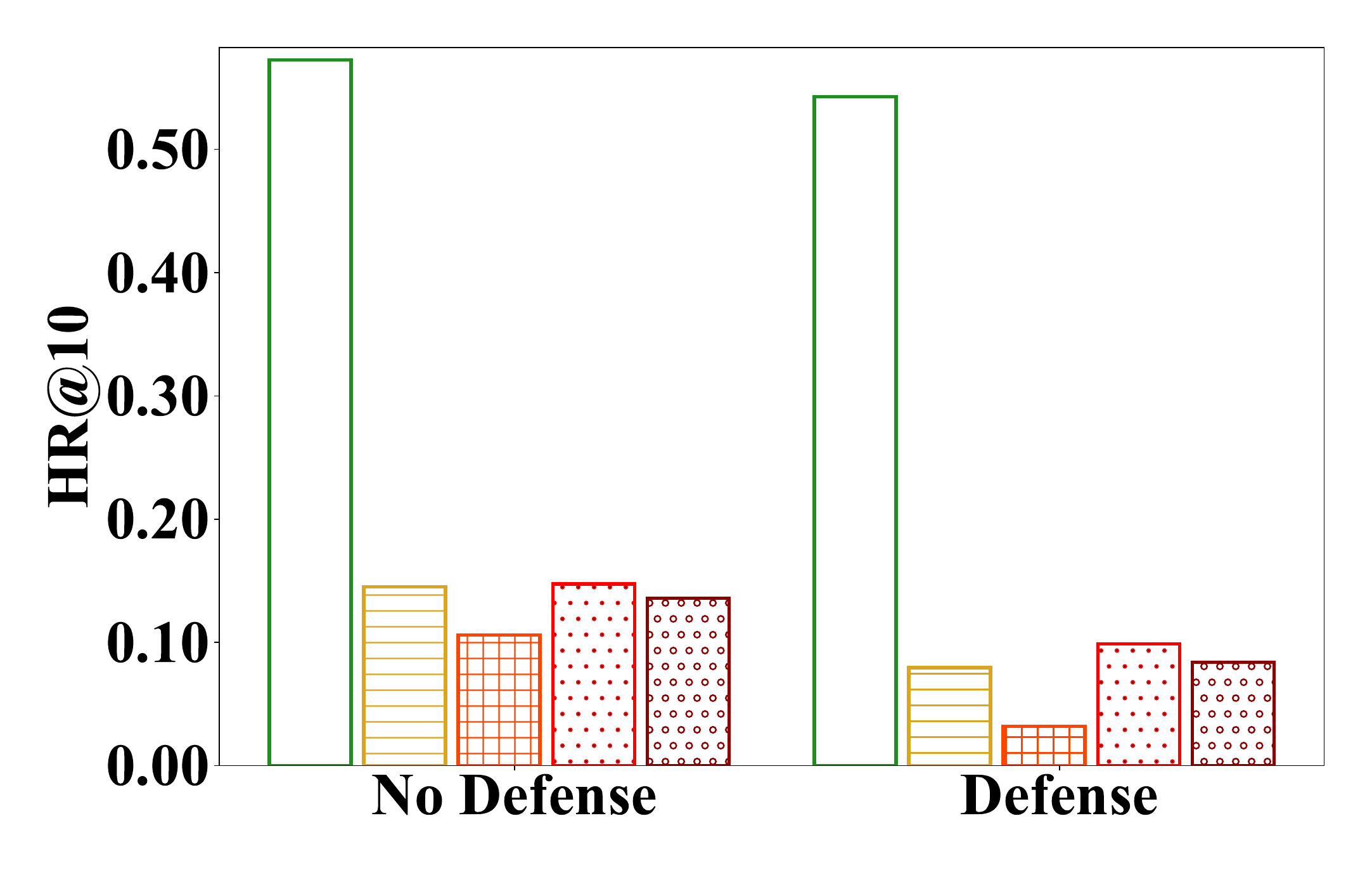}}
\endminipage\hfill
\minipage{0.24\textwidth}
\centering
 \subfigure[ML100K]{ \includegraphics[ width=\textwidth]{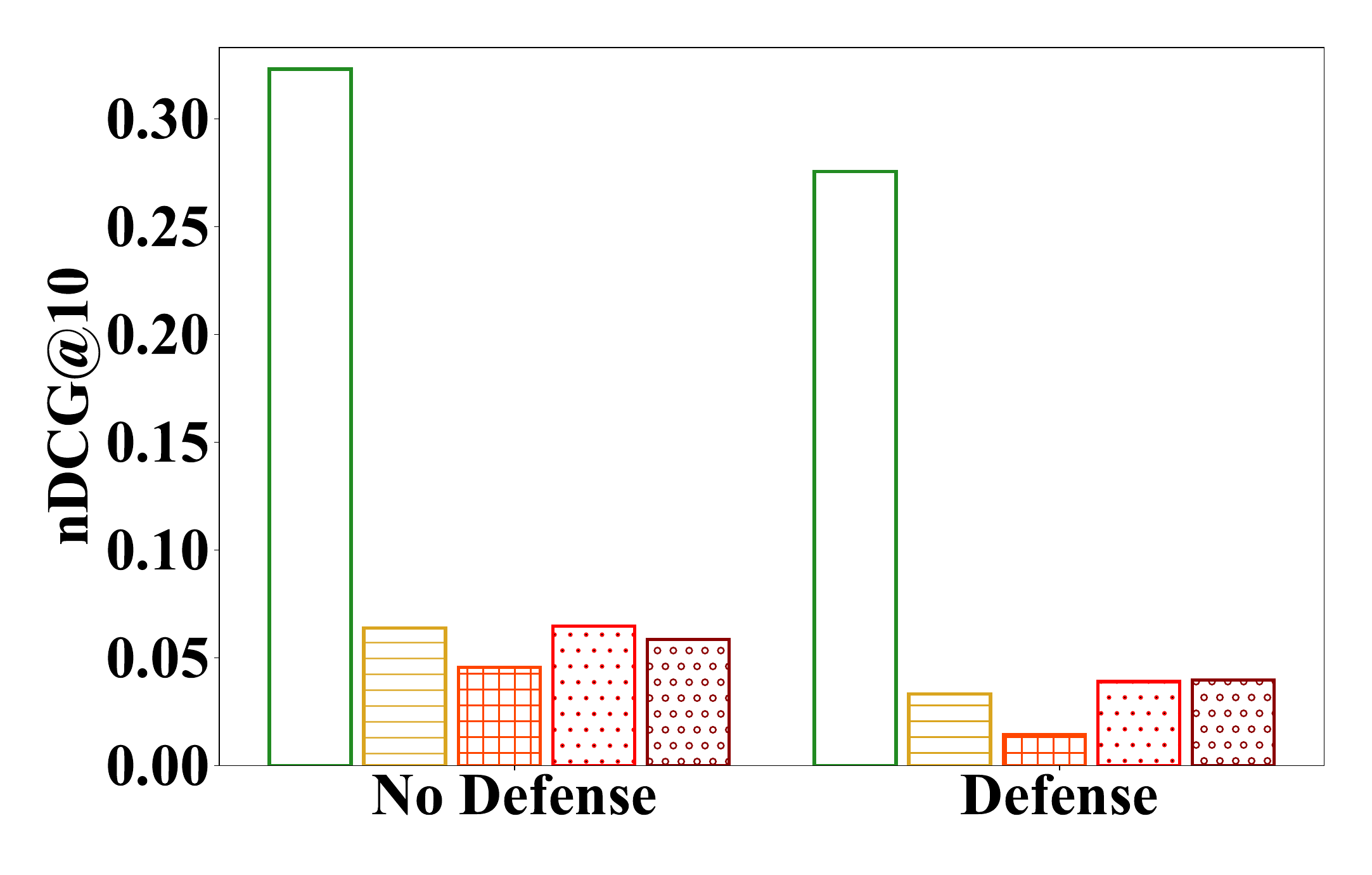}}
\endminipage\hfill

\minipage{0.24\textwidth}
\centering
 \subfigure[Steam]{\includegraphics[ width=\textwidth]{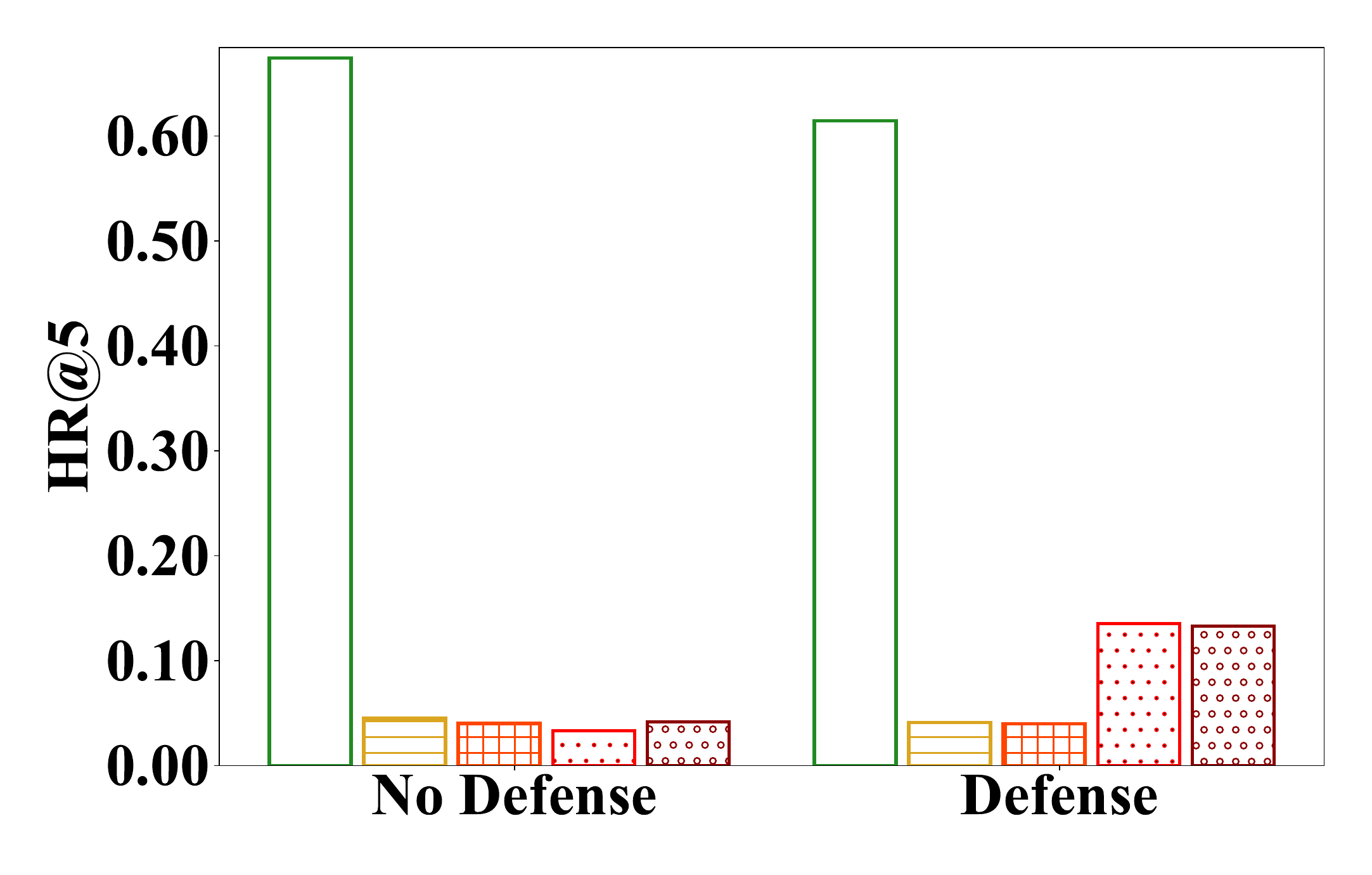}}
\endminipage\hfill
\minipage{0.24\textwidth}
\centering
 \subfigure[Steam]{\includegraphics[ width=\textwidth]{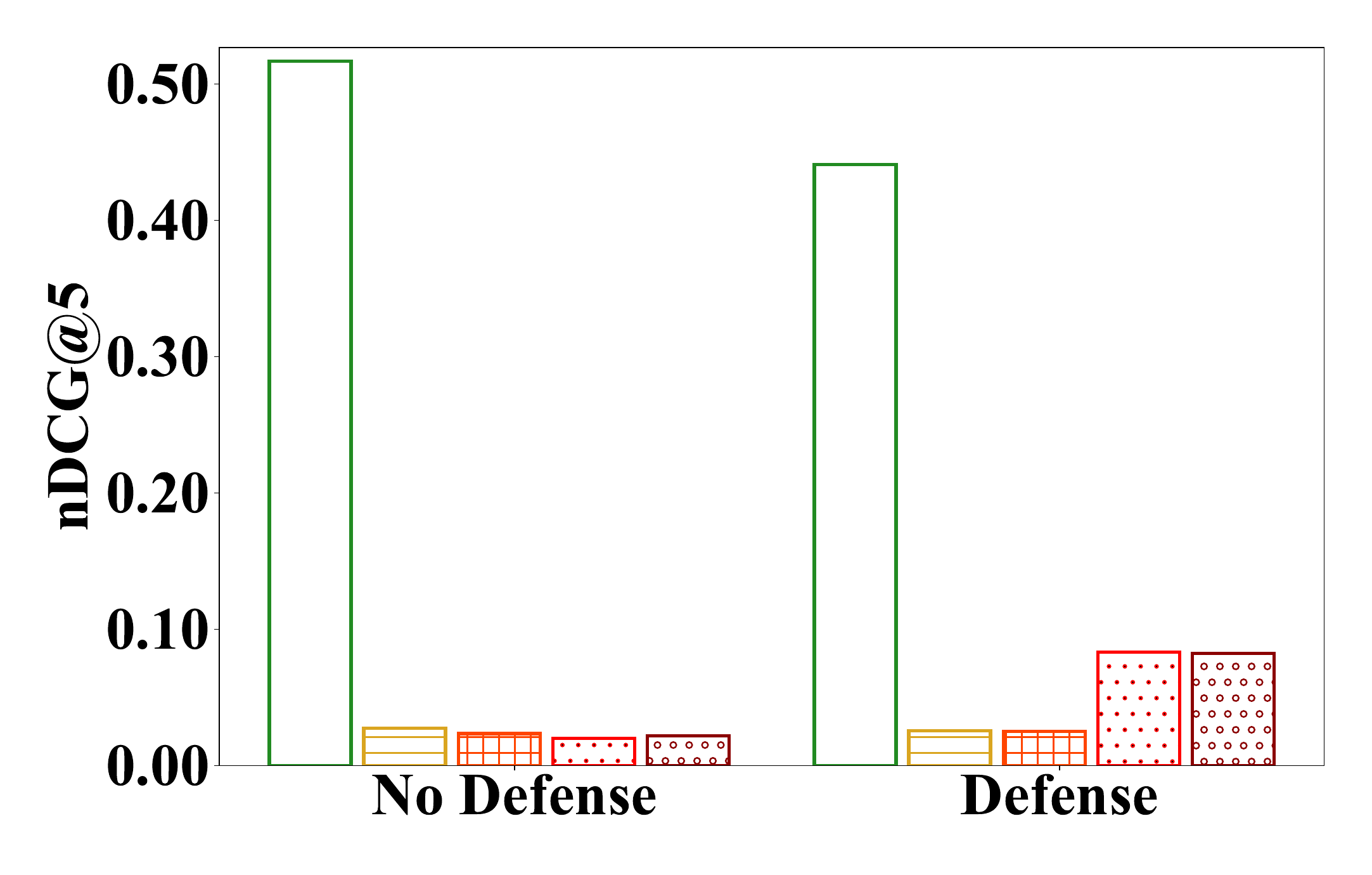}}
\endminipage\hfill
\minipage{0.24\textwidth}
\centering
 \subfigure[Steam]{\includegraphics[ width=\textwidth]{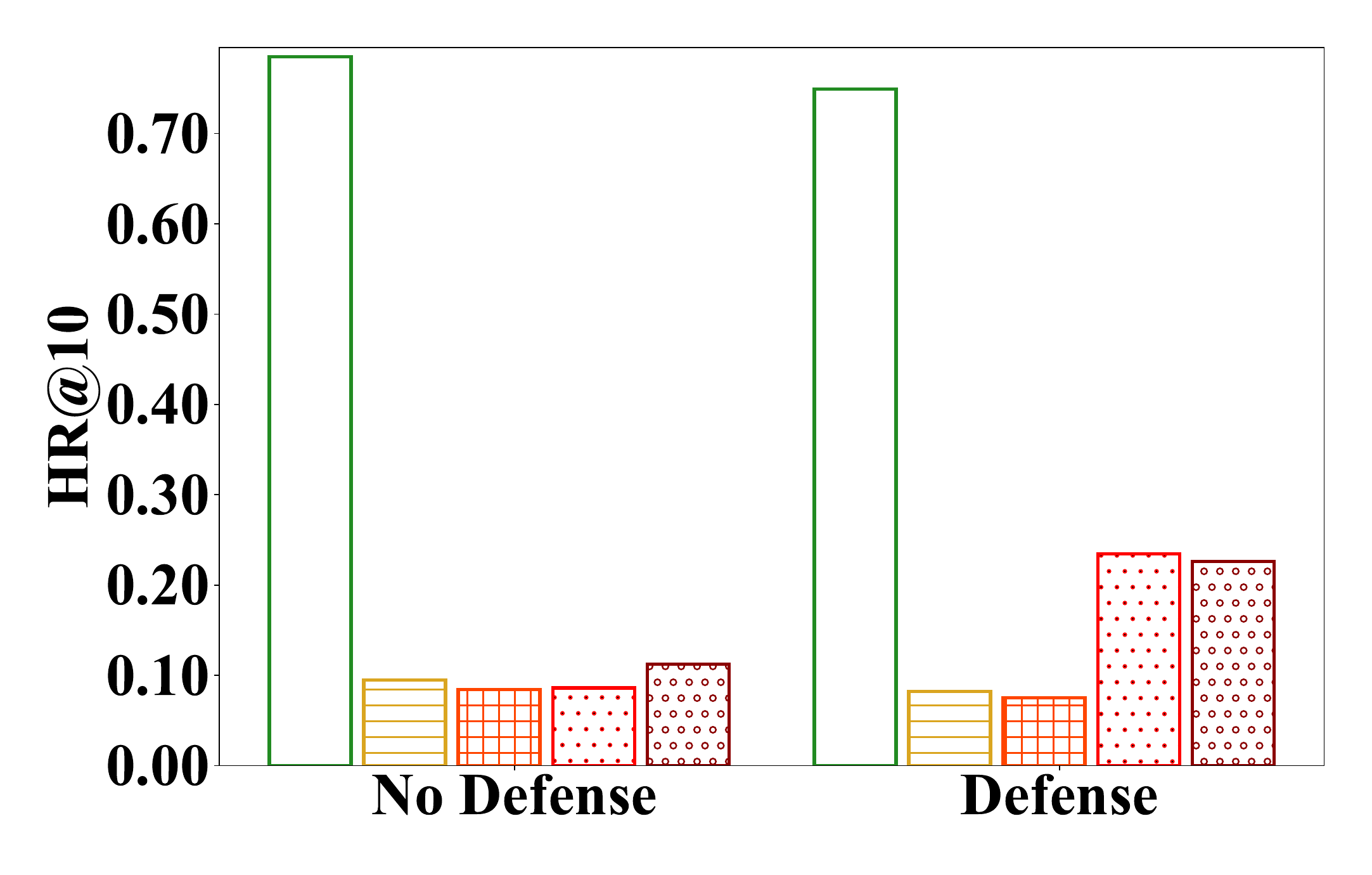}}
\endminipage\hfill
\minipage{0.24\textwidth}
\centering
 \subfigure[Steam]{ \includegraphics[ width=\textwidth]{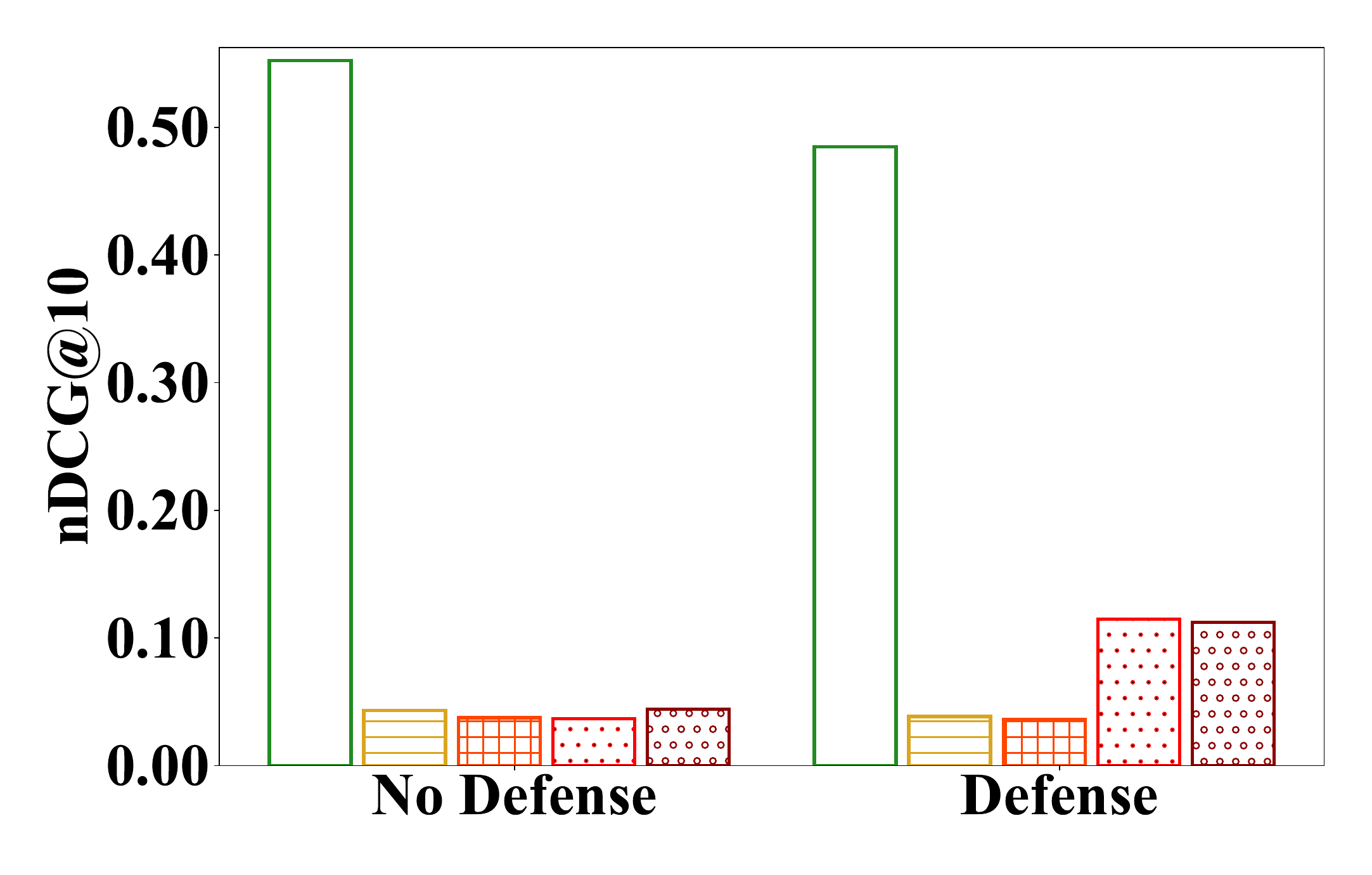}}
\endminipage\hfill

\minipage{0.24\textwidth}
\centering
 \subfigure[ML1M]{\includegraphics[ width=\textwidth]{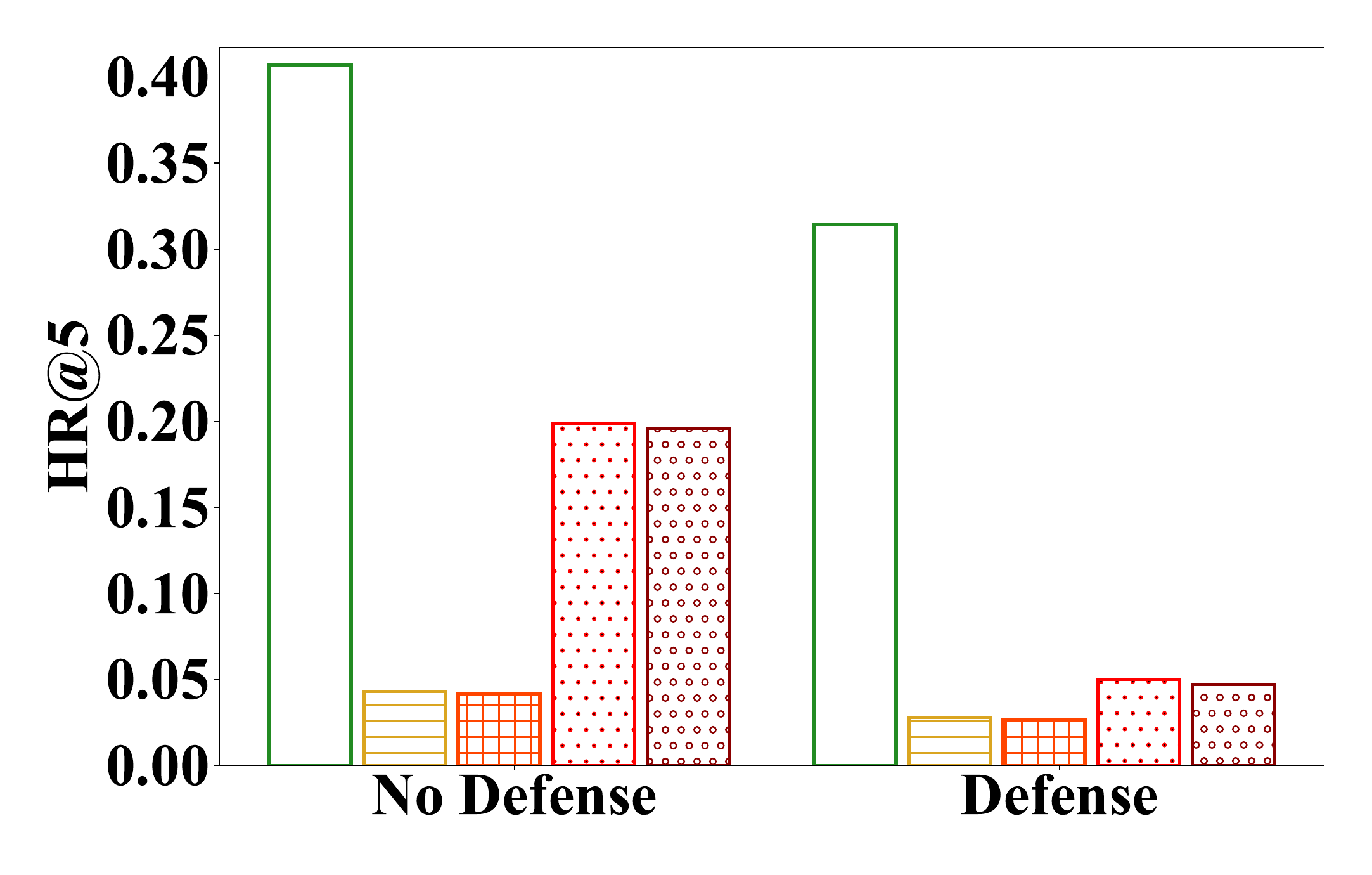}}
\endminipage\hfill
\minipage{0.24\textwidth}
\centering
 \subfigure[ML1M]{\includegraphics[ width=\textwidth]{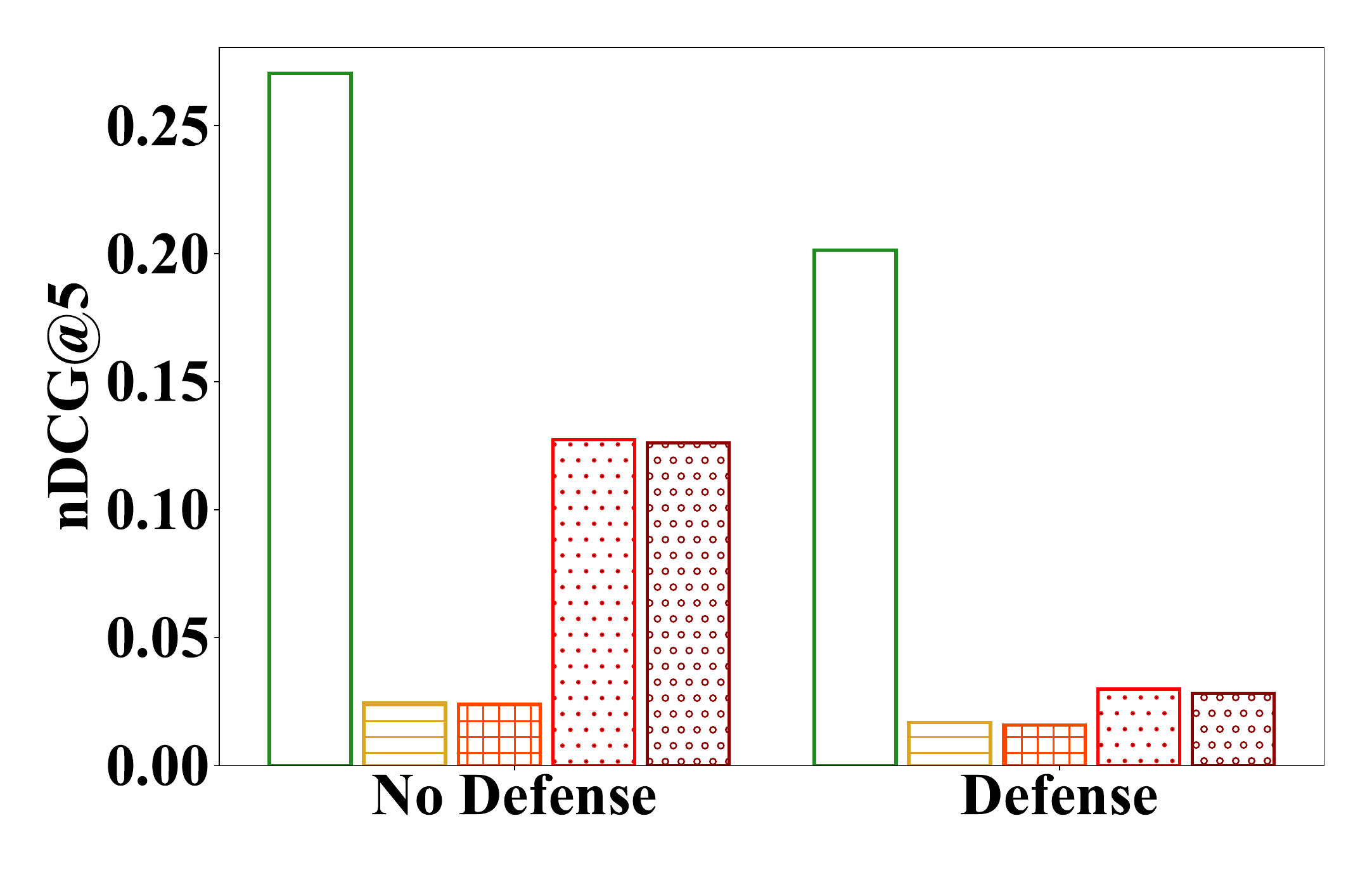}}
\endminipage\hfill
\minipage{0.24\textwidth}
\centering
 \subfigure[ML1M]{\includegraphics[ width=\textwidth]{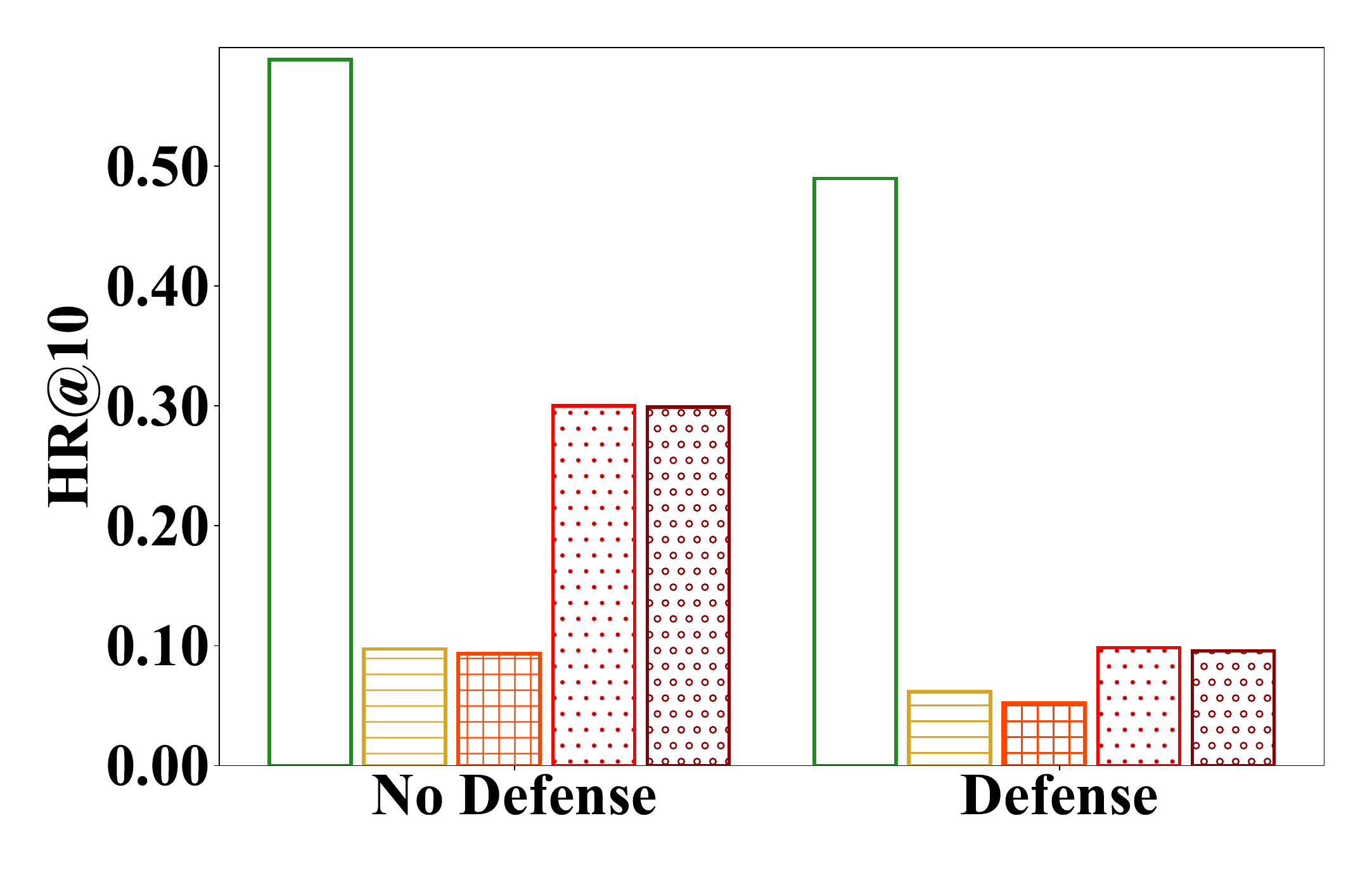}}
\endminipage\hfill
\minipage{0.24\textwidth}
\centering
 \subfigure[ML1M]{ \includegraphics[ width=\textwidth]{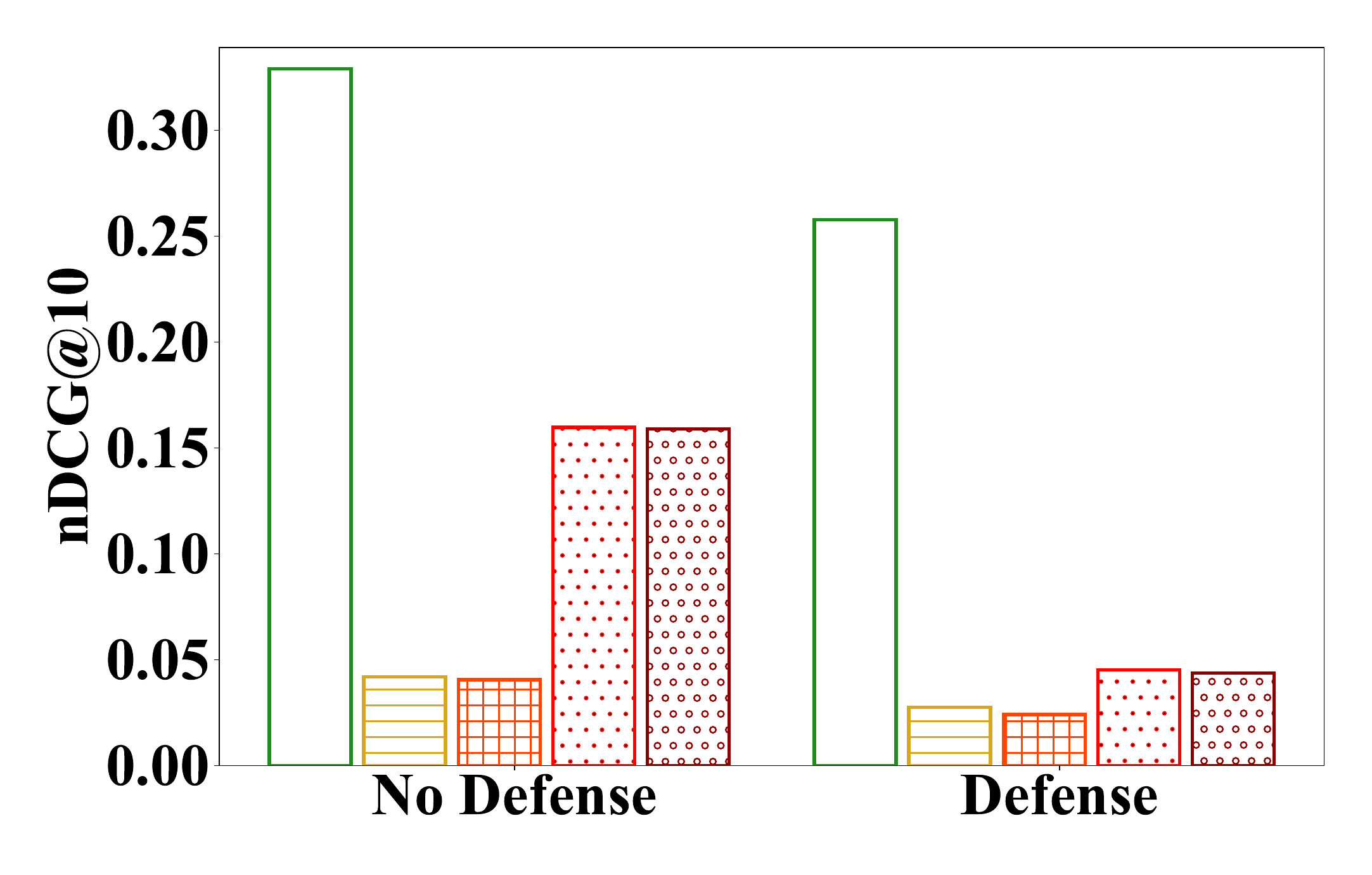}}
\endminipage\hfill
\caption{Results on FR Model with Adam Optimizer.}
\label{fig:full_adam}
\end{figure*}
\end{document}